\PassOptionsToPackage{unicode}{hyperref}
\PassOptionsToPackage{hyphens}{url}
\PassOptionsToPackage{dvipsnames,svgnames,x11names}{xcolor}
\documentclass[
  12pt]{article}

\usepackage{amsmath,amssymb,amsthm}
\usepackage{iftex}
\ifPDFTeX
  \usepackage[T1]{fontenc}
  \usepackage[utf8]{inputenc}
  \usepackage{textcomp} % provide euro and other symbols
\else % if luatex or xetex
  \usepackage{unicode-math}
  \defaultfontfeatures{Scale=MatchLowercase}
  \defaultfontfeatures[\rmfamily]{Ligatures=TeX,Scale=1}
\fi
\usepackage{lmodern}
\ifPDFTeX\else  
\fi
\IfFileExists{upquote.sty}{\usepackage{upquote}}{}
\IfFileExists{microtype.sty}{% use microtype if available
  \usepackage[]{microtype}
  \UseMicrotypeSet[protrusion]{basicmath} % disable protrusion for tt fonts
}{}
\makeatletter
\@ifundefined{KOMAClassName}{% if non-KOMA class
  \IfFileExists{parskip.sty}{%
    \usepackage{parskip}
  }{% else
    \setlength{\parindent}{0pt}
    \setlength{\parskip}{6pt plus 2pt minus 1pt}}
}{% if KOMA class
  \KOMAoptions{parskip=half}}
\makeatother
\usepackage{xcolor}
\makeatletter
\ifx\paragraph\undefined\else
  \let\oldparagraph\paragraph
  \renewcommand{\paragraph}{
    \@ifstar
      \xxxParagraphStar
      \xxxParagraphNoStar
  }
  \newcommand{\xxxParagraphStar}[1]{\oldparagraph*{#1}\mbox{}}
  \newcommand{\xxxParagraphNoStar}[1]{\oldparagraph{#1}\mbox{}}
\fi
\ifx\subparagraph\undefined\else
  \let\oldsubparagraph\subparagraph
  \renewcommand{\subparagraph}{
    \@ifstar
      \xxxSubParagraphStar
      \xxxSubParagraphNoStar
  }
  \newcommand{\xxxSubParagraphStar}[1]{\oldsubparagraph*{#1}\mbox{}}
  \newcommand{\xxxSubParagraphNoStar}[1]{\oldsubparagraph{#1}\mbox{}}
\fi
\makeatother

\usepackage{longtable,booktabs,array}
\usepackage{pdflscape}
\usepackage{calc} % for calculating minipage widths
\usepackage{etoolbox}
\makeatletter
\patchcmd\longtable{\par}{\if@noskipsec\mbox{}\fi\par}{}{}
\makeatother
\IfFileExists{footnotehyper.sty}{\usepackage{footnotehyper}}{\usepackage{footnote}}
\makesavenoteenv{longtable}
\usepackage{graphicx}
\makeatletter
\def\maxwidth{\ifdim\Gin@nat@width>\linewidth\linewidth\else\Gin@nat@width\fi}
\def\maxheight{\ifdim\Gin@nat@height>\textheight\textheight\else\Gin@nat@height\fi}
\makeatother
\setkeys{Gin}{width=\maxwidth,height=\maxheight,keepaspectratio}
\makeatletter
\def\fps@figure{htbp}
\makeatother

\makeatletter
\@ifpackageloaded{caption}{}{\usepackage{caption}}
\AtBeginDocument{%
\ifdefined\contentsname
  \renewcommand*\contentsname{Table of contents}
\else
  \newcommand\contentsname{Table of contents}
\fi
\ifdefined\listfigurename
  \renewcommand*\listfigurename{List of Figures}
\else
  \newcommand\listfigurename{List of Figures}
\fi
\ifdefined\listtablename
  \renewcommand*\listtablename{List of Tables}
\else
  \newcommand\listtablename{List of Tables}
\fi
\ifdefined\figurename
  \renewcommand*\figurename{Figure}
\else
  \newcommand\figurename{Figure}
\fi
\ifdefined\tablename
  \renewcommand*\tablename{Table}
\else
  \newcommand\tablename{Table}
\fi
}
\@ifpackageloaded{float}{}{\usepackage{float}}
\floatstyle{ruled}
\@ifundefined{c@chapter}{\newfloat{codelisting}{h}{lop}}{\newfloat{codelisting}{h}{lop}[chapter]}
\floatname{codelisting}{Listing}

\makeatother
\makeatletter
\@ifpackageloaded{caption}{}{\usepackage{caption}}
\@ifpackageloaded{subcaption}{}{\usepackage{subcaption}}
\makeatother

\ifLuaTeX
  \usepackage{selnolig}  % disable illegal ligatures
\fi

\usepackage{multirow}
\usepackage{mathtools}
\usepackage{bm}

\theoremstyle{plain}

\newtheorem{theorem}{Theorem}
\newtheorem{lemma}[theorem]{Lemma}
\newtheorem{corollary}{Corollary}

\theoremstyle{definition}
\newtheorem{assumption}{Assumption}

\theoremstyle{remark}

\newtheorem{remark}{Remark}

\usepackage[authoryear,round]{natbib}
\usepackage{bookmark}
\IfFileExists{xurl.sty}{\usepackage{xurl}}{} % add URL line breaks if available
\hypersetup{
  pdftitle={High-Dimensional Two-Sample Covariance Testing with Null-Preserving Transformations},
  pdfauthor={},
  pdfkeywords={high-dimensional inference; Gaussian approximation; Gaussian bootstrap; studentization; transformed covariance statistics},
  colorlinks=true,
  linkcolor={blue},
  filecolor={Maroon},
  citecolor={Blue},
  urlcolor={Blue},
  pdfcreator={LaTeX via pandoc}}

\newcommand{\anon}{1}
\newcommand{\bR}{\mbox{\bf R}}

\newcommand{\bx}{\mbox{\bf x}}
\newcommand{\bz}{\mbox{\bf z}}
\newcommand{\bb}{\mbox{\bf b}}
\newcommand{\bh}{\mbox{\bf h}}

\newcommand{\ba}{\mbox{\bf a}}

\newcommand{\bu}{\mbox{\bf u}}
\newcommand{\bv}{\mbox{\bf v}}
\newcommand{\bw}{\mbox{\bf w}}

\newcommand{\bL}{\mbox{\bf L}}
\newcommand{\bB}{\mbox{\bf B}}

\newcommand{\bn}{\mbox{\bf n}}
\newcommand{\bM}{\mbox{\bf M}}

\newcommand{\bQ}{\mbox{\bf Q}}
\newcommand{\bI}{\mbox{\bf I}}
\newcommand{\bG}{\mbox{\bf G}}
\newcommand{\bg}{\mbox{\bf g}}
\newcommand{\br}{\mbox{\bf r}}

\newcommand{\bH}{\mbox{\bf H}}
\newcommand{\bA}{\mbox{\bf A}}
\newcommand{\vech}{\operatorname{vech}}

\newcommand{\bZ}{\mbox{\bf Z}}

\newcommand{\E}{\mathbb{E}}
\newcommand{\Var}{\operatorname{Var}}
\newcommand{\Cov}{\operatorname{Cov}}

\newcommand{\bSigma}{\bm{\Sigma}}
\newcommand{\bD}{\bm{D}}
\newcommand{\by}{\bm{y}}

\newcommand{\hbw}{\hat{\bw}}

\newcommand{\hbSigma}{\hat{\bSigma}}

\newcommand{\bGamma}{\mbox{\boldmath{$\Gamma$}}}
\newcommand{\bOmega}{\mbox{\boldmath{$\Omega$}}}
\newcommand{\hbOmega}{\hat{\bOmega}}

\newcommand{\bxi}{\mbox{\boldmath{$\xi$}}}

\newcommand{\bmu}{\mbox{\boldmath{$\mu$}}}

\newcommand{\tbD}{\tilde{\bD}}
\newcommand{\tbw}{\tilde{\bw}}
\newcommand{\barbu}{\bar{\bu}}
\newcommand{\barbz}{\bar{\bz}}

\newcommand{\convP}{\stackrel{P}{\longrightarrow}}

\newcommand{\iidsim}{\stackrel{\mathrm{i.i.d.}}{\sim}}

\newcommand{\tbSigma}{\tilde{\bSigma}}

\begin{document}

\def\spacingset#1{\renewcommand{\baselinestretch}%
{#1}\small\normalsize} \spacingset{1}

%%%%%%%%%%%%%%%%%%%%%%%%%%%%%%%%%%%%%%%%%%%%%%%%%%%%%%%%%%%%%%%%%%%%%%%%%%%%%%

\if1\anon
{
  \title{\bf High-Dimensional Two-Sample Covariance Testing with Null-Preserving Transformations}
  \author{Haozhen Shu\thanks{Email: haozhen.shu@nie.edu.sg.}\hspace{.2cm}\\
National Institute of Education, Nanyang Technological University\\
and \\
Tianming Zhu\thanks{Email: tianming.zhu@nie.edu.sg.}\\
National Institute of Education, Nanyang Technological University}
  \maketitle
} \fi

\if0\anon
{
  \bigskip
  \bigskip
  \bigskip
  \begin{center}
    {\LARGE\bf High-Dimensional Two-Sample Covariance Testing with Null-Preserving Transformations}
\end{center}
  \medskip
} \fi

\bigskip
% \begin{abstract}
% We propose a simple studentized \(\ell_2\)-type statistic for high-dimensional
% two-sample covariance testing.  The statistic vectorizes the difference between
% the two sample covariance matrices, standardizes each covariance coordinate by
% its marginal scale, and aggregates the squared standardized discrepancies.
% To improve sensitivity beyond the original covariance-coordinate representation,
% we further introduce a transformation-enhanced version of the statistic, where
% prespecified nonsingular transformations are applied before aggregation.  This
% provides a simple power-enhancement mechanism: suitable transformations can
% redistribute and densify covariance discrepancies, thereby increasing the
% standardized quadratic signal while retaining the same implementation principle.

% The theoretical analysis is based on a Gaussian approximation over centered
% Euclidean balls, which matches the geometry of the proposed statistic and leads
% to a direct Gaussian multiplier-bootstrap implementation.  We establish validity
% of the feasible bootstrap test for both the original and transformed statistics,
% and prove power consistency under growing transformed standardized quadratic
% signal strength.  The method is easy to implement, avoids estimating the full
% covariance matrix of covariance entries, and allows power-enhancing
% transformations to be incorporated without changing the basic bootstrap
% calibration procedure.
% \end{abstract}
\begin{abstract}
Testing equality of two high-dimensional covariance matrices is challenging when many entries differ only slightly. Dependence among sample covariance entries can also affect the finite-sample size and power of tests that aggregate their differences. We propose a studentized $\ell_2$-type statistic  that averages squared, marginally standardized differences between corresponding entries of the two sample covariance matrices. Before constructing the statistic, we apply the same nonsingular linear transformation to both samples. In the transformed coordinates, the covariance-equality hypothesis is
unchanged, whereas the standardized differences and
correlations among their estimators generally change.  The transformation can therefore incorporate structural or scientific information  without reducing dimension.  We approximate the null distribution using a Gaussian multiplier
bootstrap that uses coordinatewise studentization and avoids forming
or inverting the full covariance matrix of the vectorized sample
covariance entries. For deterministic transformations, we derive nonasymptotic Gaussian and bootstrap approximation bounds and establish asymptotic size validity and power consistency. We also show that replacing a population
transformation by a same-sample estimator leaves the statistic and
bootstrap critical value asymptotically unchanged in relative terms
under an operator-norm convergence condition. Simulations show that suitable transformations improve size accuracy and power under weak dependence. An analysis of breast cancer gene-expression data illustrates the method.
\end{abstract}

\noindent%
\vfill

\newpage
\spacingset{1.8} % DON'T change the spacing!

\section{Introduction}\label{sec:intro}

Testing whether two populations have the same covariance matrix is a fundamental problem in statistical inference. It arises in gene-expression studies, where patterns of co-expression  may differ across biological conditions, and in neuroimaging, where covariance and correlation describe functional connectivity across brain regions. In these applications, the dependence structure is itself of scientific interest.

The present work is partly motivated by The Cancer Genome Atlas Breast Invasive Carcinoma (TCGA-BRCA) project. We compare Basal-like and
Luminal~A breast tumors, two subtypes with distinct molecular and
transcriptional profiles \citep{tcga_brca_molecular_portraits_2012}, using a
304-gene cell-cycle panel formed from the Hallmark E2F Targets and G2/M Checkpoint gene sets \citep{liberzon_molecular_2015}. Standard differential-expression
analyses compare the mean expression of individual genes, whereas covariance comparison examines whether gene-specific variability and patterns of co-expression differ between the two subtypes. The two
groups contain 141 and 421 tumors, respectively. Because the number of genes exceeds the sample size of
the smaller group, this covariance comparison is intrinsically
high-dimensional.

More generally, covariance testing becomes difficult when the dimension is comparable to or exceeds the sample sizes, because sample covariance matrices may be poorly conditioned or singular. Classical covariance tests such as those of \citet{box_general_1949} and \citet{nagao_some_1973} were developed under fixed-dimensional asymptotics and may not remain reliable in this regime. This difficulty has motivated high-dimensional procedures that compare covariance matrices without directly inverting their sample estimates \citep{ledoit_some_2002,chen_tests_2010}. Different procedures are
designed for different patterns of covariance differences. Entrywise
maximum tests are particularly sensitive when a small number of
covariance entries differ substantially. This line of work includes limiting results for sample correlations and coherence statistics \citep{cai_limiting_2011}, as well as maximum-type tests for covariance comparison and support recovery \citep{cai_two-sample_2013,chang_comparing_2017}. Their calibration is closely related to high-dimensional Gaussian approximation and multiplier bootstrap theory for maxima and hyperrectangles \citep{chernozhukov_gaussian_2013,chernozhukov_central_2017,chernozhukov_high-dimensional_2023}. 

In contrast, \(\ell_2\)-type procedures aggregate discrepancies across covariance entries and are natural when many entries differ. Representative quadratic-form tests include those of \citet{schott_test_2007} and \citet{srivastava_testing_2010}, while \citet{li_two_2012} proposed an unbiased U-statistic estimator of the squared Frobenius-norm difference. Other approaches focus on projections or spectral features of covariance differences. \citet{wu2020projected} proposed random-projection tests
for high-dimensional covariance matrices; \citet{ding2025two} developed an ultra-high-dimensional two-sample
test based on data splitting and selected sample eigenvalues; and
\citet{veitch2026rank} proposed a rank-adaptive procedure
based on Ky--Fan norms for low-rank covariance differences. These
methods target different structural features of the alternative and
are complementary to the coordinatewise quadratic framework considered
here. Within the class of $\ell_2$-type tests, \citet{wang_two-sample_2024} developed a normal-reference test for high-dimensional covariance matrices. Our numerical results indicate that this test and the untransformed version of the statistic proposed below can be markedly conservative under weak dependence.  This
observation motivates examining whether a change of coordinates can
improve finite-sample performance.

Before constructing the statistic, we apply the same nonsingular linear
transformation \(\bA\) to
both samples. Since \(\bA\) is nonsingular,
\[
\bSigma_1=\bSigma_2
\quad\Longleftrightarrow\quad
\bA\bSigma_1\bA^\top=\bA\bSigma_2\bA^\top.
\]
Unlike a dimension-reducing projection, this transformation retains
all $p$ coordinates and leaves the hypothesis being tested unchanged. The studentized $\ell_2$-type statistic, however, is generally not invariant under such a transformation. The choice of $\bA$ changes
both the standardized
covariance differences entering the statistic and the correlations
among their estimators. The identity transformation $\bA=\bI_p$ provides
the untransformed baseline. Motivated by the conservative behavior
observed under weak dependence,  our
numerical studies also consider a practical rule that maps a
regularized pooled covariance matrix to an equicorrelation target.
This construction illustrates one way to choose \(\bA\); it is not
claimed to be uniformly optimal over all alternatives.

Using the transformed samples, we construct a studentized
\(\ell_2\)-type statistic by averaging the squared, marginally
standardized differences between corresponding upper-triangular
entries of the two sample covariance matrices.  The resulting statistic
is invariant to coordinatewise rescaling. After standardization, it is
the squared Euclidean norm of the covariance-difference vector, so its
acceptance regions are centered Euclidean balls. We approximate its
null distribution using a Gaussian multiplier bootstrap based on the
individual covariance-coordinate contributions. Conditional on the
data, the bootstrap vector is Gaussian and reflects the empirical
dependence among these coordinates. Its computation uses
coordinatewise studentization and avoids explicitly forming or
inverting the full covariance matrix of the vectorized sample
covariance entries.

This paper makes three main contributions. First, we introduce a common nonsingular linear transformation into a studentized $\ell_2$-type covariance test. This provides a way to change the coordinates used for
aggregation while retaining all variables and preserving covariance
equality. Second, for a deterministic transformation, we derive nonasymptotic Gaussian and bootstrap approximation bounds, establish asymptotic size validity, and provide sufficient conditions for power consistency. The power analysis also shows how a transformation affects
the test through both the standardized covariance differences and the
correlations among their estimators.  Third, we study transformations
defined through unknown population quantities. We show that a
population transformation may be replaced by an estimator computed
from the same observations, without sample splitting, provided that
its relative operator-norm error is sufficiently small. Under this condition, both the test
statistic and the bootstrap critical value are asymptotically
equivalent in relative terms. This result concerns plug-in estimation of a population
transformation, rather than adaptive selection of $\bA$ from a
collection of candidate transformations. Simulations show that the equicorrelation-target transformations
considered in this paper can improve size accuracy and power,
particularly under weak dependence. The results remain stable across
several equicorrelation targets, suggesting limited sensitivity to the
precise target correlation in the settings examined. An analysis of breast cancer gene-expression data further
illustrates the proposed procedure.

The remainder of the paper is organized as follows.
Section~\ref{sec:main_results} presents the proposed method and its
theoretical properties. Section~\ref{sec:simulation} reports the
simulation studies, Section~\ref{sec:realdata} presents the breast
cancer application, and Section~\ref{sec:discussion} concludes the
paper. The Supplementary Material contains the technical proofs and
Python code implementing the proposed method and reproducing the
numerical results.

\section{Main Results}
\label{sec:main_results}

\subsection{Problem Setup}
\label{subsec:problem_setting}
Suppose we have two independent high-dimensional samples:
\begin{equation}\label{twosamp.equ}
    \by_{i1},\ldots,\by_{in_i}\ \text{are i.i.d. with}\quad
    \E(\by_{i1})=\bmu_i,\quad
    \operatorname{Cov}(\by_{i1})=\bSigma_i,\qquad i=1,2.
\end{equation}
The dimension \(p\) may be comparable to or larger than one or both sample
sizes. Write \(n=n_1+n_2\). We consider the following equal-covariance matrix testing problem:
\begin{equation*}
\label{eq:hypothesis}
H_0:\bSigma_1=\bSigma_2
\qquad\text{versus}\qquad
H_1:\bSigma_1\neq\bSigma_2.
\end{equation*}

Motivated by the conservative performance of \(\ell_2\)-type tests
under weak dependence,  we retain
the \(\ell_2\) aggregation but apply a common nonsingular transformation
\(\bA\in\mathbb R^{p\times p}\) to both samples:
\begin{equation}
    \bx_{ij} = \bA \by_{ij}, \; j = 1,\dots,n_i; \; i = 1,2.
    \label{eq:transformed-data}
\end{equation}
The transformed covariance matrices are $\bSigma_{1A} = \bA \bSigma_1 \bA^\top$ and $\bSigma_{2A} = \bA \bSigma_2 \bA^\top$. Since $\bA$ is nonsingular, the map $\bM \mapsto \bA \bM \bA^\top$ is a bijection on the
space of symmetric matrices. Consequently, the original testing problem is equivalent to
\begin{equation*}\label{eq:transformed-null}
    H_0 : \bSigma_{1A} = \bSigma_{2A} \qquad \text{versus} \qquad H_1 : \bSigma_{1A} \neq \bSigma_{2A}.
\end{equation*}

Unlike a dimension-reducing projection, the transformation  retains all
\(p\) coordinates and preserves the hypothesis being tested. It changes
the coordinates in which the covariance differences are standardized
and aggregated. The choice $\bA = \bI_p$ recovers the test in the original
coordinates, whereas other nonsingular choices change both
the standardized covariance differences and the dependence among
their sample estimators. Because these two effects depend on the alternative and
the underlying covariance structure, no transformation is expected to
be uniformly most powerful. Their numerical effects are examined in
Section~\ref{sec:simulation}.

Unless stated otherwise, the theoretical results below treat \(\bA\) as deterministic. Section~\ref{subsec:estimated-transformation} considers
a transformation defined through unknown population quantities and
studies its same-sample plug-in estimator \(\widehat{\bA}\).

\subsection{Test Statistic}
\label{subsec:statistic_gaussian}
We now define the test statistic using the transformed observations
in \eqref{eq:transformed-data}.  Let $\bar{\bx}_i
=n_i^{-1}\sum_{j=1}^{n_i}\bx_{ij}$ and $\widehat{\bSigma}_{iA}
=
(\hat{\sigma}_{i,ab})_{1\le a,b\le p}
=n_i^{-1}
\sum_{j=1}^{n_i}
(\bx_{ij}-\bar{\bx}_i)
(\bx_{ij}-\bar{\bx}_i)^\top, i=1,2$ denote the sample mean vector and sample covariance matrix, respectively, computed from the transformed samples in \eqref{eq:transformed-data}.  Let $d=p(p+1)/2$ be the number of upper-triangular covariance coordinates. Define the covariance-coordinate contribution
\begin{equation}
    \bh_{ij}
    =
    \operatorname{vech}
    \left\{
        (\bx_{ij}-\bar{\bx}_i)
        (\bx_{ij}-\bar{\bx}_i)^\top
        -
        \widehat{\bSigma}_{iA}
    \right\}
    \in\mathbb R^d,\; j=1,\ldots,n_i;\; i=1,2.
    \label{eq:covariance-contributions}
\end{equation}
Throughout this paper, \(\operatorname{vech}(\cdot)\) stacks the upper-triangular entries of a symmetric matrix, including the diagonal, in a fixed order, \(\|\cdot\|\) denotes the Euclidean norm for vectors, while \(\|\cdot\|_{\mathrm{F}}\) and \(\|\cdot\|_{\mathrm{op}}\) denote the Frobenius and operator norms for matrices, respectively. For the coordinate \(k\leftrightarrow(a,b)\), let \(h_{ij,ab}\) denote the corresponding entry of \(\bh_{ij}\), and define $s_{i,ab}^2
    =n_i^{-1}\sum_{j=1}^{n_i}h_{ij,ab}^2, i=1,2$. The coordinatewise studentized covariance difference is
\[
T_{a,b}
=
\frac{\hat{\sigma}_{1,ab}-\hat{\sigma}_{2,ab}}
{(n_1^{-1}s_{1,ab}^2+n_2^{-1}s_{2,ab}^2)^{1/2}}.
\]

We then define the global test statistic as  $T_n=d^{-1}\sum_{a=1}^{p}
\sum_{b=a}^{p}
T_{a,b}^{2}$. Large values of \(T_n\) provide evidence against \(H_0\). To express \(T_n\) in vector form, define
\begin{equation}
\label{hbw.equ}
\widehat{\bw}
=
\operatorname{vech}
\left(
\widehat{\bSigma}_{1A}
-
\widehat{\bSigma}_{2A}
\right)
\in\mathbb R^d.
\end{equation}
Let $\widehat{\bD}
=
\operatorname{diag}
(\widehat D_{11},\ldots,\widehat D_{dd})$, where $\widehat D_{kk}
=s_{1,ab}^2/n_1+s_{2,ab}^2/n_2$. We interpret \(\widehat{\bD}^{-1/2}\) as the diagonal generalized
inverse when a diagonal entry is zero. The high-probability results
below ensure that all diagonal entries are positive on the relevant
events. Then we can write
\begin{equation}
\label{eq:proposed-statistic}
T_n
=
\frac{1}{d}
\widehat{\bw}^{\top}
\widehat{\bD}^{-1}
\widehat{\bw}
=
\frac{1}{d}\|
\widehat{\bD}^{-1/2}
\widehat{\bw}\|^2.
\end{equation}
Thus, $T_n$ is invariant to nonzero rescaling of  any transformed
variable, although it is not invariant under a general nonsingular
linear transformation. Its acceptance region satisfies $\{T_n\le t\}
=\{
\|
\widehat{\bD}^{-1/2}\widehat{\bw}\|
\le \sqrt{dt}
\}$,  which motivates Gaussian approximation over centered Euclidean balls.

To analyze the feasible statistic \(T_n\), we first introduce an oracle statistic \(\widetilde T_n\) based on population-centered observations and population covariance-coordinate variances. We establish a Gaussian approximation for \(\widetilde T_n\) and then control the difference \(T_n-\widetilde T_n\) to transfer the approximation to \(T_n\).

Let $\bz_{ij}=\bx_{ij}-\E(\bx_{ij})=(z_{ij1},\ldots,z_{ijp})^\top,j=1,\ldots,n_i;\;i=1,2$ denote the population-centered observations, and define $ \widetilde\bSigma_{iA}
=n_i^{-1}\sum_{j=1}^{n_i}\bz_{ij}\bz_{ij}^\top, i=1,2$. Thus, \(\widetilde{\bSigma}_{iA}\) is the empirical second-moment
matrix based on population-centered observations. For \(j=1,\ldots,n_i; i=1,2\), define $\bu_{ij}=(u_{ij1},\ldots,u_{ijd})^\top=\vech(\bz_{ij}\bz_{ij}^\top)$, and let $\barbu_i=n_i^{-1}\sum_{j=1}^{n_i}\bu_{ij}$. Then $\E(\bu_{i1})=\vech(\bSigma_{iA})$ and \(\bar{\bu}_i=\operatorname{vech}(\widetilde{\bSigma}_{iA})\).  The empirical covariance difference based on population-centered
observations is
\begin{equation}\label{tbw.equ}
    \widetilde{\bw}
=
\operatorname{vech}
(\widetilde{\bSigma}_{1A}
-
\widetilde{\bSigma}_{2A})
=\bar{\bu}_1-\bar{\bu}_2.
\end{equation}
Let
\(\bu_{ij}^*
=\bu_{ij}-\operatorname{vech}(\bSigma_{iA})\) and
\(\bar{\bu}_i^*
=n_i^{-1}\sum_{j=1}^{n_i}\bu_{ij}^*\). Since the two samples are independent, $\bOmega_n=\operatorname{Cov}(\widetilde{\bw})
=\bOmega_1/n_1+\bOmega_2/n_2$, where $\bOmega_i
=\operatorname{Cov}(\bu_{i1})=\mathbb E(\bu_{i1}^*\bu_{i1}^{*\top}), i=1,2$. Finally, define
\(\widetilde{\bD}=\operatorname{diag}(\bOmega_n)\), the diagonal
matrix containing the coordinatewise variances of
\(\widetilde{\bw}\). For \(k\leftrightarrow(a,b)\), the \(k\)-th diagonal entry of $\widetilde{\bD}$ satisfies $\widetilde D_{kk}
=\operatorname{Var}(u_{11k})/n_1+\operatorname{Var}(u_{21k})/n_2$, where
\(\operatorname{Var}(u_{i1k})
=\operatorname{Var}(z_{i1a}z_{i1b})\).
Under Assumption~\ref{ass:coordinate_nondegeneracy} below, we have 
\begin{equation}
\label{tDkk.inequ}
\widetilde D_{kk}
\ge
c_H (n_1^{-1}+n_2^{-1}) \geq c_H/n_{\min},
\end{equation}
where $n_{\min}=\min\{n_1,n_2\}$. Hence \(\widetilde{\bD}\) is invertible, and the oracle statistic is well defined as
\begin{equation}
\label{eq:oracle-statistic}
\widetilde T_n
=
\frac1d
\widetilde{\bw}^{\top}
\widetilde{\bD}^{-1}
\widetilde{\bw}
=
\frac1d\|
\widetilde{\bD}^{-1/2}
\widetilde{\bw}\|^2.
\end{equation}
Under \(H_0\),
\(\operatorname{vech}(\bSigma_{1A})
=\operatorname{vech}(\bSigma_{2A})\), and hence $\widetilde{\bw}
=
\bar{\bu}_1^*-\bar{\bu}_2^*$. Therefore, the oracle statistic in (\ref{eq:oracle-statistic}) can equivalently be written as
\begin{equation}
\label{Tn0.equ}
\widetilde T_{n,0}=
\frac{1}{d}
(\bar{\bu}_1^*-\bar{\bu}_2^*)^\top
\widetilde{\bD}^{-1}
(\bar{\bu}_1^*-\bar{\bu}_2^*).
\end{equation}
Under \(H_0\), the vector \(\widetilde{\bD}^{-1/2}\widetilde{\bw}\) has mean zero and covariance
matrix
\begin{equation}
\label{R.equ}
\bR
=
\widetilde{\bD}^{-1/2}
\bOmega_n
\widetilde{\bD}^{-1/2}.
\end{equation}
This leads to the Gaussian reference statistic
\begin{equation}
\label{eq:gaussian-reference}
T_Z=\frac{1}{d}\|\bZ\|^2,
\qquad
\bZ\sim N_d(\mathbf{0},\bR).
\end{equation}

Since
\(\widetilde{\bD}=\operatorname{diag}(\bOmega_n)\),
the matrix \(\bR\) is positive semidefinite with unit diagonal.
If the vectors \(\bu_{ij}^*\) were Gaussian, \(\widetilde T_{n,0}\) and
\(T_Z\) would have the same distribution for every \(n_1,n_2\). For
general observations, \(T_Z\) instead serves as the Gaussian reference for
\(\widetilde T_{n,0}\). Let
\(\lambda_1\ge\cdots\ge\lambda_d\ge0\) be the eigenvalues of \(\bR\).
For \(s=1,2\), define
\begin{equation}
\label{eq:reference-eigen-geometry}
\Lambda_s
=
\left(\sum_{\ell=s}^d\lambda_\ell^2\right)^{1/2},
\quad \text{ and }\quad
\kappa(\bR)=(\Lambda_1\Lambda_2)^{-1/2}
\quad\text{whenever }\Lambda_2>0.
\end{equation}
Here $\Lambda_1=\|\bR\|_{\mathrm{F}}$ is the overall Frobenius scale of $\bR$, whereas $\Lambda_2$ is the residual Frobenius scale after excluding the largest eigenvalue. The factor $\kappa(\bR)$ depends jointly on these two scales and enters the Gaussian and bootstrap approximation bounds below. It becomes large when their product is small, as when a fixed-dimensional correlation matrix approaches rank one. We impose the following regularity conditions in the transformed coordinate system.

\begin{assumption}
\label{ass:coordinate_high_moment}
There exists a constant \(K_8<\infty\), independent of \(n_1,n_2\) and \(p\),
such that $\max_{i=1,2}\max_{1\le a\le p}\mathbb E|z_{i1a}|^8
\le K_8$.
\end{assumption}

\begin{assumption}
\label{ass:coordinate_nondegeneracy}
There exists a constant \(c_H>0\), independent of \(n_1,n_2\) and \(p\), such
that for each \(i=1,2\) and every
\(1\le a\le b\le p\), $\operatorname{Var}(z_{i1a}z_{i1b})\ge c_H$.
\end{assumption}
\begin{assumption}
\label{ass:gaussian_reference_geometry}
The oracle Gaussian reference matrix \(\bR\) satisfies $\Lambda_2>0$.
\end{assumption}

Assumption~\ref{ass:coordinate_high_moment} imposes a uniform
eighth-moment bound on the transformed coordinates, while
Assumption~\ref{ass:coordinate_nondegeneracy} prevents any covariance
coordinate from becoming asymptotically degenerate. Assumption~\ref{ass:gaussian_reference_geometry} concerns their joint
correlation structure; the condition \(\Lambda_2>0\) is equivalent to
requiring \(\operatorname{rank}(\bR)\geq2\).
% All asymptotic statements below are along sequences with
% \(n_{\min}\to\infty\); \(n_1,n_2,p,\bSigma_1,\bSigma_2,\bA\), and hence
% \(d\) and \(\bR\), may vary with \(n\), while \(K_8\) and \(c_H\) remain
% uniform.

\begin{remark}
\label{rem:requirements-on-A}
Let \(\boldsymbol a_\ell^\top\)
denote the \(\ell\)-th row of \(\bA\), and suppose
\[
0<c_A
\le
\min_{1\le\ell\le p}\|\boldsymbol a_\ell\|
\le
\max_{1\le\ell\le p}\|\boldsymbol a_\ell\|
\le
C_A<\infty,
\]
where \(c_A\) and \(C_A\) do not depend on \(n_1,n_2\) or \(p\). 
These bounds make it straightforward to verify the transformed-coordinate
assumptions using conditions on the original observations. Because
\(T_n\) is invariant to nonzero rescaling of individual rows of
\(\bA\), the rows may be normalized so that
\(\|\boldsymbol a_\ell\|=1\), giving \(c_A=C_A=1\). No uniform bound
on \(\|\bA\|_{\mathrm{op}}\) is required.
%Beyond row scale, the relevant conditions depend on the row directions and the induced covariance-coordinate correlation geometry. 
\end{remark}

\begin{remark}
\label{rem:original-scale-assumptions}
Let \(\by_{i1}^*=\by_{i1}-\bmu_i,i=1,2\). Suppose there exist constants
\(K_Y<\infty\) and \(c_Y>0\), independent of \(n_1,n_2\) and \(p\),
such that
\[
\max_{i=1,2}
\sup_{\|\boldsymbol v\|=1}
\mathbb E|
\boldsymbol v^\top\by_{i1}^*|^8
\le K_Y,\;\text{ and }\;
\min_{i=1,2}
\inf_{\substack{\|\boldsymbol u\|=1,
                 \|\boldsymbol v\|=1}}
\operatorname{Var}\{
(\boldsymbol u^\top\by_{i1}^*)
(\boldsymbol v^\top\by_{i1}^*)\}
\ge c_Y.
\]
Then, for every \(\bA\) satisfying the rowwise bounds in
Remark~\ref{rem:requirements-on-A}, the transformed coordinates satisfy
Assumptions~\ref{ass:coordinate_high_moment} and
\ref{ass:coordinate_nondegeneracy}, with $K_8=C_A^8K_Y$ and $c_H=c_A^4c_Y$.
Under row normalization, these constants reduce to \(K_8=K_Y\) and
\(c_H=c_Y\). These uniform original-scale conditions are sufficient but not
necessary. For a fixed \(\bA\), it is enough that the corresponding
moment and nondegeneracy bounds hold along the normalized row
directions of \(\bA\) and all pairs of such directions.

For Assumption~\ref{ass:gaussian_reference_geometry}, define $\bGamma_i^{(Y)}
=
\operatorname{Cov}
\left[
\operatorname{vech}(\by_{i1}^*\by_{i1}^{*\top})
\right], i=1,2$. Since \(\bA\) is nonsingular, its induced linear transformation on
the covariance coordinates is invertible. Diagonal standardization
by \(\widetilde{\bD}^{-1/2}\) also preserves rank. Therefore,  $\operatorname{rank}(\bR)
=
\operatorname{rank}(\bGamma_1^{(Y)}/n_1+\bGamma_2^{(Y)}/n_2)$. Thus, the rank condition in
Assumption~\ref{ass:gaussian_reference_geometry} is invariant under
nonsingular transformations, although \(\Lambda_1\), \(\Lambda_2\),
and \(\kappa(\bR)\) generally depend on \(\bA\). 
\end{remark}

The first step is to approximate the oracle statistic
\(\widetilde T_n\) by \(T_Z\). The following theorem gives a
nonasymptotic bound for this approximation, uniformly over all
thresholds.
\begin{theorem}
\label{thm:oracle_gaussian}
Suppose \(H_0\) holds and
Assumptions~\ref{ass:coordinate_high_moment}--%
\ref{ass:gaussian_reference_geometry} are satisfied. There exists a
constant \(C_1>0\), depending only on \(K_8\) and \(c_H\), such that
\[
\sup_{r\ge0}
\left|
\mathbb P(\widetilde T_n\le r)
-
\mathbb P(T_Z\le r)
\right|
\le
C_1
\left\{
\kappa^{3/4}(\bR)\delta_1^{1/4}
+
\kappa^{2/3}(\bR)\delta_2^{1/3}
\right\},
\]
where
\begin{equation}
\label{delta1.equ}
\delta_1
=
\frac{\Lambda_1d^2}{n_{\min}}
+
\frac{\lambda_1^{3/2}d^{3/2}}{\sqrt{n_{\min}}},\quad \text{ and }\quad  \delta_2
=\frac{d^2}{n_{\min}}+
\frac{\lambda_1^{1/2}d^{3/2}}{\sqrt{n_{\min}}}.
\end{equation}

%where $\|\cdot\|_{\mathrm{HS}}$ and $\|\cdot\|_{\mathrm{op}}$  denote the  Hilbert--Schmidt norm and operator norm, respectively.
\end{theorem}
Theorem~\ref{thm:oracle_gaussian} concerns the oracle statistic, which
uses population centering and the unknown diagonal matrix
\(\widetilde{\bD}\). To transfer this approximation to the test
statistic \(T_n\), we must control the errors caused by sample centering
and diagonal variance estimation. The next theorem provides a
high-probability bound for \(\lvert T_n-\widetilde T_n\rvert\).
\begin{theorem}
\label{thm:sample_oracle_transfer}
Suppose that \(H_0\) holds and that
Assumptions~\ref{ass:coordinate_high_moment} and
\ref{ass:coordinate_nondegeneracy} are satisfied. There exists a constant \(C_2>0\), depending only on \(K_8\) and \(c_H\), such that,
for every \(0<\eta<1\) satisfying $C_2\sqrt{d/(n_{\min}\eta)}<1/2$, we have, with probability at least \(1-\eta\), $|T_n-\widetilde T_n|
\le
\Delta(\eta)$,
where
\begin{align}
\label{eq:sample-oracle-error}
\Delta(\eta)
=
C_2\left[
\sqrt{\frac{d}{n_{\min}\eta}}
\left(
\eta^{-1/4}
+
\frac{1}{\sqrt{n_{\min}\eta}}
\right)^{2}
+
\frac{1}{\sqrt{n_{\min}}\eta^{3/4}}
+
\frac{1}{n_{\min}\eta}
\right].
\end{align}
In particular, if \(\eta=n_{\min}^{-1/8}\) and $C_2\sqrt{dn_{\min}^{-7/8}}<1/2$, then, with probability at least \(1-n_{\min}^{-1/8}\),
\[
|T_n-\widetilde T_n|
\le
C_2
\left(
\frac{d}{n_{\min}^{3/4}}
\right)^{1/2}.
\]
%Consequently, if \(d=o(n_{\min}^{3/4})\), then $|T_n-\widetilde T_n|=o_{\mathbb P}(1)$.
\end{theorem}
On the event in
Theorem~\ref{thm:sample_oracle_transfer},
\(\lvert T_n-\widetilde T_n\rvert\le\Delta(\eta)\). Therefore,
probabilities involving \(T_n\) at a threshold \(r\) can be bounded by
probabilities involving \(\widetilde T_n\) at the shifted thresholds
\(r-\Delta(\eta)\) and \(r+\Delta(\eta)\). A bound on the probability
that the Gaussian reference statistic \(T_Z\) falls in such a short
interval, together with Theorem~\ref{thm:oracle_gaussian}, yields the
following approximation for the feasible statistic.

\begin{corollary}
\label{cor:feasible_gaussian}
Suppose that \(H_0\) holds and that
Assumptions~\ref{ass:coordinate_high_moment}--\ref{ass:gaussian_reference_geometry} are satisfied. There exists a
universal constant \(C_3>0\) such that, for every \(0<\eta<1\)
satisfying $C_2\sqrt{d/(n_{\min}\eta)}<1/2$, we have 
\[
\sup_{r\ge0}
\left|
\mathbb P(T_n\le r)
-
\mathbb P(T_Z\le r)
\right|
\le
\varepsilon(\eta),
\]
where $\varepsilon(\eta)
=
C_1\{
\kappa(\bR)^{3/4}\delta_1^{1/4}
+
\kappa(\bR)^{2/3}\delta_2^{1/3}\}
+
\eta
+
C_3\kappa(\bR)d\,\Delta(\eta)$. Here $C_1, \kappa(\bR),\delta_1$, and $\delta_2$ are as in
Theorem~\ref{thm:oracle_gaussian}, and $\Delta(\eta)$ is defined in
Theorem~\ref{thm:sample_oracle_transfer}.
\end{corollary}
Corollary~\ref{cor:feasible_gaussian} establishes the Gaussian
approximation for the feasible statistic. Its reference distribution
still depends on the unknown correlation matrix \(\bR\), however, and
therefore cannot be used directly to compute a critical value. The
next subsection replaces this unknown reference distribution by a
Gaussian multiplier bootstrap.

\subsection{Bootstrap Calibration}
\label{subsec:bootstrap_calibration}
% Corollary~\ref{cor:feasible_gaussian} approximates the null distribution of $T_n$ by the
% Gaussian reference distribution of $T_Z$. Because the correlation
% matrix $\bR$ is unknown, we approximate this reference distribution
% using a Gaussian multiplier bootstrap.

Let $e_{ij}$, $j=1,\ldots,n_i$; $i=1,2$, be independent
\(N(0,1)\) variables that are also independent of the data. Define 
\begin{equation}
\label{eq:bootstrap-vector}
\widehat{\bw}^{\mathrm{mb}}
=\frac1{n_1}\sum_{j=1}^{n_1}e_{1j}\, \bh_{1j}
-
\frac1{n_2}\sum_{j=1}^{n_2}e_{2j}\, \bh_{2j},
\end{equation}
and 
\begin{equation}
\label{eq:bootstrap-statistic}
T_n^{\mathrm{mb}}
=
\frac1d\|
\widehat \bD^{-1/2}\widehat \bw^{\mathrm{mb}}\|^2.
\end{equation}

Let \(\mathcal X\) denote the observed data. Conditional on
$\mathcal X$, the vectors $\bh_{ij}$ are fixed, and hence 
 $\widehat{\bw}^{\mathrm{mb}}
    \mid\mathcal X
    \sim
    N_d(\mathbf{0},\widehat{\bOmega}_n)$, where $\widehat\bOmega_n
= \hbOmega_1/n_1+\hbOmega_2/n_2$, with $\hbOmega_i=n_i^{-1}\sum_{j=1}^{n_i}\bh_{ij}\bh_{ij}^\top, i=1,2$. By the definition of \(s_{i,ab}^2\),
\(\operatorname{diag}(\widehat{\bOmega}_n)=\widehat{\bD}\).
% By the definition of \(s_{i,ab}^2\), the diagonal entries of
% \(\widehat{\bOmega}_n\) are $s_{1,ab}^2/n_1+s_{2,ab}^2/n_2, 1\leq a\leq b\leq p$, arranged in the same order as \(\operatorname{vech}(\cdot)\). Hence, $\operatorname{diag}(\widehat{\bOmega}_n)=\widehat{\bD}$. 
Therefore, on the event that every diagonal entry of
\(\widehat{\bD}\) is positive,
\begin{equation}
\label{Rhat.equ}
    \widehat{\bD}^{-1/2}
    \widehat{\bw}^{\mathrm{mb}}
    \mid\mathcal X
    \sim
    N_d(\mathbf{0},\widehat{\bR}),
    \qquad
    \widehat{\bR}
    =
    \widehat{\bD}^{-1/2}
    \widehat{\bOmega}_n
    \widehat{\bD}^{-1/2}.
\end{equation}
Thus, conditional
on the data, \(T_n^{\mathrm{mb}}\) has the same form as \(T_Z\), with the unknown \(\bR\) replaced by
\(\widehat{\bR}\).  The validity of this replacement is established
below.

Throughout this paper, $\alpha\in(0,1)$ is fixed and denotes the nominal
significance level. Let $c_{1-\alpha}^{\mathrm{mb}}
=
\inf
\left\{
t\ge0:
\mathbb P(
T_n^{\mathrm{mb}}\le t\mid\mathcal X)
\ge 1-\alpha
\right\}$ be the conditional \((1-\alpha)\)-quantile of
\(T_n^{\mathrm{mb}}\). 
The proposed test rejects \(H_0\) at nominal
level $\alpha$ when
$T_n>c_{1-\alpha}^{\mathrm{mb}}$. In practice, \(c_{1-\alpha}^{\mathrm{mb}}\) is approximated by Monte
Carlo simulation. For $b=1,\ldots,B$, generate independent copies of
the multipliers in \eqref{eq:bootstrap-vector}, compute
$T_{n,b}^{\mathrm{mb}}$ from \eqref{eq:bootstrap-statistic}, and take
their empirical $(1-\alpha)$-quantile. The conditional bootstrap tail
probability is estimated by the proportion of bootstrap statistics
that are at least as large as the observed statistic $T_n$. %When the exceedance count is zero, the empirical tables display \(p<1/B\) as a finite-bootstrap resolution convention.

Conditional on the data, \(T_n^{\mathrm{mb}}\) is the squared norm of
a Gaussian vector with correlation matrix \(\widehat{\bR}\), whereas
\(T_Z\) is the squared norm of a Gaussian vector with correlation
matrix \(\bR\). Bootstrap validity
therefore reduces to comparing the conditional bootstrap distribution with the Gaussian reference distribution.
The next theorem gives a nonasymptotic bound for this comparison.
\begin{theorem}
\label{thm:bootstrap_explicit}
Suppose that
Assumptions~\ref{ass:coordinate_high_moment}--\ref{ass:gaussian_reference_geometry} hold. There exists a constant
\(C_4>0\), depending only on \(K_8\) and \(c_H\), such that the
following holds. For every \(\eta_2\in(0,1)\) satisfying
\[
C_4
\sqrt{\frac{d^3}{n_{\min}\eta_2}}
\le
\frac{\Lambda_2}{2},
\]
with probability at least \(1-\eta_2\),
\[
\sup_{r\ge0}
\left|
\mathbb P
\left(
T_n^{\mathrm{mb}}\le r\mid\mathcal X
\right)
-
\mathbb P(T_Z\le r)
\right|
\le
C_4\kappa(\bR)
\sqrt{\frac{d^3}{n_{\min}\eta_2}}.
\]
\end{theorem}
Theorem~\ref{thm:bootstrap_explicit} controls the difference between
the conditional bootstrap distribution and the Gaussian reference
distribution. Corollary~\ref{cor:feasible_gaussian} controls the
difference between the null distribution of \(T_n\) and the same
Gaussian reference distribution. Combining these two approximations
gives the asymptotic validity of the bootstrap critical value.

\begin{corollary}
\label{cor:size_validity} Suppose that \(H_0\) holds and that
Assumptions~\ref{ass:coordinate_high_moment}--\ref{ass:gaussian_reference_geometry} are satisfied. Let
\(\eta_n\to0\) be a deterministic sequence such that the conditions of
Corollary~\ref{cor:feasible_gaussian}, with \(\eta=\eta_n\), and
Theorem~\ref{thm:bootstrap_explicit}, with \(\eta_2=\eta_n\), hold for
all sufficiently large \(n\).  If $\varepsilon(\eta_n)\to0$ and $
\kappa(\bR)
\sqrt{d^3/(n_{\min}\eta_{n})}
\to0$, then $\mathbb P_{H_0}(
T_n>c_{1-\alpha}^{\mathrm{mb}})
\to\alpha$.
\end{corollary}

\subsection{Power Analysis}
\label{subsec:power}

We now study the power of \(T_n\) under
\(H_1:\bSigma_{1A}\neq\bSigma_{2A}\).  Define
\begin{equation}
\label{eq:transformed-signal}
\boldsymbol{\Delta}_{\Sigma}
=
\operatorname{vech}(
\bSigma_{1A}-\bSigma_{2A}),
\qquad
\boldsymbol{\theta}_n
=
\widetilde{\bD}^{-1/2}
\boldsymbol{\Delta}_{\Sigma},
\quad\text{ and }\quad
h_n
=
\frac{1}{d}
\|\boldsymbol{\theta}_n\|^2.
\end{equation}
Thus, \(\boldsymbol{\theta}_n\) is the coordinatewise standardized covariance difference in the transformed representation, while \(h_n\) is its average squared magnitude. 

Under \(H_1\), the standardized oracle covariance-difference vector
\(\widetilde{\bD}^{-1/2}\widetilde{\bw}\) has mean
\(\boldsymbol{\theta}_n\) and covariance matrix \(\bR\). Thus,
\(h_n\) measures the contribution of
the population mean shift to the quadratic statistic. The power of the
proposed test also depends on the random variation around this mean
and on the error from estimating the diagonal studentizer. The next
theorem first gives a finite-sample lower
bound for the rejection probability and then provides two sufficient
conditions for power consistency.

% The next theorem gives a finite-sample lower bound for the power in terms of $h_n$ and $d/n_{\min}$, together with two
% sufficient conditions for power consistency.

\begin{theorem}
\label{thm:power_consistency}  Suppose \(H_1\) holds and that
Assumptions~\ref{ass:coordinate_high_moment} and
\ref{ass:coordinate_nondegeneracy} are satisfied. There exists a
constant \(C_5>0\), depending only on \(K_8\) and \(c_H\), such that,
whenever \(h_n>6/\alpha\),
\begin{equation}
\label{eq:finite-sample-power-bound}
\mathbb P_{H_1}
\left(
T_n>c_{1-\alpha}^{\mathrm{mb}}
\right)
\ge
1-
C_5
\left(
\frac{d}{n_{\min}}
+
\frac1{h_n}
\right).
\end{equation}
Moreover, the test is consistent, in the sense that $\mathbb P_{H_1}(
T_n>c_{1-\alpha}^{\mathrm{mb}})
\to1$ if either of the following conditions holds:
\begin{enumerate}
\item[\textup{(i)}]
\(d/n_{\min}\to0\) and \(h_n\to\infty\);

\item[\textup{(ii)}]
\(d^{3/2}/(\sqrt{n_{\min}}\Lambda_1)\to0\) and
\(dh_n/\Lambda_1
=\|\boldsymbol{\theta}_n\|^2/\|\bR\|_{\mathrm{F}}\to\infty\).
\end{enumerate}
\end{theorem}
Condition~\textup{(i)} follows directly from the finite-sample bound
\eqref{eq:finite-sample-power-bound}, whereas condition~\textup{(ii)}
uses the correlation geometry through
\(\Lambda_1=\|\bR\|_{\mathrm F}\). To interpret condition (i), suppose that there exist an index set
\(\mathcal I_n\subset\{1,\ldots,d\}\), with
\(|\mathcal I_n|=m_n\), and a deterministic sequence
\(\zeta_n>0\) such that
% \begin{equation}
% \label{eq:coordinate-alternative}
$\min_{k\in\mathcal I_n}
|\theta_{n,k}|
\ge
\zeta_n$.
%\end{equation}
Then $\|\boldsymbol{\theta}_n\|^2
\ge
\sum_{k\in\mathcal I_n}\theta_{n,k}^2
\ge
m_n\zeta_n^2$, and hence $h_n
\ge  m_n\zeta_n^2/d$. Therefore, condition~\textup{(i)}
holds if
\(d/n_{\min}\to0\) and \(m_n\zeta_n^2/d\to\infty\). When a constant fraction of the covariance coordinates satisfies the
lower bound, so that $m_n\asymp d$, the  signal condition
reduces to $\zeta_n^2\to\infty$. When only a vanishing fraction
satisfies the lower bound, so that $m_n=o(d)$, it requires $\zeta_n^2\gg d/m_n$. This shows the dilution, within this
sufficient condition, caused by averaging over all $d$ covariance
coordinates. Because $\zeta_n$ is defined on the standardized scale,
its magnitude already incorporates the sample sizes through
$\widetilde{\bD}^{-1/2}$. 
 
 The requirement \(h_n\to\infty\) in condition (i) is simple
to interpret but can be stronger than necessary, because
it requires \(\|\boldsymbol{\theta}_n\|^2\) to dominate \(d\). Condition~(ii) instead compares the
accumulated signal with
\(\Lambda_1=\|\bR\|_{\mathrm{F}}\). Since \(\bR\) is a correlation
matrix, \(\sqrt d\le\Lambda_1\le d\). Thus, the signal requirement
\(\|\boldsymbol{\theta}_n\|^2/\Lambda_1\to\infty\) can be substantially
weaker than \(h_n\to\infty\) when \(\Lambda_1=o(d)\), although condition (ii) also requires \(d^{3/2}/(\sqrt{n_{\min}}\Lambda_1)\to0\).

Condition~\textup{(ii)} also clarifies how the transformation affects
power.  For a given transformation $\bA$, write
$\boldsymbol{\theta}_{n,A}$, $\bR_A$,
$\Lambda_{1,A}=\|\bR_A\|_{\mathrm{F}}$, and
$h_{n,A}=d^{-1}\|\boldsymbol{\theta}_{n,A}\|^2$ for the corresponding
quantities induced by $\bA$. The second requirement in condition~(ii) directly motivates defining
\begin{equation}
\label{eq:transformation-signal-ratio}
    J_n(\bA)
    =
    \frac{\|\boldsymbol{\theta}_{n,A}\|^2}
         {\Lambda_{1,A}}
    =
    \frac{d h_{n,A}}{\Lambda_{1,A}}.
\end{equation}
The signal requirement in condition~(i) is $h_{n,A}\to\infty$,
whereas that in condition~(ii) is $J_n(\bA)\to\infty$. Thus, a transformation can affect the sufficient
conditions for consistency through both the standardized covariance
difference \(\boldsymbol{\theta}_{n,A}\) and the correlation matrix
\(\bR_A\). A larger value of \(J_n(\bA)\) makes the signal requirement
in condition (ii) easier to satisfy, but does not imply a uniform
finite-sample power ordering. We report \(J_n(\bA)\) in the simulation
study to help interpret the transformations considered there; it is
not used to select \(\bA\).

\subsection{Plug-in Estimation of the Transformation}
\label{subsec:estimated-transformation}
The preceding results treat $\bA$ as a deterministic 
transformation. In practice, however, its numerical value may be
unknown because a specified transformation rule may depend on unknown
population quantities. Let \(\bA\) denote the deterministic population
 transformation induced by such a rule, and let \(\widehat{\bA}\) be its sample analogue, constructed
from the same observations used to compute the test statistic. We now
study the effect of replacing \(\bA\) by \(\widehat{\bA}\) on both the
test statistic and its multiplier-bootstrap calibration. The matrix \(\widehat{\bA}\) is computed once from the observed
data and then held fixed during multiplier resampling.

Let \(\widehat T_n\) denote the statistic obtained from \(T_n\) by
replacing \(\bA\) with \(\widehat{\bA}\). Similarly, let
\(\widehat T_n^{\mathrm{mb}}\) and
\(\widehat c_{1-\alpha}^{\mathrm{mb}}\) denote the corresponding
multiplier-bootstrap statistic and critical value. %Quantities without an additional hat are computed using the deterministic reference transformation \(\bA\). 
Let \(\convP\) denote convergence in
probability. 

\begin{theorem}
\label{thm:estimated-transformation}
Suppose
Assumptions~\ref{ass:coordinate_high_moment} and
\ref{ass:coordinate_nondegeneracy} hold for the observations
transformed by \(\bA\). If \(d/n_{\min}\to0\) and
\(\sqrt d\|\widehat{\bA}\bA^{-1}-\bI_p\|_{\mathrm{op}}\convP0\), then
$\widehat T_n
=T_n\{1+o_{\mathbb P}(1)\}$ and $\widehat c_{1-\alpha}^{\mathrm{mb}}
=
c_{1-\alpha}^{\mathrm{mb}}
\{1+o_{\mathbb P}(1)\}$.
\end{theorem}
Theorem~\ref{thm:estimated-transformation} concerns the plug-in
implementation of a given transformation rule, rather than the
selection of that rule.  It allows \(\widehat{\bA}\) to be
computed from the same observations used by the test and therefore
does not require sample splitting.  The theorem does not
imply that one transformation is uniformly optimal over all
alternatives. Section~\ref{Simu1.sec} considers a specific regularized
pooled-covariance transformation and its sample implementation.

\section{Simulation Studies}\label{sec:simulation}
In this section, we evaluate the finite-sample performance of the
proposed test $T_n$. In Simulation~1, we  examine how the choice of the
transformation matrix \(\bA\) affects empirical size and power. In Simulation~2,
we compare $T_n$ with three existing procedures: the two-sample
covariance test of \citet{li_two_2012}, denoted by $T_{LC}$; a multiplier-bootstrap implementation of the entrywise
maximum-type test of \citet{chang_comparing_2017}, denoted by
\(T_{CZZW}\); and the normal-reference test of \citet{wang_two-sample_2024}, denoted by $T_{WZZ}$.

The two samples \eqref{twosamp.equ} are generated from $\by_{ij}=\bmu_i+\mathbf L_i\bv_{ij}, j=1,\ldots,n_i;\ i=1,2$, where $\mathbf L_i$ is the lower-triangular Cholesky factor satisfying $\mathbf L_i\mathbf L_i^\top=\bSigma_i$. The vector
\(\bv_{ij}=(v_{ij1},\ldots,v_{ijp})^\top\) has i.i.d. entries with mean zero and variance one. We
consider the following three distributions:
\begin{description}
    \item[Model~1 (standard Gaussian):] $v_{ijk}\iidsim N(0,1)$;
    \item[Model~2 (standardized centered $\chi^2_1$):] $v_{ijk}\iidsim(\chi^2_1-1)/\sqrt2$;
    \item[Model~3 (standardized Laplace): ]$v_{ijk}\stackrel{\mathrm{i.i.d.}}{\sim}\operatorname{Laplace}(0,1/\sqrt{2})$
\end{description}
These models represent, respectively, a Gaussian baseline, a skewed non-Gaussian distribution, and a symmetric
distribution with heavier tails than the Gaussian distribution. Since the tests considered here are
location invariant, we set
\(\bmu_1=\bmu_2=\mathbf{0}\) throughout.

Throughout this section, the nominal significance level is \(\alpha=0.05\), and each experiment is
based on \(1{,}000\) Monte Carlo replications. The empirical rejection
probability is defined as the proportion of replications in which the
null hypothesis is rejected; it is referred to as the empirical size
under \(H_0\) and as the empirical power under the alternatives. All
bootstrap-based procedures use \(B=500\) bootstrap replications.
\subsection{Simulation 1}\label{Simu1.sec}
We first examine the finite-sample performance of \(T_n\) under
different choices of \(\bA\) and different covariance structures. Two
classes of covariance matrices are considered: (1) Equicorrelation covariance: $\bSigma_i=(1-\rho_i)\bI_p+\rho_i\mathbf{1}_p\mathbf{1}_p^\top$;
and (2) Toeplitz covariance: $\bSigma_i=(\rho_i^{|a-b|})_{1\le a,b\le p}$. We consider three sample-size configurations \(\bn_1=(60,60)\),
\(\bn_2=(80,100)\), and \(\bn_3=(120,150)\), and set \(p\in\{100,200\}\).

Motivated by the conservative behavior observed under weak dependence,
we consider one practical construction that maps a regularized pooled
covariance matrix to a prespecified equicorrelation target. Define $\bSigma_{\mathrm{tar}}(\rho)=
(1-\rho)\bI_p+\rho\mathbf{1}_p\mathbf{1}_p^\top$, where \(\mathbf 1_p\) denotes the \(p\)-dimensional vector of ones, and $\rho$ is the common off-diagonal correlation of the target
matrix. We consider three
prespecified target correlations, \(\rho_{\mathrm{tar}}\in\{0.3,0.5,0.7\}\). To describe the population version of this construction, let
$\bSigma_{\mathrm{pool}}=(n_1\bSigma_1+n_2\bSigma_2)/n$ and set
$\gamma=10^{-6}\operatorname{tr}(\bSigma_{\mathrm{pool}})/p$. The population ridge is included so that the population and sample
transformations are defined by the same regularized rule. For a given
$\rho_{\mathrm{tar}}$, define $\bA(\rho_{\mathrm{tar}})
=
\bSigma_{\mathrm{tar}}(\rho_{\mathrm{tar}})^{1/2}
\left(
\bSigma_{\mathrm{pool}}+\gamma\bI_p
\right)^{-1/2}$. Thus,
$\bA(\rho_{\mathrm{tar}})
(\bSigma_{\mathrm{pool}}+\gamma\bI_p)
\bA(\rho_{\mathrm{tar}})^\top
=
\bSigma_{\mathrm{tar}}(\rho_{\mathrm{tar}})$.

For its sample implementation, let \(\widehat{\bSigma}_i=n_i^{-1}\sum_{j=1}^{n_i}
(\by_{ij}-\bar\by_i)(\by_{ij}-\bar\by_i)^\top\) be the sample covariance
matrix in the original coordinates, and define
\(\widehat\bSigma_{\mathrm{pool}}=
(n_1\widehat\bSigma_1+n_2\widehat\bSigma_2)/n\). For a given value of
\(\rho_{\mathrm{tar}}\), the sample transformation is
\begin{equation}
\label{eq:simulation-G}
\widehat \bA(\rho_{\mathrm{tar}})
=
\bSigma_{\mathrm{tar}}(\rho_{\mathrm{tar}})^{1/2}
\left(\widehat\bSigma_{\mathrm{pool}}+\hat{\gamma} \bI_p\right)^{-1/2},
\qquad
\hat{\gamma}
=
10^{-6}\frac{\operatorname{tr}(\widehat\bSigma_{\mathrm{pool}})}{p}.
\end{equation}
Here, all matrix square roots are the principal symmetric square roots.
The sample ridge ensures that the inverse square root is well defined
even when $\widehat{\bSigma}_{\mathrm{pool}}$ is singular, as necessarily
occurs when $p\gg n$. By construction,
$\widehat{\bA}(\rho_{\mathrm{tar}})
(\widehat{\bSigma}_{\mathrm{pool}}+\widehat{\gamma}\bI_p)
\widehat{\bA}(\rho_{\mathrm{tar}})^\top
=
\bSigma_{\mathrm{tar}}(\rho_{\mathrm{tar}})$.
Thus, at both the population and sample levels, the target is imposed
on the regularized pooled covariance matrix, rather than on either
group-specific covariance matrix. Under an alternative, the
individual transformed covariance matrices need not equal
$\bSigma_{\mathrm{tar}}(\rho_{\mathrm{tar}})$. Nevertheless, because the
transformation is nonsingular, it preserves whether the two population covariance matrices are equal. In every Monte Carlo replication, \(\widehat\bA(\rho_{\mathrm{tar}})\) is recomputed from that replication's two samples and then held fixed during
multiplier resampling. 

The deliberately high-dimensional simulation settings considered below are designed to complement the asymptotic theory by examining the finite-sample performance of the regularized construction in challenging regimes. We also include $\bA=\bI_p$ as the untransformed benchmark. This comparison allows us to distinguish the effect of applying the proposed transformation from that of varying its target correlation. For conciseness, in tables and figures reporting empirical rejection rates, the label \(\bA(\rho_{\mathrm{tar}})\) denotes the test implemented using \(\widehat{\bA}(\rho_{\mathrm{tar}})\) in \eqref{eq:simulation-G}.

We first compare the empirical sizes of \(T_n\) under different choices
of the transformation matrix \(\bA\). For both covariance classes, we set
\(\rho_1=\rho_2=\rho\), where
\(\rho\in\{0.1,0.5,0.9\}\), so that the null hypothesis holds. These values represent weak, moderate,
and strong dependence, respectively. Here, \(\rho\) determines the
covariance structure of the data-generating process, whereas
\(\rho_{\mathrm{tar}}\) specifies the equicorrelation target used to construct
\(\bA\). 

Tables~\ref{tab:size_equicorr} and~\ref{tab:size_toeplitz} report empirical
sizes under equicorrelation and Toeplitz covariances, respectively. Under equicorrelation
with \(\rho=0.1\), the untransformed test \(\bA=\bI_p\) severely under-rejects
(from \(0.1\%\) to \(1.6\%\)), whereas the transformed tests have rejection rates much closer to \(5\%\).
For \(\rho\in\{0.5,0.9\}\), the untransformed test generally has reasonable,
though imperfect, size control. Thus, transformation substantially
improves calibration when pairwise dependence is weak.

\begin{table}[!h]
\centering
\captionsetup{font=small,skip=3pt}
\caption{Simulation 1: Empirical sizes (in \%) of $T_n$ under equicorrelation covariance structures over \(1{,}000\) Monte Carlo
replications. Columns labeled \(\bA(r)\) use target
correlation \(\rho_{\rm tar}=r\), whereas \(\bA=\bI_p\) is the
untransformed benchmark.}
\label{tab:size_equicorr}
\scriptsize
\setlength{\tabcolsep}{3pt}
\renewcommand{\arraystretch}{0.6}
\resizebox{\textwidth}{!}{%
\begin{tabular}{@{}llc*{12}{c}@{}}
\toprule
& & &
\multicolumn{4}{c}{\(\rho=0.1\)} &
\multicolumn{4}{c}{\(\rho=0.5\)} &
\multicolumn{4}{c}{\(\rho=0.9\)} \\
\cmidrule(lr){4-7}
\cmidrule(lr){8-11}
\cmidrule(lr){12-15}
Model & \(\bn\) & \(p\)
& \(\bA(0.3)\) & \(\bA(0.5)\) & \(\bA(0.7)\) & \(\bA=\bI_p\)
& \(\bA(0.3)\) & \(\bA(0.5)\) & \(\bA(0.7)\) & \(\bA=\bI_p\)
& \(\bA(0.3)\) & \(\bA(0.5)\) & \(\bA(0.7)\) & \(\bA=\bI_p\) \\
\midrule

\multirow{6}{*}{1}
& \multirow{2}{*}{\(\bn_1\)}
& 100 & 4.6 & 5.0 & 5.1 & 0.1
      & 4.9 & 6.5 & 6.2 & 5.4
      & 5.7 & 6.2 & 6.9 & 6.0 \\
& & 200 & 5.0 & 5.1 & 5.6 & 0.1
        & 6.2 & 6.8 & 6.6 & 6.3
        & 4.2 & 4.4 & 4.1 & 4.3 \\
& \multirow{2}{*}{\(\bn_2\)}
& 100 & 4.6 & 5.0 & 4.7 & 0.5
      & 5.4 & 5.4 & 5.2 & 5.0
      & 4.4 & 5.3 & 4.9 & 5.4 \\
& & 200 & 5.9 & 6.0 & 6.0 & 0.4
        & 5.3 & 5.0 & 5.3 & 5.3
        & 5.7 & 6.1 & 5.6 & 5.8 \\
& \multirow{2}{*}{\(\bn_3\)}
& 100 & 5.8 & 6.2 & 5.7 & 1.1
      & 6.0 & 6.6 & 7.0 & 7.3
      & 5.1 & 5.4 & 5.3 & 4.9 \\
& & 200 & 4.9 & 5.3 & 5.3 & 0.5
        & 5.4 & 5.3 & 5.8 & 5.1
        & 5.8 & 6.0 & 5.6 & 6.1 \\
\midrule
\multirow{6}{*}{2}
& \multirow{2}{*}{\(\bn_1\)}
& 100 & 6.1 & 6.2 & 6.4 & 0.5
      & 4.7 & 4.6 & 4.3 & 4.5
      & 5.0 & 5.4 & 4.5 & 4.9 \\
& & 200 & 4.2 & 4.9 & 4.5 & 0.8
        & 5.6 & 5.6 & 5.4 & 5.9
        & 4.7 & 4.6 & 4.4 & 4.4 \\
& \multirow{2}{*}{\(\bn_2\)}
& 100 & 5.8 & 5.8 & 6.2 & 1.0
      & 4.9 & 4.5 & 5.4 & 4.8
      & 5.9 & 6.2 & 6.1 & 6.1 \\
& & 200 & 5.0 & 5.3 & 4.8 & 0.6
        & 6.9 & 6.4 & 6.1 & 7.3
        & 5.2 & 5.3 & 5.7 & 5.3 \\
& \multirow{2}{*}{\(\bn_3\)}
& 100 & 6.4 & 5.7 & 6.4 & 1.6
      & 4.0 & 4.0 & 3.9 & 4.0
      & 5.0 & 4.4 & 4.8 & 4.0 \\
& & 200 & 4.9 & 5.3 & 4.6 & 1.3
        & 5.4 & 5.2 & 4.9 & 4.6
        & 4.3 & 4.7 & 5.0 & 4.4 \\
\midrule
\multirow{6}{*}{3}
& \multirow{2}{*}{\(\bn_1\)}
& 100 & 4.8 & 6.2 & 5.8 & 0.3
      & 5.3 & 6.0 & 6.0 & 6.2
      & 4.0 & 4.4 & 4.4 & 4.7 \\
& & 200 & 5.3 & 5.6 & 5.5 & 0.1
        & 4.6 & 4.6 & 4.6 & 4.5
        & 5.1 & 5.2 & 5.1 & 4.9 \\
& \multirow{2}{*}{\(\bn_2\)}
& 100 & 6.2 & 6.6 & 6.3 & 0.8
      & 5.6 & 6.0 & 5.7 & 5.4
      & 5.4 & 5.7 & 5.6 & 5.8 \\
& & 200 & 6.1 & 6.7 & 6.6 & 0.5
        & 6.1 & 6.2 & 7.2 & 6.5
        & 5.1 & 5.5 & 5.3 & 5.8 \\
& \multirow{2}{*}{\(\bn_3\)}
& 100 & 5.2 & 5.6 & 6.5 & 0.8
      & 4.7 & 5.1 & 4.9 & 5.2
      & 5.2 & 5.4 & 5.2 & 5.4 \\
& & 200 & 4.4 & 4.6 & 5.2 & 0.6
        & 5.0 & 5.1 & 5.0 & 4.8
        & 3.9 & 4.2 & 4.2 & 4.1 \\
\bottomrule
\end{tabular}%
}
\end{table}

The benefit is more pronounced under Toeplitz covariances. Here,
\(\rho\) measures adjacent-coordinate correlation, while
\(\rho^{|a-b|}\) decays with separation. Even at \(\rho=0.9\), most pairs
may therefore be weakly correlated for large \(p\): strong local
correlation need not imply strong aggregate dependence. This pattern is consistent with
the untransformed test's persistent conservativeness. The transformed
tests generally maintain rejection rates close to the nominal level
across sample sizes, dimensions, and innovation distributions.

\begin{table}[!h]
\centering
\captionsetup{font=small,skip=3pt}
\caption{Simulation 1: Empirical sizes (in \%) of $T_n$ under Toeplitz covariance structures over \(1{,}000\) Monte Carlo
replications. Columns labeled \(\bA(r)\) use target
correlation \(\rho_{\rm tar}=r\), whereas \(\bA=\bI_p\) is the
untransformed benchmark.}
\label{tab:size_toeplitz}
\scriptsize
\setlength{\tabcolsep}{3pt}
\renewcommand{\arraystretch}{0.6}
\resizebox{\textwidth}{!}{%
\begin{tabular}{@{}llc*{12}{c}@{}}
\toprule
& & &
\multicolumn{4}{c}{\(\rho=0.1\)} &
\multicolumn{4}{c}{\(\rho=0.5\)} &
\multicolumn{4}{c}{\(\rho=0.9\)} \\
\cmidrule(lr){4-7}
\cmidrule(lr){8-11}
\cmidrule(lr){12-15}
Model & \(\bn\) & \(p\)
& \(\bA(0.3)\) & \(\bA(0.5)\) & \(\bA(0.7)\) & \(\bA=\bI_p\)
& \(\bA(0.3)\) & \(\bA(0.5)\) & \(\bA(0.7)\) & \(\bA=\bI_p\)
& \(\bA(0.3)\) & \(\bA(0.5)\) & \(\bA(0.7)\) & \(\bA=\bI_p\) \\
\midrule

\multirow{6}{*}{1}
& \multirow{2}{*}{\(\bn_1\)}
& 100 & 5.0 & 6.1 & 6.1 & 0.0
      & 5.3 & 5.7 & 6.2 & 0.0
      & 5.5 & 5.6 & 6.1 & 3.4 \\
& & 200 & 5.2 & 5.6 & 6.1 & 0.0
        & 4.2 & 4.6 & 4.6 & 0.0
        & 5.4 & 5.6 & 5.7 & 1.1 \\
& \multirow{2}{*}{\(\bn_2\)}
& 100 & 5.9 & 6.3 & 6.5 & 0.0
      & 5.4 & 6.2 & 5.9 & 0.0
      & 5.6 & 6.3 & 5.6 & 3.9 \\
& & 200 & 5.3 & 5.3 & 5.6 & 0.0
        & 5.0 & 4.9 & 5.1 & 0.0
        & 5.2 & 5.4 & 5.6 & 1.5 \\
& \multirow{2}{*}{\(\bn_3\)}
& 100 & 6.0 & 6.4 & 6.1 & 0.0
      & 5.6 & 5.7 & 5.2 & 0.0
      & 5.7 & 5.3 & 4.8 & 5.3 \\
& & 200 & 6.6 & 7.0 & 6.6 & 0.0
        & 5.1 & 5.7 & 5.7 & 0.0
        & 4.6 & 5.6 & 5.4 & 2.4 \\

\midrule
\multirow{6}{*}{2}
& \multirow{2}{*}{\(\bn_1\)}
& 100 & 5.8 & 6.1 & 6.3 & 0.0
      & 4.2 & 5.4 & 5.1 & 0.0
      & 4.7 & 5.4 & 5.3 & 3.0 \\
& & 200 & 4.5 & 4.5 & 4.9 & 0.0
        & 5.9 & 5.9 & 5.5 & 0.0
        & 4.3 & 4.4 & 4.7 & 0.5 \\
& \multirow{2}{*}{\(\bn_2\)}
& 100 & 6.0 & 6.1 & 6.3 & 0.0
      & 5.3 & 5.8 & 5.5 & 0.0
      & 5.0 & 5.5 & 5.8 & 2.7 \\
& & 200 & 4.4 & 4.8 & 4.8 & 0.0
        & 5.4 & 6.3 & 5.6 & 0.0
        & 5.6 & 5.7 & 5.6 & 1.1 \\
& \multirow{2}{*}{\(\bn_3\)}
& 100 & 5.6 & 5.3 & 5.3 & 0.0
      & 4.4 & 4.9 & 4.6 & 0.0
      & 5.2 & 5.8 & 5.6 & 3.1 \\
& & 200 & 7.0 & 7.5 & 7.0 & 0.0
        & 6.6 & 6.8 & 6.6 & 0.0
        & 6.1 & 6.1 & 6.1 & 2.0 \\

\midrule
\multirow{6}{*}{3}
& \multirow{2}{*}{\(\bn_1\)}
& 100 & 4.2 & 5.0 & 5.0 & 0.0
      & 5.8 & 6.1 & 6.2 & 0.0
      & 4.7 & 5.1 & 5.0 & 2.3 \\
& & 200 & 5.4 & 5.2 & 5.6 & 0.0
        & 6.0 & 5.6 & 6.0 & 0.0
        & 5.9 & 5.8 & 5.4 & 0.7 \\
& \multirow{2}{*}{\(\bn_2\)}
& 100 & 4.7 & 5.1 & 5.0 & 0.0
      & 5.9 & 6.7 & 6.2 & 0.0
      & 4.2 & 5.2 & 5.2 & 4.2 \\
& & 200 & 4.7 & 5.2 & 5.6 & 0.0
        & 4.9 & 4.8 & 5.1 & 0.0
        & 5.3 & 4.9 & 5.7 & 1.6 \\
& \multirow{2}{*}{\(\bn_3\)}
& 100 & 4.2 & 4.4 & 4.4 & 0.0
      & 5.8 & 6.2 & 6.1 & 0.0
      & 5.3 & 5.1 & 5.3 & 3.6 \\
& & 200 & 5.0 & 5.4 & 5.1 & 0.0
        & 4.8 & 5.3 & 5.2 & 0.0
        & 4.9 & 5.2 & 5.2 & 2.5 \\

\bottomrule
\end{tabular}%
}
\end{table}

Across the two covariance classes, the three target correlations
\(\rho_{\mathrm{tar}}\in\{0.3,0.5,0.7\}\) produce broadly comparable
empirical sizes. Although individual rejection percentages vary across
configurations, there is no systematic deterioration as
\(\rho_{\mathrm{tar}}\) changes over the range considered. The size
performance of \(T_n\) therefore appears relatively insensitive to the
choice of the target correlation within this range.
% We consequently
% use the intermediate value \(\rho_{\rm tar}=0.5\) as a representative
% default in the subsequent simulations.

We next examine the empirical power of \(T_n\) under Model~1, considering two classes of alternatives involving equicorrelation covariance matrices. Setting~(a): $\bSigma_1=\bI_p,\bSigma_2
=
(1-\rho_a)\bI_p
+
\rho_a\mathbf{1}_p\mathbf{1}_p^\top$, where \(\rho_a\) increases from 0 to 0.1. The null hypothesis holds at
\(\rho_a=0\), and the discrepancy between the two covariance matrices
increases with \(\rho_a\). Setting~(b): $\bSigma_1
=
(1-0.5)\bI_p
+
0.5\mathbf{1}_p\mathbf{1}_p^\top, \bSigma_2
=
(1-\rho_b)\bI_p
+
\rho_b\mathbf{1}_p\mathbf{1}_p^\top$, where $\rho_b\in\{0.1,0.2,\ldots,0.9\}$. Here, the null hypothesis holds at \(\rho_b=0.5\), whereas values of
\(\rho_b\) farther from \(0.5\) correspond to increasingly strong
alternatives. This setting allows us to assess power against both
decreases and increases in the common correlation relative to the
baseline value \(0.5\). For both settings, we take \((n_1,n_2)=(120,150)\) and \(p=200\).
Figure~\ref{fig:power_identity_targets} compares the empirical rejection rates
of the untransformed and target-transformed tests.

\begin{figure}[!h]
\centering
\includegraphics[width=0.96\textwidth]{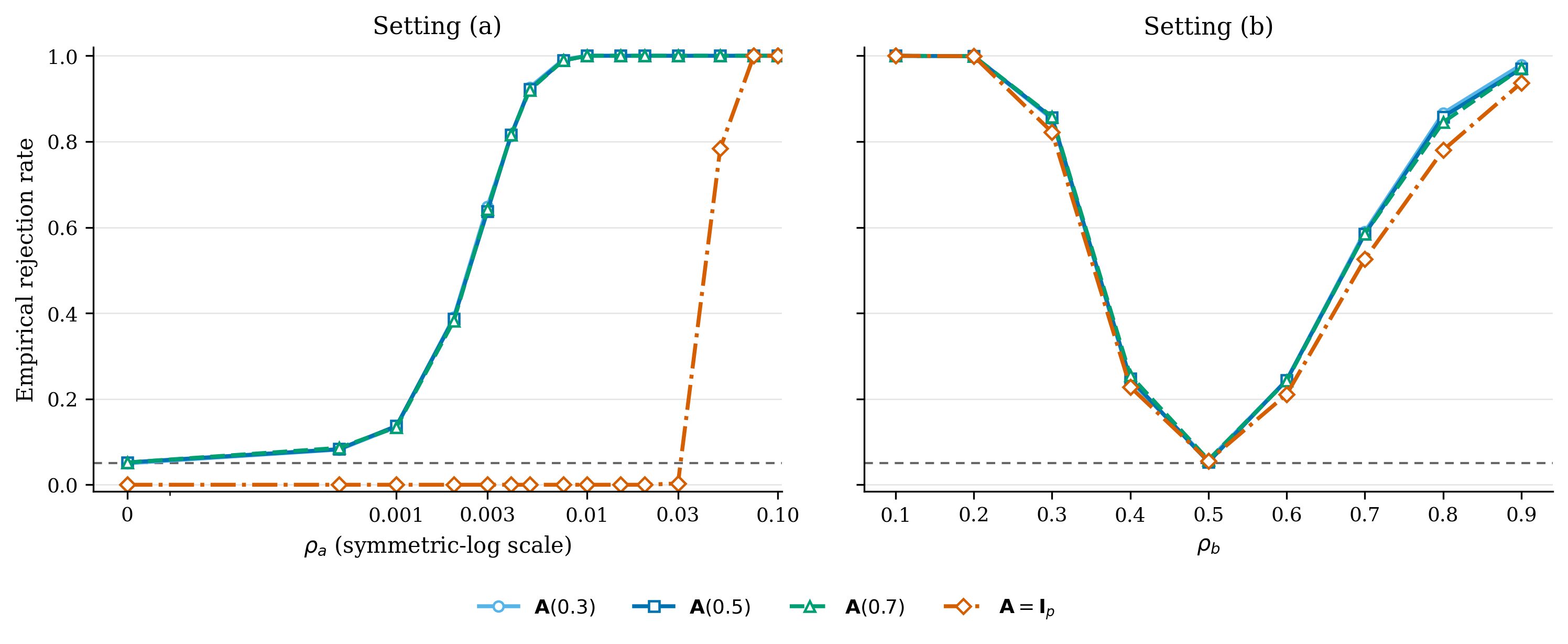}
\caption{Simulation 1: Empirical rejection rates of \(T_n\) under Setting~(a)
(left) and Setting~(b) (right). The horizontal
dashed line marks the nominal \(5\%\) level. The \(\rho_a\)-axis in
the left panel uses a symmetric logarithmic scale that is linear near
zero.}
\label{fig:power_identity_targets}
\end{figure}
Under Setting~(a), the untransformed test's ($\bA=\bI_p$) rejection rate is close to
zero at \(\rho_a=0\), consistent with Table~\ref{tab:size_equicorr},
whereas the transformed tests have empirical sizes close to \(5\%\). As \(\rho_a\)
increases, transformed power rises rapidly while the untransformed
test has low power for small departures. Transformation therefore
substantially improves detection of weak dense differences from
the identity covariance matrix.

In the right panel, corresponding to Setting~(b), the null hypothesis holds at \(\rho_b=0.5\), and
the discrepancy between the two covariance matrices increases as
\(\rho_b\) moves away from \(0.5\). The empirical rejection rates generally increase in both directions away from the
null value. The power advantage of the transformed tests over the
untransformed test is less pronounced than in Setting~(a), but remains
evident for moderate departures from \(0.5\).

To help interpret the different power patterns in
Figure~\ref{fig:power_identity_targets}, we compare \(J_n(\bA)\) introduced in
\eqref{eq:transformation-signal-ratio} across
the transformations considered above. Condition~(ii) of
Theorem~\ref{thm:power_consistency} shows that a larger value of
\(J_n(\bA)\) makes its sufficient signal condition easier to satisfy,
although it does not imply a finite-sample power ordering. The values are calculated under the same settings used for the power
experiments. 

% For the population diagnostic below, we use the deterministic counterpart
% of \eqref{eq:simulation-G}, obtained by replacing
% \(\widehat\bSigma_{\rm pool}+\gamma\bI_p\) with
% \(\bSigma_{\rm pool}=(n_1\bSigma_1+n_2\bSigma_2)/(n_1+n_2)\).
% Because Settings~(a) and~(b) are Gaussian, the covariance of the product
% coordinates used to construct \(\bR_A\) is evaluated from
% \[
% \operatorname{Cov}(X_aX_b,X_cX_d)
% =\Sigma_{ac}\Sigma_{bd}+\Sigma_{ad}\Sigma_{bc}.
% \]

Table~\ref{tab:geometry-adjusted-signal} provides the values of \(J_n(\bA)\) under the alternative Settings~(a) and~(b). At the null values
\(\rho_a=0\) and \(\rho_b=0.5\), the covariance difference vanishes, and therefore
\(J_n(\bA)=0\) for every transformation. Under Setting~(a), \(J_n(\bI_p)\) remains very small for weak departures from the identity covariance matrix. In contrast, all three transformations substantially
increase \(J_n(\bA)\). For example, at \(\rho_a=0.001\), \(J_n(\bA)\) increases from \(0.009\) under \(\bA=\bI_p\) to approximately \(1.1\)
under the transformed versions. This increase is consistent with
the substantial power gains observed
in the left panel of Figure~\ref{fig:power_identity_targets}. The three target correlations produce very similar values of
\(J_n(\bA)\) under Setting~(a). This pattern suggests that the improvement is associated more with chaning the original weak-dependence representation than with the precise value of \(\rho_{\mathrm{tar}}\) over the range considered.

\begin{table*}[!h]
\centering
\captionsetup{font=small,skip=3pt}
\caption{Simulation~1: Values of \(J_n(\bA)\) at selected parameter values under
Settings~(a) and~(b).}
\label{tab:geometry-adjusted-signal}
\scriptsize
\setlength{\tabcolsep}{3pt}
\renewcommand{\arraystretch}{0.6}
\begin{tabular*}{\textwidth}{
@{\extracolsep{\fill}}
ccccc
@{\hspace{1.5em}}
ccccc
@{}
}
\toprule
\multicolumn{5}{c}{Setting~(a)}
&
\multicolumn{5}{c}{Setting~(b)}
\\
\cmidrule(lr){1-5}
\cmidrule(lr){6-10}
\(\rho_a\)
&
\(\bI_p\)
&
\(\bA(0.3)\)
&
\(\bA(0.5)\)
&
\(\bA(0.7)\)
&
\(\rho_b\)
&
\(\bI_p\)
&
\(\bA(0.3)\)
&
\(\bA(0.5)\)
&
\(\bA(0.7)\)
\\
\midrule
0.000 & 0.000 & 0.000  & 0.000  & 0.000
&
0.300 & 7.284 & 7.294  & 7.278  & 7.275
\\
0.001 & 0.009 & 1.090  & 1.102  & 1.105
&
0.400 & 1.554 & 1.562  & 1.553  & 1.551
\\
0.002 & 0.037 & 3.658  & 3.698  & 3.706
&
0.500 & 0.000 & 0.000  & 0.000  & 0.000
\\
0.005 & 0.234 & 14.023 & 14.164 & 14.191
&
0.600 & 1.097 & 1.127  & 1.100  & 1.094
\\
0.010 & 0.932 & 28.850 & 29.091 & 29.136
&
0.700 & 3.685 & 3.876  & 3.710  & 3.676
\\
\bottomrule
\end{tabular*}
\end{table*}

The pattern is different under Setting~(b), where the baseline
covariance already has common correlation \(0.5\). In this case, the untransformed procedure already has a substantial value of
\(J_n(\bI_p)\) away from the null, and the transformations produce only modest
changes in \(J_n(\bA)\). No target correlation dominates uniformly:
\(\bA(0.3)\) gives a small improvement for some alternatives, whereas the
other transformations yield values of \(J_n(\bA)\) close to or slightly below
\(J_n(\bI_p)\). This is consistent with the smaller
power differences observed in the right panel of
Figure~\ref{fig:power_identity_targets}.

Across both settings, the power curves corresponding to
\(\rho_{\mathrm{tar}}\in\{0.3,0.5,0.7\}\) are nearly indistinguishable at the
scale of the figure. This suggests that the finite-sample performance of
\(T_n\) is relatively insensitive to the choice of
\(\rho_{\mathrm{tar}}\) over the range considered. Taken together with the
size results, these findings support the use of the intermediate value
\(\rho_{\mathrm{tar}}=0.5\) as a representative specification in the
subsequent comparisons.

\subsection{Simulation 2}
In this subsection, we compare the performance of the proposed test $T_n$
using $\bA(0.5)$ with the three existing tests $T_{LC}$, $T_{CZZW}$,
and $T_{WZZ}$ in terms of size control and power. The sample sizes and
dimension are the same as in Simulation~1. We first compare the procedures in terms of size control, using the same null configurations as in Simulation~1. Tables~\ref{tab:size2_equicorr}
and~\ref{tab:size2_toeplitz} report the empirical sizes (in \%) of \(T_{LC}\), \(T_{CZZW}\), \(T_{WZZ}\), and \(T_n\), based on \(1{,}000\) Monte Carlo replications, under the equicorrelation and
Toeplitz covariance structures, respectively.

Consider first the equicorrelation results in
Table~\ref{tab:size2_equicorr}. \(T_{LC}\) is generally
liberal, with empirical sizes ranging from \(7.3\%\) to
\(26.9\%\). The size distortion is particularly pronounced under the
centered Chi-squared (Model 2) and Laplace (Model 3) distributions and tends to increase with
the dependence parameter \(\rho\). This pattern is consistent with the
findings of \citet{wang_two-sample_2024}, who noted that the normal
approximation underlying \(T_{LC}\) may be inaccurate in the presence
of strong dependence. The maximum-type test \(T_{CZZW}\) provides reasonable size control in
many of the Gaussian configurations, although it is moderately liberal
in some cases. Under the non-Gaussian models, however, its performance
is less stable. In particular, it can be severely conservative when
\(\rho=0.1\), with empirical sizes as low as \(0.5\%\).
Its empirical size generally improves as \(\rho\) increases, but some
under-rejection remains. The normal-reference test \(T_{WZZ}\) generally improves upon
\(T_{LC}\) and \(T_{CZZW}\) when \(\rho\in\{0.5,0.9\}\). Nevertheless,
it is conservative when the coordinates are weakly correlated. For
\(\rho=0.1\), its empirical sizes range from \(0.1\%\)
to \(3.3\%\). This observation is consistent with the findings in
Simulation~1 that an untransformed \(\ell_2\)-type procedure can be
poorly calibrated when the coordinates exhibit weak aggregate
dependence. By comparison, the proposed test \(T_n\) provides the most stable size
control across the configurations considered. Its empirical sizes range from \(4.0\%\) to \(6.8\%\) and remain close to the
nominal \(5\%\) level across all configurations. 

\begin{table}[!h]
\centering
\captionsetup{font=small,skip=3pt}
\caption{Simulation 2: Empirical sizes (in \%) of $T_{LC}$, $T_{CZZW}$,
$T_{WZZ}$, and $T_n$ under equicorrelation covariance structure over
\(1{,}000\) Monte Carlo replications.  The proposed test $T_n$ uses
$\bA(0.5)$.}
\label{tab:size2_equicorr}
\scriptsize
\setlength{\tabcolsep}{6pt}
\renewcommand{\arraystretch}{0.6}
\resizebox{\textwidth}{!}{%
\begin{tabular}{@{}llc*{12}{c}@{}}
\toprule
& & &
\multicolumn{4}{c}{\(\rho=0.1\)} &
\multicolumn{4}{c}{\(\rho=0.5\)} &
\multicolumn{4}{c}{\(\rho=0.9\)} \\
\cmidrule(lr){4-7}
\cmidrule(lr){8-11}
\cmidrule(lr){12-15}
Model & \(\bn\) & \(p\)
& $T_{LC}$ & $T_{CZZW}$ & $T_{WZZ}$ & $T_n$
& $T_{LC}$ & $T_{CZZW}$ & $T_{WZZ}$ & $T_n$
& $T_{LC}$ & $T_{CZZW}$ & $T_{WZZ}$ & $T_n$ \\
\midrule

\multirow{6}{*}{1}
& \multirow{2}{*}{\(\bn_1\)}
& 100 & 7.3 & 8.3 & 1.7 & 5.0
      & 10.6 & 6.9 & 4.9 & 6.5
      & 10.8 & 6.3 & 6.3 & 6.2 \\
& & 200 & 8.8 & 7.7 & 2.5 & 5.1
        & 10.9 & 9.1 & 6.0 & 6.8
        & 8.5 & 4.3 & 4.0 & 4.4 \\
& \multirow{2}{*}{\(\bn_2\)}
& 100 & 9.7 & 6.2 & 3.3 & 5.0
      & 11.1 & 6.3 & 4.7 & 5.4
      & 10.2 & 4.9 & 4.8 & 5.3 \\
& & 200 & 8.9 & 5.3 & 2.8 & 6.0
        & 10.1 & 5.9 & 5.0 & 5.0
        & 11.0 & 6.3 & 5.1 & 6.1 \\
& \multirow{2}{*}{\(\bn_3\)}
& 100 & 7.6 & 5.4 & 3.0 & 6.2
      & 10.8 & 6.3 & 6.1 & 6.6
      & 11.0 & 5.0 & 4.6 & 5.4 \\
& & 200 & 8.9 & 5.7 & 3.2 & 5.3
        & 10.1 & 6.4 & 4.8 & 5.3
        & 11.7 & 5.5 & 5.6 & 6.0 \\

\midrule
\multirow{6}{*}{2}
& \multirow{2}{*}{\(\bn_1\)}
& 100 & 9.0 & 1.0 & 0.1 & 6.2
      & 17.8 & 3.4 & 4.2 & 4.6
      & 25.8 & 4.4 & 4.5 & 5.4 \\
& & 200 & 8.9 & 0.8 & 0.3 & 4.9
        & 19.8 & 4.5 & 4.8 & 5.6
        & 23.9 & 3.5 & 4.0 & 4.6 \\
& \multirow{2}{*}{\(\bn_2\)}
& 100 & 9.9 & 1.1 & 0.4 & 5.8
      & 20.6 & 2.4 & 3.7 & 4.5
      & 26.9 & 5.2 & 5.8 & 6.2 \\
& & 200 & 10.1 & 0.5 & 0.6 & 5.3
        & 21.6 & 6.0 & 6.1 & 6.4
        & 25.6 & 4.3 & 4.5 & 5.3 \\
& \multirow{2}{*}{\(\bn_3\)}
& 100 & 10.1 & 1.4 & 1.1 & 5.7
      & 19.9 & 2.9 & 3.8 & 4.0
      & 26.7 & 3.4 & 4.0 & 4.4 \\
& & 200 & 9.7 & 0.9 & 1.0 & 5.3
        & 18.5 & 3.1 & 4.7 & 5.2
        & 25.5 & 3.6 & 4.2 & 4.7 \\
\midrule
\multirow{6}{*}{3}
& \multirow{2}{*}{\(\bn_1\)}
& 100 & 8.5 & 2.3 & 1.6 & 6.2
      & 13.5 & 4.8 & 5.5 & 6.0
      & 18.2 & 4.5 & 3.9 & 4.4 \\
& & 200 & 8.6 & 2.7 & 1.5 & 5.6
        & 12.7 & 5.3 & 4.0 & 4.6
        & 17.7 & 5.6 & 4.6 & 5.2 \\
& \multirow{2}{*}{\(\bn_2\)}
& 100 & 8.6 & 2.0 & 1.7 & 6.6
      & 15.9 & 5.8 & 5.0 & 6.0
      & 20.0 & 6.6 & 5.3 & 5.7 \\
& & 200 & 8.2 & 1.9 & 2.6 & 6.7
        & 15.2 & 6.6 & 5.8 & 6.2
        & 20.1 & 5.3 & 4.7 & 5.5 \\
& \multirow{2}{*}{\(\bn_3\)}
& 100 & 8.0 & 2.1 & 2.1 & 5.6
      & 13.8 & 5.9 & 4.5 & 5.1
      & 21.3 & 5.3 & 4.8 & 5.4 \\
& & 200 & 8.9 & 1.5 & 2.6 & 4.6
        & 13.2 & 5.0 & 3.9 & 5.1
        & 16.7 & 4.0 & 4.0 & 4.2 \\

\bottomrule
\end{tabular}%
}
\end{table}

Under Toeplitz covariances (Table~\ref{tab:size2_toeplitz}),
\(T_{LC}\) has reasonable size in several Gaussian (Model 1) and Laplace (Model 3) settings,
but reaches \(12.8\%\) under Model 2.
\(T_{CZZW}\) is reasonably calibrated in most Gaussian settings
(with sizes up to \(8.0\%\)), but often conservative under non-Gaussian
models, especially for \(\rho\in\{0.1,0.5\}\). Its empirical
sizes under Model 2 are \(0.1\%-0.5\%\) for \(\rho=0.1\) and \(1.0\%-1.7\%\) for \(\rho=0.5\);
calibration improves at \(\rho=0.9\), though it remains variable.  \(T_{WZZ}\) is the most conservative of the four procedures:
its empirical sizes do not exceed \(1.2\%\) when \(\rho=0.1\), remain
low in most configurations when \(\rho=0.5\), and generally remain
below \(5\%\) even when \(\rho=0.9\). As discussed in
Simulation~1, weak aggregate dependence can persist despite large
adjacent-coordinate correlation, a pattern consistent with the observed conservativeness.
In contrast, \(T_n\) has empirical sizes of \(4.4\%-7.5\%\), mostly near \(5\%\).
Together, Tables~\ref{tab:size2_equicorr}
and~\ref{tab:size2_toeplitz} show that \(T_n\) has the most stable
empirical size among the four procedures considered.
\begin{table}[!h]
\centering
\captionsetup{font=small,skip=3pt}
\caption{Simulation 2: Empirical sizes (in \%) of $T_{LC}$, $T_{CZZW}$,
$T_{WZZ}$, and $T_n$ under  Toeplitz covariance structure over
\(1{,}000\) Monte Carlo replications.  The proposed test $T_n$ uses
$\bA(0.5)$.}
\label{tab:size2_toeplitz}
\scriptsize
\setlength{\tabcolsep}{6pt}
\renewcommand{\arraystretch}{0.6}
\resizebox{\textwidth}{!}{%
\begin{tabular}{@{}llc*{12}{c}@{}}
\toprule
& & &
\multicolumn{4}{c}{\(\rho=0.1\)} &
\multicolumn{4}{c}{\(\rho=0.5\)} &
\multicolumn{4}{c}{\(\rho=0.9\)} \\
\cmidrule(lr){4-7}
\cmidrule(lr){8-11}
\cmidrule(lr){12-15}
Model & \(\bn\) & \(p\)
& $T_{LC}$ & $T_{CZZW}$ & $T_{WZZ}$ & $T_n$
& $T_{LC}$ & $T_{CZZW}$ & $T_{WZZ}$ & $T_n$
& $T_{LC}$ & $T_{CZZW}$ & $T_{WZZ}$ & $T_n$ \\
\midrule

\multirow{6}{*}{1}
& \multirow{2}{*}{\(\bn_1\)}
& 100 & 5.1 & 8.0 & 0.8 & 6.1
      & 5.3 & 6.2 & 0.9 & 5.7
      & 8.2 & 7.3 & 3.9 & 5.6 \\
& & 200 & 6.1 & 5.8 & 0.4 & 5.6
        & 4.3 & 7.5 & 0.1 & 4.6
        & 6.7 & 7.6 & 2.5 & 5.6 \\
& \multirow{2}{*}{\(\bn_2\)}
& 100 & 4.5 & 5.6 & 0.9 & 6.3
      & 5.5 & 5.3 & 1.5 & 6.2
      & 7.5 & 7.9 & 4.3 & 6.3 \\
& & 200 & 5.6 & 7.2 & 0.4 & 5.3
        & 5.5 & 5.7 & 0.7 & 4.9
        & 6.9 & 6.4 & 3.3 & 5.4 \\
& \multirow{2}{*}{\(\bn_3\)}
& 100 & 5.0 & 6.1 & 1.2 & 6.4
      & 5.7 & 5.0 & 2.0 & 5.7
      & 9.3 & 6.6 & 4.7 & 5.3 \\
& & 200 & 4.9 & 5.7 & 0.6 & 7.0
        & 5.4 & 6.0 & 0.8 & 5.7
        & 5.7 & 4.8 & 3.4 & 5.6 \\

\midrule
\multirow{6}{*}{2}
& \multirow{2}{*}{\(\bn_1\)}
& 100 & 10.6 & 0.4 & 0.0 & 6.1
      & 9.3 & 1.5 & 0.0 & 5.4
      & 10.3 & 5.9 & 1.7 & 5.4 \\
& & 200 & 11.2 & 0.3 & 0.0 & 4.5
        & 6.4 & 1.3 & 0.0 & 5.9
        & 7.6 & 5.6 & 1.3 & 4.4 \\
& \multirow{2}{*}{\(\bn_2\)}
& 100 & 9.8 & 0.1 & 0.0 & 6.1
      & 7.8 & 1.1 & 0.0 & 5.8
      & 9.1 & 4.3 & 2.2 & 5.5 \\
& & 200 & 7.7 & 0.3 & 0.0 & 4.8
        & 7.3 & 1.0 & 0.0 & 6.3
        & 6.1 & 4.7 & 1.3 & 5.7 \\
& \multirow{2}{*}{\(\bn_3\)}
& 100 & 12.8 & 0.5 & 0.0 & 5.3
      & 11.1 & 1.7 & 0.1 & 4.9
      & 10.1 & 5.2 & 2.4 & 5.8 \\
& & 200 & 10.0 & 0.3 & 0.0 & 7.5
        & 5.4 & 1.2 & 0.0 & 6.8
        & 7.9 & 4.2 & 2.7 & 6.1 \\

\midrule
\multirow{6}{*}{3}
& \multirow{2}{*}{\(\bn_1\)}
& 100 & 4.9 & 1.4 & 0.0 & 5.0
      & 6.0 & 3.6 & 0.1 & 6.1
      & 7.6 & 6.4 & 2.9 & 5.1 \\
& & 200 & 4.4 & 1.4 & 0.0 & 5.2
        & 4.8 & 3.1 & 0.0 & 5.6
        & 6.5 & 6.3 & 2.3 & 5.8 \\
& \multirow{2}{*}{\(\bn_2\)}
& 100 & 5.8 & 1.8 & 0.0 & 5.1
      & 5.5 & 2.9 & 0.3 & 6.7
      & 8.9 & 5.0 & 3.0 & 5.2 \\
& & 200 & 6.3 & 1.2 & 0.0 & 5.2
        & 6.2 & 3.0 & 0.0 & 4.8
        & 6.8 & 7.0 & 2.6 & 4.9 \\
& \multirow{2}{*}{\(\bn_3\)}
& 100 & 7.2 & 2.3 & 0.0 & 4.4
      & 6.5 & 2.1 & 0.6 & 6.2
      & 7.9 & 5.5 & 3.3 & 5.1 \\
& & 200 & 4.2 & 2.2 & 0.0 & 5.4
        & 5.6 & 2.2 & 0.0 & 5.3
        & 9.0 & 4.6 & 4.9 & 5.2 \\

\bottomrule
\end{tabular}%
}
\end{table}

For the power comparison, we first consider the equicorrelation
alternatives introduced in Simulation~1.
Figure~\ref{fig:eqcor_power} presents the empirical rejection rates of the four tests. Under Setting~(a), the
null hypothesis holds at \(\rho_a=0\), where
\(\bSigma_1=\bSigma_2=\bI_p\). At this null point, the empirical rejection rate of \(T_{LC}\) is
slightly above the nominal \(5\%\) level, those of \(T_{CZZW}\) and
\(T_n\) are close to \(5\%\), and that of \(T_{WZZ}\) is substantially
below the nominal level.
As \(\rho_a\) increases, the discrepancy between the two covariance
matrices becomes larger and the rejection probabilities of the four
tests generally increase. The proposed test is particularly powerful against
small dense departures from the identity covariance matrix: its
empirical power is \(38.7\%\) at \(\rho_a=0.002\), \(92.2\%\) at
\(\rho_a=0.005\), and reaches \(100\%\) by \(\rho_a=0.010\). Over the
same range, \(T_{CZZW}\) remains close to its null rejection level,
\(T_{WZZ}\) has relatively low power, and the power of \(T_{LC}\)
increases more slowly. Thus, the proposed transformation yields
substantial power gains against weak dense alternatives while
maintaining accurate size control.
\begin{figure}[!h]
\centering
\includegraphics[width=0.96\textwidth]{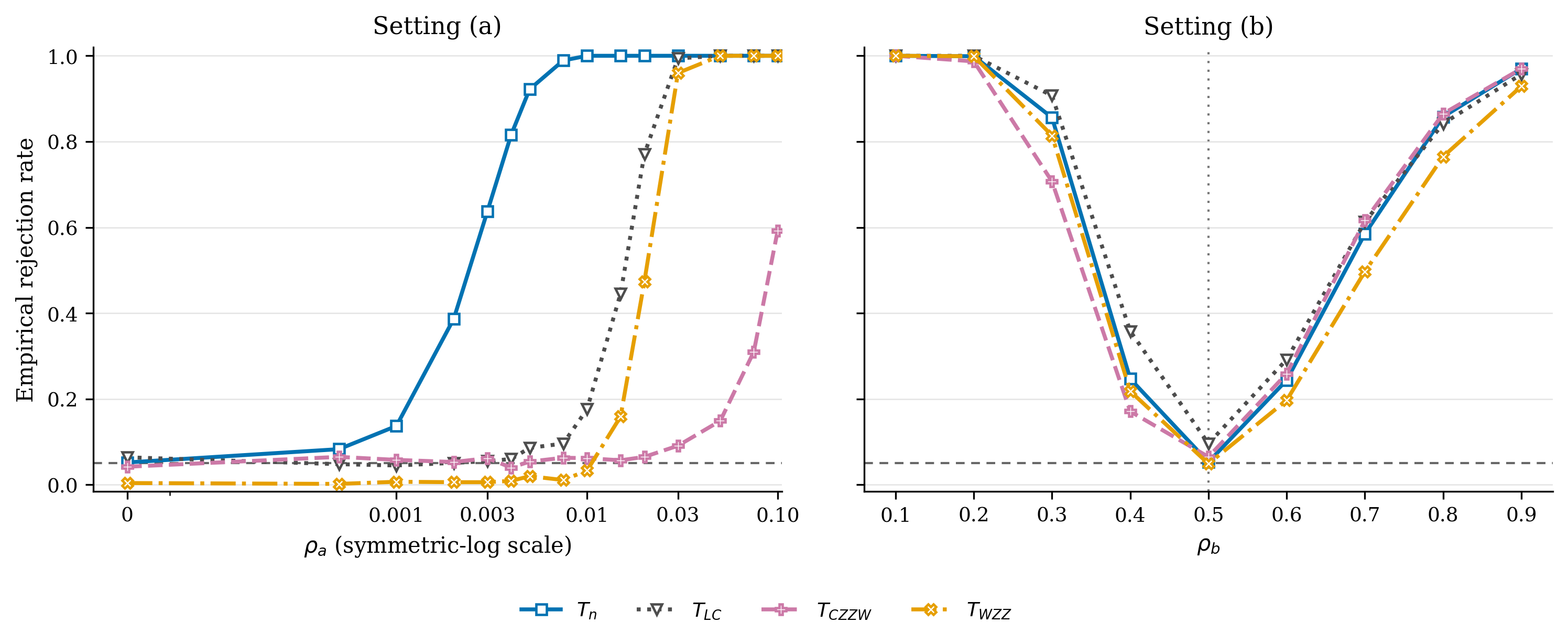}
\caption{Simulation 2: Empirical rejection rates of $T_n$ with $\bA(0.5)$ and the benchmark
tests under Settings~(a) (left) and~(b) (right). The horizontal dashed line marks the nominal \(5\%\) level. The
horizontal axis in~(a) uses a symmetric logarithmic scale that is linear near zero.}
\label{fig:eqcor_power}
\end{figure}

Under Setting~(b), the null hypothesis holds at \(\rho_b=0.5\). At this
point, \(T_{CZZW}\), \(T_{WZZ}\), and \(T_n\) have rejection
rates close to \(5\%\), whereas \(T_{LC}\)
over-rejects. As \(\rho_b\) moves away from \(0.5\), the discrepancy
between \(\bSigma_1\) and \(\bSigma_2\) increases, and the rejection
rates generally increase in both directions. Although
\(T_{LC}\) sometimes has a slightly higher raw rejection probability,
this apparent advantage should be interpreted cautiously because the
test is already oversized at the null. Among the procedures with
reasonable null calibration, \(T_n\) remains competitive and has higher
empirical power than \(T_{WZZ}\) at moderate departures from \(0.5\).
Relative to \(T_{CZZW}\), its power advantage occurs for downward
departures; for upward departures, its power is similar or slightly lower.

In addition to the equicorrelation alternatives considered above, we
examine a block-structured alternative designed to assess the effect of
high-variance nuisance coordinates. We set
$\bSigma_1
=\operatorname{diag}\{\bI_q,h\bI_{p-q}\}$ and $\bSigma_2(\rho_{\mathrm{block}})
=\operatorname{diag}
\{\bB_q(\rho_{\mathrm{block}}),h\bI_{p-q}\}$, where
$\bB_q(\rho)
=(1-\rho)\bI_q+\rho\mathbf{1}_q\mathbf{1}_q^\top$, \(q=80\), and \(h=25\). We take \(p=200\), so the first \(q=80\) coordinates constitute the
signal block and the remaining \(p-q=120\) coordinates constitute a
high-variance nuisance block. Samples follow Model~1 with \((n_1,n_2)=(120,150)\). The null hypothesis holds at \(\rho_{\mathrm{block}}=0\), for which
\(\bSigma_1=\bSigma_2\) and the coordinates are mutually uncorrelated
in both groups. Under the alternative \(\rho_{\mathrm{block}}>0\), only the
covariance structure of the signal block changes: its pairwise
correlations increase from 0 to \(\rho_{\mathrm{block}}\), while its
marginal variances remain equal to one. The nuisance coordinates retain
marginal variance \(h=25\), and all off-diagonal covariances involving
the nuisance block remain zero. In particular, the squared Frobenius
norm of the population covariance difference is $\|
\bSigma_2(\rho_{\mathrm{block}})-\bSigma_1\|_{\mathrm{F}}^2
=
q(q-1)\rho_{\mathrm{block}}^2$. Thus, the nuisance block does not contribute to the population signal,
but its large marginal variances can substantially inflate the sampling
variability of an unstudentized statistic and reduce the effective
signal-to-noise ratio.

In addition to \(T_{LC}\), \(T_{CZZW}\), and \(T_{WZZ}\), we consider
an unstudentized \(\ell_2\)-type bootstrap test, denoted by \(T_2\).
It uses the original coordinates \(\bA=\bI_p\) and omits the diagonal
studentization in \eqref{eq:proposed-statistic}, with Gaussian multiplier
calibration. In contrast, \(T_n\) uses \(\bA(0.5)\) and diagonal
studentization. Their comparison therefore assesses the combined
procedure rather than isolating studentization alone.

Figure~\ref{fig:studentized_target_power} reports the empirical
rejection rates under this block-structured alternative. The
tests \(T_2\), \(T_{LC}\), and \(T_{WZZ}\) exhibit negligible power over much
of the parameter range. These procedures use \(\ell_2\)-type
aggregation without the coordinatewise diagonal studentization
employed by \(T_n\). Their low power is consistent with increased
sampling variability introduced by the high-variance nuisance coordinates,
which carry no population signal. By contrast,
\(T_n\) combines the target transformation with diagonal studentization
to control coordinatewise variability and aggregate the weak but
dense signal. Its empirical size is close to the
nominal level when \(\rho_{\mathrm{block}}=0\), and its empirical powers are
\(58.9\%\), \(80.4\%\), and \(90.3\%\) when
\(\rho_{\mathrm{block}}=0.05\), \(0.10\), and \(0.15\), respectively. The maximum-type test \(T_{CZZW}\) is less affected by  the aggregate
variability contributed by  the nuisance coordinates, but becomes powerful only at larger values of
\(\rho_{\mathrm{block}}\). This pattern is consistent with its sensitivity to large individual
covariance differences rather than weak differences distributed over
many entries.

\begin{figure}[!h]
\centering
\includegraphics[width=0.6\textwidth]{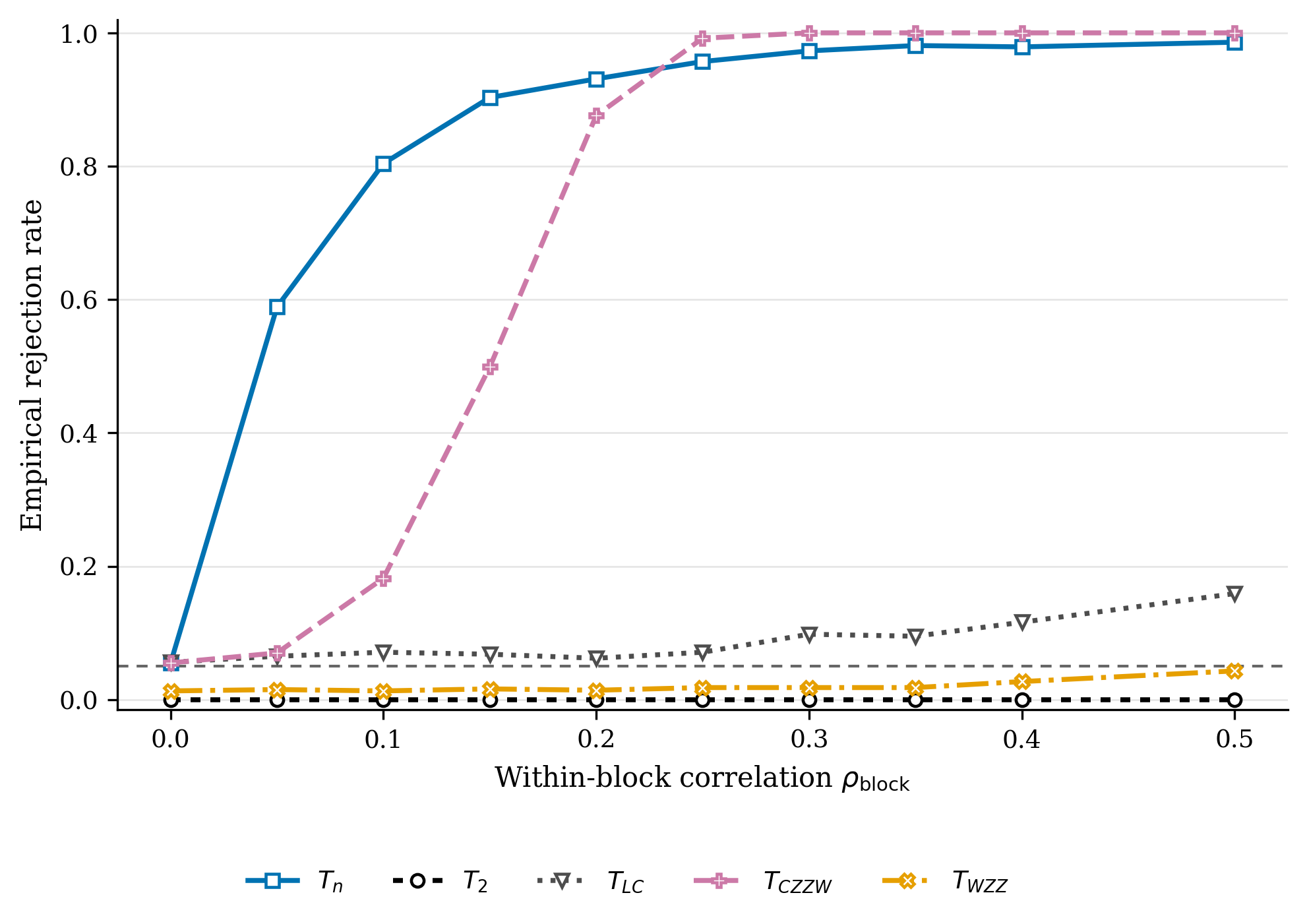}
\caption{Simulation 2: Empirical rejection rates under the block-structured
alternative. The proposed test \(T_n\) uses \(\bA(0.5)\), whereas \(T_2\) is the untransformed and unstudentized  bootstrap test.
The horizontal dashed line marks the nominal
\(5\%\) level.}
\label{fig:studentized_target_power}
\end{figure}
 
Overall, in the configurations examined, the proposed procedure
maintains empirical size close to the nominal level and shows strong
sensitivity to weak, dense covariance differences in the presence of
high-variance nuisance coordinates. Because \(T_n\) and \(T_2\) differ
in both transformation and studentization, their comparison reflects
the combined effect of these two components and does not isolate the
contribution of studentization.

\section{Real Data Application}\label{sec:realdata}

We analyze breast cancer gene-expression data from The Cancer Genome Atlas
Breast Invasive Carcinoma project (TCGA-BRCA), a public resource for studying
molecular differences among breast tumors
\citep{tcga_brca_molecular_portraits_2012}. The RNA sequencing 
profiles and subtype annotations were obtained from the University of
California, Santa Cruz Xena platform \citep{goldman_visualizing_2020}.
We compare primary tumors classified as Basal-like or Luminal~A using the
supplied PAM50 labels, which are based on a 50-gene subtype classifier
\citep{parker_pam50_2009}. The two groups contain \(n_1=141\) and
\(n_2=421\) tumors, respectively, from distinct patients.

Our analysis focuses on \(p=304\) genes drawn from two established gene sets
in the Molecular Signatures Database Hallmark collection: E2F
Targets and G2/M Checkpoint \citep{liberzon_molecular_2015}. Both sets concern
cell division, and previous analyses of TCGA data have linked E2F
activity to tumor proliferation and aggressiveness
\citep{oshi_e2f_2020}.\footnote{We used the Xena files
\texttt{HiSeqV2.gz} and \texttt{BRCA\_clinicalMatrix} from
\url{https://tcga.xenahubs.net/download/TCGA.BRCA.sampleMap/}, and the Molecular Signatures Database release \texttt{h.all.v2025.1.Hs.symbols.gmt}.}
The two gene sets contain 187 and 183 genes, with 66 shared genes
counted only once. They provide a biologically motivated panel for
examining whether gene-specific variability and patterns of
co-expression differ between the two breast cancer subtypes. Because \(p=304\) exceeds the sample size \(n_1=141\) of the
Basal-like group, this is a high-dimensional covariance comparison.

We retain the log-transformed expression scale supplied by Xena, average expression values
for duplicate gene symbols, and exclude genes with missing values or negligible variance.
 Within each subtype, each gene is centered by its sample mean. We do not rescale the genes
to have unit variance, because the formal analysis concerns covariance
rather than correlation matrices.

Figure~\ref{fig:realdata_correlations} displays the sample correlation matrices for the two subtypes and their difference using the same gene ordering. Overall, the two subtype-specific correlation matrices exhibit broadly similar patterns in terms of the distribution of positive and negative correlations, while some differences in correlation strength are visible across the gene panel. However, the extent to which these apparent differences reflect systematic changes in the underlying dependence structure is difficult to assess from visual inspection alone. This motivates a formal statistical comparison of the two correlation structures.

\begin{figure}[!h]
\centering
\includegraphics[width=0.96\textwidth]{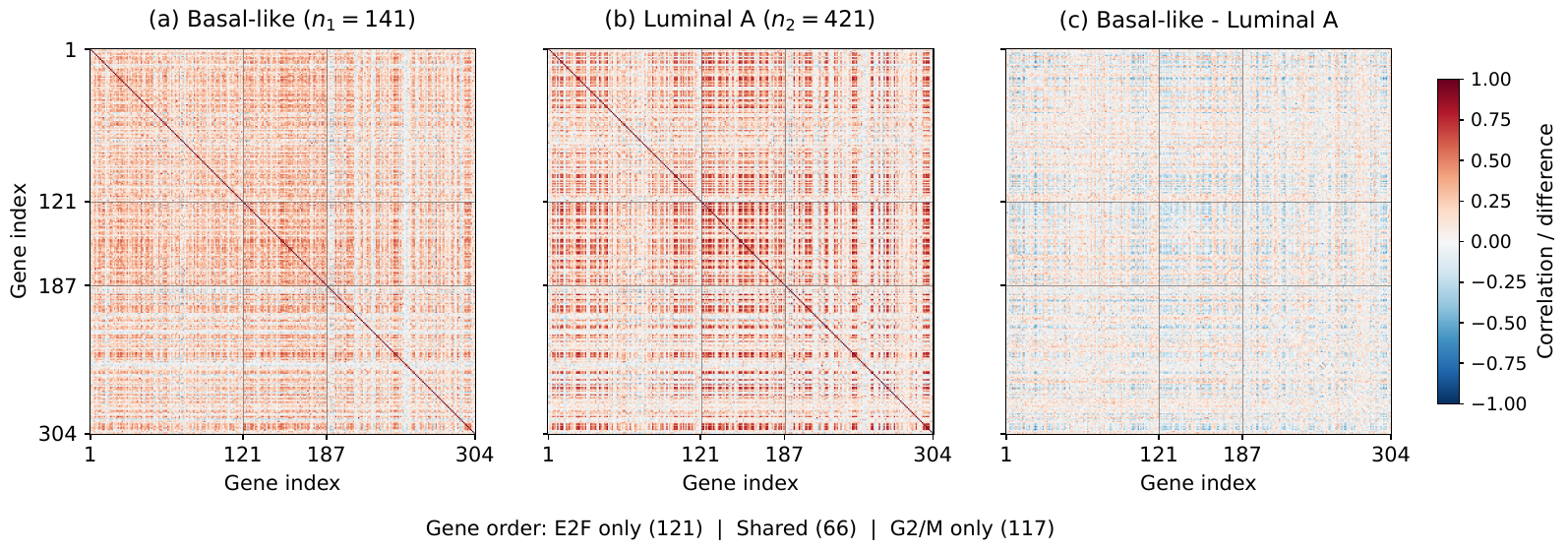}
\caption{Sample correlations for Basal-like and Luminal~A tumors and their
signed difference. All panels share one color scale and the same 304 genes,
ordered by gene-set membership and alphabetically within each block;
shared genes appear once.}
\label{fig:realdata_correlations}
\end{figure}

We apply \(T_n\) using both \(\bA=\bI_p\) and the estimated pooled-covariance transformation \(\widehat{\bA}(0.5)\) in \eqref{eq:simulation-G} to this dataset. Gaussian multiplier calibration uses \(B=2000\) draws, with the estimated transformation and studentizers held fixed during resampling. For comparison, we also apply the U-statistic-based test \(T_{LC}\) \citep{li_two_2012}, the maximum-type bootstrap test \(T_{CZZW}\) \citep{chang_comparing_2017}, and the normal-reference test \(T_{WZZ}\) \citep{wang_two-sample_2024},  with \(2000\) bootstrap draws used for \(T_{CZZW}\).

All five procedures produce \(p\)-values well below the nominal significance level of \(5\%\), consistently rejecting the null hypothesis of equal covariance matrices. The agreement across methods with substantially different constructions provides strong empirical evidence that the covariance structures of the two breast cancer subtypes differ within the selected gene panel. 

Although all five tests strongly reject covariance equality, their
finite-sample calibration may be affected by the high dimension,
unequal sample sizes, and complex dependence structure of the
application. To assess the robustness of the empirical conclusion, we
conduct an application-calibrated Gaussian simulation under the null
hypothesis. Specifically, we generate two independent zero-mean
Gaussian samples with the observed sample sizes and a common covariance
matrix equal to the pooled empirical covariance of the TCGA data.
This design preserves the dimension, marginal scales, and estimated
dependence structure of the application while imposing covariance
equality. It therefore provides an application-specific assessment of
whether the observed rejection could be associated with substantial
finite-sample size distortion.

Let
\(\widehat{\bSigma}_{\mathrm{pool}}
=(n_1\widehat{\bSigma}_1+n_2\widehat{\bSigma}_2)/(n_1+n_2)\),
where the two covariance matrices are computed after separate
within-subtype centering. Eigenvalues below \(10^{-6}\) are raised to
that threshold before sampling. We use \(500\) Monte Carlo
replications and \(B=500\) draws for each bootstrap test. All
sample-based quantities, including the transformation, are re-estimated
in every replication and held fixed during the corresponding
multiplier bootstrap.

% \begin{table}[!h]
% \centering
% \captionsetup{font=small,skip=3pt}
% \caption{Empirical sizes in the TCGA-calibrated Gaussian simulation.
% The nominal significance level is \(5\%\). Results are based on
% \(R=500\) replications, with \(B=500\) draws for each bootstrap test.}
% \label{tab:tcga_realdata_merged}
% \scriptsize
% \setlength{\tabcolsep}{3pt}
% \renewcommand{\arraystretch}{0.5}
% \begin{tabular}{@{}lccccc@{}}
% \toprule
%  & \multicolumn{2}{c}{\(T_n\)} & \(T_{CZZW}\) & \(T_{WZZ}\) & \(T_{LC}\) \\
% \cmidrule(lr){2-3}
%  & \(\bA=\bI_p\) & \(\bA(0.5)\) & & & \\
% \midrule
% % Observed statistic & 4.239 & 7.232 & 11.980 & 1777.476 & 27.004 \\
% % \(p\)-value & \(0^{\dagger}\) & \(0^{\dagger}\) & \(0^{\dagger}\) & \(5.1\!\times\!10^{-13}\) & \(6.6\!\times\!10^{-161}\) \\
% % \addlinespace
% Size (\%) & 7.2 & 6.6 & 13.4 & 6.2 & 11.2 \\
% \bottomrule
% \end{tabular}
% \par\smallskip
% \begin{minipage}{\textwidth}
% \footnotesize
% % \(\dagger\) Monte Carlo estimate: no statistic among 2000 bootstrap draws
% % was at least as large as the observed statistic; this does not imply an
% % exact zero tail probability. Empirical sizes use \(R=500\), with
% % \(B=500\) for bootstrap tests.
% \end{minipage}
% \end{table}

Table~\ref{tab:tcga_realdata_size} reports the empirical sizes (in \%) from the TCGA-based Gaussian null simulation. In this
experiment, the transformation reduces the empirical size of \(T_n\)
from \(7.2\%\) to \(6.6\%\), bringing it closer to the nominal
\(5\%\) level. Both implementations of \(T_n\) show less size
inflation than \(T_{CZZW}\) and \(T_{LC}\), while the transformed
implementation performs similarly to \(T_{WZZ}\). In particular, the
equicorrelation-target transformation continues to perform well when
the underlying covariance structure is estimated from real data,
rather than specified by the equicorrelation or Toeplitz models used
in the simulation studies. Moreover,
\(\rho_{\mathrm{tar}}=0.5\) is carried over directly from the
simulation study and is not tuned to the TCGA data. These results
provide further evidence that the empirical performance of the
procedure is not highly sensitive to the precise target correlation
in the settings examined.
% The application therefore demonstrates
% detection of covariance differences in a relevant high-dimensional setting,
% with comparatively moderate size inflation. It complements the controlled
% weak-signal simulations, which assess when the combined transformation
% and studentization procedure improves power.
\begin{table}[!h]
\centering
\captionsetup{font=small,skip=3pt}
\caption{Empirical sizes (in \%) under the Gaussian null model based
on the pooled empirical covariance matrix of the TCGA data. }
% Results
% are based on \(R=500\) Monte Carlo replications and \(B=500\)
% multiplier draws for bootstrap procedures; the nominal level is
% \(5\%\).}
\label{tab:tcga_realdata_size}
\small
\setlength{\tabcolsep}{6pt}
\renewcommand{\arraystretch}{0.6}
\begin{tabular*}{\textwidth}{
@{\extracolsep{\fill}}ccccc@{}
}
\toprule
\(T_n\;(\bI_p)\)
&
\(T_n\;(\bA(0.5))\)
&
\(T_{CZZW}\)
&
\(T_{WZZ}\)
&
\(T_{LC}\)
\\
\midrule
7.2 & 6.6 & 13.4 & 6.2 & 11.2 \\
\bottomrule
\end{tabular*}
\end{table}
\section{Concluding Remarks}\label{sec:discussion}
We proposed a studentized \(\ell_2\)-type test for comparing two
high-dimensional covariance matrices. Before constructing the
statistic, the same nonsingular linear transformation is applied to
both samples. This transformation preserves covariance equality but
changes the standardized covariance differences and the correlations
among their estimators. The null distribution is approximated using a
Gaussian multiplier bootstrap that requires only coordinatewise
studentization and avoids explicitly forming or inverting the full
covariance matrix of the vectorized sample covariance entries. For a
fixed transformation, we established nonasymptotic Gaussian and
bootstrap approximation bounds, asymptotic size validity, and
sufficient conditions for power consistency. We also showed that,
under a relative operator-norm convergence condition, replacing a
population transformation by a same-sample plug-in estimator leaves
the test statistic and bootstrap critical value asymptotically
unchanged in relative terms. The simulation studies showed that the equicorrelation-target
transformations improved empirical size and power under weak
dependence, with similar results across the target correlations
considered. The combined transformation and studentization also
performed well under block-structured alternatives with high-variance
nuisance coordinates. An application to breast cancer gene-expression
data illustrated the proposed procedure in a high-dimensional
biological setting.

The equicorrelation target considered here is one practical choice and
is not expected to be uniformly optimal. Future work may study other
structured targets and develop theory for selecting a transformation
adaptively from multiple candidates while retaining valid bootstrap
inference.

% Equicorrelation-based transformations improve the untransformed test's
% calibration and power against small equicorrelation departures under weak
% dependence, with similar performance across targets and more modest gains
% under stronger dependence. Block simulations support the combined
% procedure when high-variance coordinates carry little or no signal.
% Both versions of \(T_n\) detect covariance differences between Basal-like
% and Luminal~A tumors in the selected cell-cycle panel; Gaussian resampling
% supports a strong empirical contrast.

% Future work includes other structured targets and data-adaptive choices
% of \(\bA\) with valid bootstrap calibration.

% \section*{Funding}
% Shu and Zhu's work was supported by the Ministry of Education, Singapore, under its Academic Research Fund Tier 1 (Project RG17/25).
\section*{Disclosure Statement}\label{disclosure-statement}

The authors report no conflicts of interest.

% \section{Data Availability Statement}\label{data-availability-statement}

% The simulation and real-data analysis code are available at
% \url{https://github.com/shuhaozhen1/covtestsim}.  The TCGA-BRCA expression and
% clinical annotations used in the application are publicly available from UCSC
% Xena, and the Hallmark gene-set definitions are publicly available from MSigDB.
% The real-data pipeline in the code repository downloads these public resources
% and reproduces the processed analysis files.

\appendix

\section*{Appendix}
For proof-specific positive semidefinite matrices
\(\bQ\in\mathbb R^{d\times d}\), we extend the spectral notation in
\eqref{eq:reference-eigen-geometry} by writing
\(\lambda_1(\bQ)\ge\cdots\ge\lambda_d(\bQ)\ge0\),
\(\Lambda_s(\bQ)=\{\sum_{\ell=s}^d\lambda_\ell^2(\bQ)\}^{1/2}\), and
\(\kappa(\bQ)=\{\Lambda_1(\bQ)\Lambda_2(\bQ)\}^{-1/2}\) when
\(\Lambda_2(\bQ)>0\).  We also use \(\|\cdot\|_{S_1}\) for the nuclear
(Schatten--1) norm.
\setcounter{equation}{0}
\global\long\def\theequation{A.\arabic{equation}}%
\renewcommand{\theHequation}{A.\arabic{equation}}
\setcounter{theorem}{0}
\renewcommand{\thetheorem}{A.\arabic{theorem}}
\renewcommand{\theHtheorem}{A.\arabic{theorem}}
\begin{lemma}
\label{lem:diagonal_studentizer}
Suppose that Assumptions~\ref{ass:coordinate_high_moment} and
\ref{ass:coordinate_nondegeneracy} hold.  Let 
\[
r_D
=
\max_{1\le k\le d}
\left|
\frac{\widehat D_{kk}}{\widetilde D_{kk}}-1
\right|.
\]
There exists a constant \(C_D>0\), depending only on \(K_8\) and \(c_H\),
such that, for every \(\delta\in(0,1)\),
\begin{equation}
\label{eq:studentizer-high-probability}
\mathbb P
\left\{
r_D>
\sqrt{\frac{C_Dd}{n_{\min}\delta}}
\right\}
\le\delta.
\end{equation}
Consequently,
\(r_D=O_{\mathbb P}\{(d/n_{\min})^{1/2}\}\). In particular, if
\(d/n_{\min}\to0\), then \(r_D=o_{\mathbb P}(1)\).
\end{lemma}
\begin{proof}[Proof of Lemma~\ref{lem:diagonal_studentizer}]
For \(i\in\{1,2\}\) and \(k\leftrightarrow(a,b)\), let
\(u_{ijk}^*\) and \(\bar u_{ik}^*\) denote the \(k\)-th entries of
\(\bu_{ij}^*\) and \(\bar{\bu}_i^*\), respectively. By the definitions of \(\widehat{\bD}\) and
\(\widetilde{\bD}\), we have
\[
\widehat D_{kk}-\widetilde D_{kk}
=\frac{s_{1,ab}^2-\Omega_{1,kk}}{n_1}
+
\frac{s_{2,ab}^2-\Omega_{2,kk}}{n_2}.
\]
% For \(k\leftrightarrow(a,b)\), let \(h_{ij,k}=h_{ij,ab}\), and define
% $\bh_{ij}
% =
% (h_{ij,1},\ldots,h_{ij,d})^\top
% =
% \operatorname{vech}\{
% (\bx_{ij}-\bar\bx_i)(\bx_{ij}-\bar\bx_i)^\top
% -
% \widehat{\bSigma}_{iA}\}$. 
Let $\bar\bz_i
=n_i^{-1}\sum_{j=1}^{n_i}\bz_{ij}$. Since $\bx_{ij}-\bar\bx_i=\bz_{ij}-\bar\bz_i$, we may write 
\begin{equation}\label{eq:h-a-d-decomposition}
\begin{split}
\bh_{ij}
&=\vech\{
(\bz_{ij}-\bar\bz_i)(\bz_{ij}-\bar\bz_i)^\top\}
-n_i^{-1}\sum_{\ell=1}^{n_i}
\vech\{
(\bz_{i\ell}-\bar\bz_i)(\bz_{i\ell}-\bar\bz_i)^\top\}\\
&=
\left(
\bu_{ij}-\bar\bu_i
\right)
-
\vech\left(
\bz_{ij}\bar\bz_i^\top
+
\bar\bz_i\bz_{ij}^\top
\right)
+
2\vech\left(\bar\bz_i\bar\bz_i^\top\right)=\bg_{ij}+\br_{ij},
\end{split}
\end{equation}
where
$\bg_{ij}= \bu_{ij}-\bar\bu_i=\bu^*_{ij}-\bar\bu_i^*$, 
and $\br_{ij}=-\vech(\bz_{ij}\bar\bz_i^\top+\bar\bz_i\bz_{ij}^\top)
+2\vech(\bar\bz_i\bar\bz_i^\top)$.
Using the preceding decomposition of the $k$-th entry of $\bh_{ij}$, we have
\begin{align*}
s_{i,ab}^2-\Omega_{i,kk}
=\left(
\frac1{n_i}\sum_{j=1}^{n_i}g_{ij,k}^2
-
\Omega_{i,kk}
\right) +
\frac2{n_i}\sum_{j=1}^{n_i}g_{ij,k}r_{ij,k}
+
\frac1{n_i}\sum_{j=1}^{n_i}r_{ij,k}^2,
\end{align*}
where $g_{ij,k}$ and $r_{ij,k}$ are the $k$-th entries of $\bg_{ij}$ and $\br_{ij}$, respectively.

By Assumption~\ref{ass:coordinate_high_moment} and the
Cauchy--Schwarz inequality, we have $\mathbb E|u_{i1k}^*|^4
\le C_u$, where $C_u=16K_8$. Moreover, since the variables \(z_{i1a},\ldots,z_{in_i a}\) are
independent and centered, we have
\[
\mathbb E\left|
\sum_{j=1}^{n_i}z_{ija}
\right|^8
\le
C\left\{
n_i\mathbb E|z_{i1a}|^8
+
n_i^4[\mathbb E(z_{i1a}^2)]^4
\right\}
\le
C_z n_i^4,
\]
where \(C_z\) depends only on \(K_8\). Consequently,
\begin{equation}\label{Ezia8.equ}
\max_{1\le a\le p}
\mathbb E|\bar z_{ia}|^8
\le
\frac{C_{z}}{n_i^4}.
\end{equation}
Therefore, using
\((x+y+z)^4\le27(x^4+y^4+z^4)\) and the Cauchy--Schwarz inequality,
\begin{equation}
\label{eq:r-coordinate-fourth-moment}
\max_{1\le j\le n_i}
\max_{1\le k\le d}
\mathbb E|r_{ij,k}|^4
\le
\frac{C_r}{n_i^2},
\end{equation}
where \(C_r\) depends only on \(K_8\).

We first observe that $n_i^{-1}\sum_{j=1}^{n_i}g_{ij,k}^2-\Omega_{i,kk}
=n_i^{-1}\sum_{j=1}^{n_i}
(u_{ijk}^{*2}-\Omega_{i,kk})-\bar u_{ik}^{*2}$.
Independence and centering imply 
\[
\mathbb E
\left[\frac{1}{n_i}\sum_{j=1}^{n_i}
(u_{ijk}^{*2}-\Omega_{i,kk})\right]^2
\le \frac{C_u}{n_i},\quad \mbox{ and }\quad \mathbb E|\bar u_{ik}^*|^4
\le \frac{3C_u}{n_i^2}.
\]
Consequently, $\mathbb E(n_i^{-1}\sum_{j=1}^{n_i}g_{ij,k}^2
-\Omega_{i,kk})^2
\le 8C_u/n_i$. Furthermore, sample centering and Jensen's inequality give
$\mathbb E(n_i^{-1}\sum_{j=1}^{n_i}g_{ij,k}^2)^2
\le C_u$, and $\mathbb E(n_i^{-1}\sum_{j=1}^{n_i}r_{ij,k}^2)^2
\le C_r/n_i^2$. By the Cauchy--Schwarz inequality,
\begin{align*}
\left(
\frac2{n_i}\sum_{j=1}^{n_i}g_{ij,k}r_{ij,k}
+
\frac1{n_i}\sum_{j=1}^{n_i}r_{ij,k}^2
\right)^2\le
8
\left(
\frac1{n_i}\sum_{j=1}^{n_i}g_{ij,k}^2
\right)
\left(
\frac1{n_i}\sum_{j=1}^{n_i}r_{ij,k}^2
\right)
+
2
\left(
\frac1{n_i}\sum_{j=1}^{n_i}r_{ij,k}^2
\right)^2.
\end{align*}
Therefore,
\[
\mathbb E
\left(
\frac2{n_i}\sum_{j=1}^{n_i}g_{ij,k}r_{ij,k}
+
\frac1{n_i}\sum_{j=1}^{n_i}r_{ij,k}^2
\right)^2
\le
\frac{C_{r,u}}{n_i},
\]
where $C_{r,u}$ depends on $K_8$ only. Combining the preceding bounds yields
\[
\max_{1\le k\le d}
\mathbb E
\left(
s_{i,ab}^2-\Omega_{i,kk}
\right)^2
\le
\frac{C_s}{n_i},
\qquad i=1,2,
\]
where \(k\leftrightarrow(a,b)\) and
\(C_s\) depends on $K_8$ only. Hence
\[
\max_{1\le k\le d}
\mathbb E
(\widehat D_{kk}-\widetilde D_{kk})^2
\le
2C_s
\left(
\frac1{n_1^3}+\frac1{n_2^3}
\right)
\le
\frac{4C_s}{n_{\min}^3}.
\]
Since \(\widetilde D_{kk}\ge c_H/n_{\min}\), setting
\(C_D=4C_s/c_H^2\) gives
\[
\max_{1\le k\le d}
\mathbb E
\left(
\frac{\widehat D_{kk}}{\widetilde D_{kk}}-1
\right)^2
\le
\frac{C_D}{n_{\min}}.
\]
For every \(t>0\), the union bound and Markov's inequality give
\[
\mathbb P(r_D>t)
\le
\sum_{k=1}^d
\mathbb P
\left\{
\left|
\frac{\widehat D_{kk}}{\widetilde D_{kk}}-1
\right|>t
\right\}
\le
\frac{C_Dd}{n_{\min}t^2}.
\]
Taking
\(t=\{C_Dd/(n_{\min}\delta)\}^{1/2}\) gives \eqref{eq:studentizer-high-probability}. 
\end{proof}

\begin{lemma}
\label{lem:bootstrap-correlation-consistency}
Let \(\bR\) and \(\widehat{\bR}\) be, respectively, the oracle
correlation matrix and its empirical bootstrap counterpart defined in
\eqref{R.equ} and \eqref{Rhat.equ}.
Suppose Assumptions~\ref{ass:coordinate_high_moment} and
\ref{ass:coordinate_nondegeneracy} hold.  There exists a constant
\(C_R>0\), depending only on \(K_8\) and \(c_H\), such
that, for every \(\eta_2\in(0,1)\), if
\[
C_R\sqrt{\frac{d}{n_{\min}\eta_2}}
\le
\frac12,
\]
then, with probability at least \(1-\eta_2\), every diagonal entry of
\(\widehat{\bD}\) is positive and
\begin{equation}
\label{eq:R-nuclear-bound}
\|\widehat{\bR}-\bR\|_{S_1}
\le
C_R\sqrt{\frac{d^3}{n_{\min}\eta_2}}.
\end{equation}
\end{lemma}
\begin{proof}[Proof of Lemma~\ref{lem:bootstrap-correlation-consistency}]
    We first control the covariance estimator before diagonal standardization.  
Define
\[
\widehat\bOmega_i^*
=
\frac1{n_i}
\sum_{j=1}^{n_i}(\bu_{ij}^*-\bar\bu_i^*)(\bu_{ij}^*-\bar\bu_i^*)^\top
=\frac1{n_i}\sum_{j=1}^{n_i}\bu_{ij}^*\bu_{ij}^{*\top}-\bar\bu_i^*\bar\bu_i^{*\top}.
\]

It follows that
\begin{equation}\label{eq:covariance-estimator-decomposition}
    \|\widehat\bOmega_i^*-\bOmega_i\|_{\mathrm{F}}^2
\leq
2\Big\|n_i^{-1}\sum_{j=1}^{n_i}\bu_{ij}^*\bu_{ij}^{*\top}-\bOmega_i\Big\|_{\mathrm{F}}^2
+2\|\bar\bu_i^*\|^4.
\end{equation}
Since the random matrices $\bu_{ij}^*\bu_{ij}^{*\top}-\bOmega_i$, $j=1,\ldots,n_i$, are independent and have mean zero, 
\begin{equation}\label{avg-bound.equ}
\mathbb E\Big\|n_i^{-1}\sum_{j=1}^{n_i}\bu_{ij}^*\bu_{ij}^{*\top}-\bOmega_i\Big\|_{\mathrm{F}}^2
=n_i^{-1}\,\mathbb E\big\|\bu_{i1}^*\bu_{i1}^{*\top}-\bOmega_i\big\|_{\mathrm{F}}^2.
\end{equation}

For \(k\leftrightarrow(a,b)\), Assumption~\ref{ass:coordinate_high_moment}
and the standard shift inequality give
\begin{equation}\label{Eu4.equ}
    \mathbb E|u_{i1k}^*|^4\le16K_8.
\end{equation}
Hence
\(\mathbb E\|\bu_{i1}^*\|^4
\le d\sum_{k=1}^d\mathbb E|u_{i1k}^*|^4
\le16K_8d^2\).  Using the variance identity in the Frobenius space,
\[
\mathbb E
\left\|
\bu_{i1}^*\bu_{i1}^{*\top}-\bOmega_i
\right\|_{\mathrm{F}}^2
=
\mathbb E\|\bu_{i1}^*\|^4-\|\bOmega_i\|_{\mathrm{F}}^2
\le16K_8d^2.
\]
Moreover, independence and \(\mathbb E(u_{ijk}^*)=0\) give
\[
\mathbb E|\bar u_{ik}^*|^4
=
n_i^{-3}\mathbb E|u_{i1k}^*|^4
+
3n_i^{-3}(n_i-1)
\bigl(\mathbb E|u_{i1k}^*|^2\bigr)^2
\le
\frac{48K_8}{n_i^2}.
\]
Consequently,
\begin{equation}
\label{eq:sample-mean-fourth-bound}
\mathbb E\|\bar\bu_i^*\|^4
\le
d\sum_{k=1}^d\mathbb E|\bar u_{ik}^*|^4
\le
\frac{48K_8d^2}{n_i^2}.
\end{equation}
Taking expectations in
\eqref{eq:covariance-estimator-decomposition} and using
\eqref{avg-bound.equ} and
\eqref{eq:sample-mean-fourth-bound}, we obtain
\begin{equation}
\label{eq:oracle-covariance-mse}
\mathbb E
\|\widehat{\bOmega}_i^*-\bOmega_i\|_{\mathrm{F}}^2
\le
\frac{128K_8d^2}{n_i}.
\end{equation}
We next compare the feasible residuals with the oracle centered residuals. 
Using the decomposition of $\bh_{ij}$ in \eqref{eq:h-a-d-decomposition}, we may write
\(\widehat{\bOmega}_i^*
=n_i^{-1}\sum_{j=1}^{n_i}\bg_{ij}\bg_{ij}^\top\), and therefore
\begin{equation*}
\label{eq:Omega-hat-difference-decomposition}
\widehat{\bOmega}_i-\widehat{\bOmega}_i^*
=
\frac1{n_i}\sum_{j=1}^{n_i}
\left(
\bg_{ij}\br_{ij}^\top
+
\br_{ij}\bg_{ij}^\top
+
\br_{ij}\br_{ij}^\top
\right).
\end{equation*}
Define $L_i
=n_i^{-1}\sum_{j=1}^{n_i}\|\bg_{ij}\|^2$ and $M_i
=n_i^{-1}\sum_{j=1}^{n_i}\|\br_{ij}\|^2$. By the triangle inequality and the Cauchy--Schwarz inequality, we have
% \begin{align}
% \|\widehat{\bOmega}_i-\widehat{\bOmega}_i^*\|_{\mathrm F}
% &\leq
% \frac{2}{n_i}\sum_{j=1}^{n_i}
% \|\bg_{ij}\|\|\br_{ij}\|
% +
% \frac{1}{n_i}\sum_{j=1}^{n_i}\|\br_{ij}\|^2
% \nonumber\\
% &\leq
% 2\left\{
% \frac{1}{n_i}\sum_{j=1}^{n_i}\|\bg_{ij}\|^2
% \right\}^{1/2}
% \left\{
% \frac{1}{n_i}\sum_{j=1}^{n_i}\|\br_{ij}\|^2
% \right\}^{1/2}
% +
% \frac{1}{n_i}\sum_{j=1}^{n_i}\|\br_{ij}\|^2.
% \label{eq:Omega-hat-difference-bound}
% \end{align}
$\|\widehat{\bOmega}_i-\widehat{\bOmega}_i^*\|_{\mathrm{F}}
\leq
2L_i^{1/2}M_i^{1/2}+M_i$.
Consequently,
\begin{equation}
\label{eq:Omega-hat-difference-square-bound}
\|\widehat{\bOmega}_i-\widehat{\bOmega}_i^*\|_{\mathrm{F}}^2
\leq
8L_iM_i+2M_i^2.
\end{equation}
We now derive second-moment bounds for \(L_i\) and \(M_i\). Since the
sample mean minimizes the sum of squared Euclidean distances, $L_i
\leq n_i^{-1}\sum_{j=1}^{n_i}\|\bu_{ij}^*\|^2$. Jensen's inequality therefore gives
\begin{equation}
\label{eq:L-second-moment-bound}
\mathbb E (L_i^2)
\leq
\mathbb E\left(
\frac1{n_i}\sum_{j=1}^{n_i}\|\bu_{ij}^*\|^2
\right)^2
\leq
\mathbb E\|\bu_{i1}^*\|^4
\leq 16K_8d^2.
\end{equation}

To control \(M_i\), recall from
\eqref{eq:r-coordinate-fourth-moment} that $\max_{1\le j\le n_i}
\max_{1\le k\le d}
\mathbb E|r_{ij,k}|^4
\le C_r/n_i^2$, where \(C_r\) depends only on \(K_8\).  Since $\|\br_{ij}\|^4
=(
\sum_{k=1}^d r_{ij,k}^2)^2\le
d\sum_{k=1}^d r_{ij,k}^4$, we have $\mathbb E\|\br_{ij}\|^4
\le C_rd^2/n_i^2$. Jensen's inequality therefore gives
\begin{equation}
\label{eq:M-second-moment-bound}
\mathbb E(M_i^2)
\le
\frac1{n_i}\sum_{j=1}^{n_i}
\mathbb E\|\br_{ij}\|^4
\le
\frac{C_rd^2}{n_i^2}.
\end{equation}

Taking expectations in
\eqref{eq:Omega-hat-difference-square-bound} and using
\eqref{eq:L-second-moment-bound},
\eqref{eq:M-second-moment-bound}, and the Cauchy--Schwarz inequality, we obtain
\begin{align}\label{eq:feasible-oracle-covariance-mse}
\mathbb E
\|\widehat{\bOmega}_i-\widehat{\bOmega}_i^*\|_{\mathrm{F}}^2
&\leq
8\{\mathbb E (L_i^2)\}^{1/2}\{\mathbb E (M_i^2)\}^{1/2}
+2\mathbb E (M_i^2) \leq
\left(
32\sqrt{K_8C_r}+2C_r
\right)
\frac{d^2}{n_i}.
\end{align}
Using
\(\|\bA+\bB\|_{\mathrm{F}}^2
\le2\|\bA\|_{\mathrm{F}}^2+2\|\bB\|_{\mathrm{F}}^2\),
together with \eqref{eq:oracle-covariance-mse} and
\eqref{eq:feasible-oracle-covariance-mse}, gives $\mathbb E
\|\widehat{\bOmega}_i-\bOmega_i\|_{\mathrm{F}}^2
\le
C_0d^2/n_i$, where $C_0$ depends on $K_8$ only. It follows that
\begin{equation}\label{eq:Omega-i-mse}
\begin{split}
    \mathbb E
\|\widehat{\bOmega}_n-\bOmega_n\|_{\mathrm{F}}^2
&\leq
2
\left\{
\frac{
\mathbb E\|\widehat{\bOmega}_1-\bOmega_1\|_{\mathrm{F}}^2
}{n_1^2}
+
\frac{
\mathbb E\|\widehat{\bOmega}_2-\bOmega_2\|_{\mathrm{F}}^2
}{n_2^2}
\right\} \\
&\leq
2C_0d^2
\left(
n_1^{-3}+n_2^{-3}
\right) \leq
4C_0\frac{d^2}{n_{\min}^3},
\end{split}
\end{equation}
where the last inequality uses \(n_i\geq n_{\min}\).

We now pass from $\widehat{\bOmega}_n-\bOmega_n$ to
$\widehat{\bR}-\bR$. By Lemma~\ref{lem:diagonal_studentizer}, applied with
\(\delta=\eta_2/2\),
\begin{equation}
\label{eq:diagonal-high-probability}
\mathbb P
\left\{
r_D>\sqrt{\frac{2C_Dd}{n_{\min}\eta_2}}
\right\}
\le
\frac{\eta_2}{2}.
\end{equation}
 By Markov's inequality and \eqref{eq:Omega-i-mse},
\begin{equation}
\label{eq:bootstrap-Omega-n-high-probability}
\mathbb P
\left\{
\|\widehat{\bOmega}_n-\bOmega_n\|_{\mathrm{F}}
>\sqrt{\frac{8C_0}{n^3_{\min}\eta_2}}\,d
\right\}
\le
\frac{\eta_2}{2}.
\end{equation}
Combining \eqref{eq:diagonal-high-probability} and
\eqref{eq:bootstrap-Omega-n-high-probability}, we obtain an event
\(\mathcal E_{\eta_2}\) with probability at least \(1-\eta_2\) on
which
\begin{equation*}
\label{eq:joint-covariance-diagonal-event}
\|\widehat{\bOmega}_n-\bOmega_n\|_{\mathrm{F}}
\le \sqrt{\frac{8C_0}{n^3_{\min}\eta_2}}\,d,
\quad\text{ and }\quad
r_D
\le \sqrt{\frac{2C_Dd}{n_{\min}\eta_2}}.
\end{equation*}
Choose $C_R
=4\sqrt{2C_0}/c_H+
(2+\sqrt2)\sqrt{2C_D}.$ Since \(C_R\ge\sqrt{2C_D}\), the condition of the lemma implies that, on
\(\mathcal E_{\eta_2}\),
\[
r_D\le\sqrt{2C_D}
\sqrt{\frac{d}{n_{\min}\eta_2}}
\le
C_R\sqrt{\frac{d}{n_{\min}\eta_2}}
\le
\frac12.
\]

In particular, for every \(k=1,\ldots,d\), $\widehat D_{kk}
\ge
(1-r_D)\widetilde D_{kk}
\ge\widetilde D_{kk}/2
\ge c_H/(2n_{\min})
>
0$. Thus every diagonal entry of \(\widehat{\bD}\) is positive on
\(\mathcal E_{\eta_2}\).
Define $\bB=
\widehat{\bD}^{-1/2}\widetilde{\bD}^{1/2}$. 
Since
\(\widetilde{\bD}^{1/2}\bR\widetilde{\bD}^{1/2}
=\bOmega_n\), we have the exact decomposition
\begin{align}
\widehat{\bR}-\bR
&=
\widehat{\bD}^{-1/2}
(\widehat{\bOmega}_n-\bOmega_n)
\widehat{\bD}^{-1/2}
+
\bB\bR\bB-\bR.
\label{eq:R-Rhat-decomposition}
\end{align}

Since \(r_D\le1/2\) and
\(\min_k\widetilde D_{kk}\ge c_H/n_{\min}\), we have 
$\|\widehat{\bD}^{-1/2}\|_{\mathrm{op}}^2
\le 2n_{\min}/c_H$, $\|\bB\|_{\mathrm{op}}
\le
\sqrt{2}$, and $\|\bB-\bI_d\|_{\mathrm{op}}
\le
\sqrt{2}\,r_D$. Using
\(\|\bA\|_{S_1}\le\sqrt d\,\|\bA\|_{\mathrm{F}}\), we obtain
\begin{align*}
&
\left\|
\widehat{\bD}^{-1/2}
(\widehat{\bOmega}_n-\bOmega_n)
\widehat{\bD}^{-1/2}
\right\|_{S_1}\le
\sqrt d\,
\|\widehat{\bD}^{-1/2}\|_{\mathrm{op}}^2
\|\widehat{\bOmega}_n-\bOmega_n\|_{\mathrm{F}}\le\frac{4\sqrt{2C_0}}{c_H}
\sqrt{\frac{d^3}{n_{\min}\eta_2}}.
\end{align*}
For the standardization error, write $\bB\bR\bB-\bR
=
(\bB-\bI_d)\bR\bB
+
\bR(\bB-\bI_d)$. Since \(\bR\) is positive semidefinite with unit diagonal, $\|\bR\|_{S_1}
=
\operatorname{tr}(\bR)
=
d$. It follows that
\begin{align*}
\|\bB\bR\bB-\bR\|_{S_1}
&\le
\|\bB-\bI_d\|_{\mathrm{op}}
\|\bR\|_{S_1}
\|\bB\|_{\mathrm{op}}
+
\|\bR\|_{S_1}
\|\bB-\bI_d\|_{\mathrm{op}}
\\
&\le
(2+\sqrt{2})d\,r_D\le 
(2+\sqrt2)\sqrt{2C_D}
\sqrt{\frac{d^3}{n_{\min}\eta_2}}.
\end{align*}
Therefore, by \eqref{eq:R-Rhat-decomposition} and the triangle
inequality, we have
\[
\|\widehat{\bR}-\bR\|_{S_1}
\le
C_R
\sqrt{\frac{d^3}{n_{\min}\eta_2}}.
\]
This completes the proof.
\end{proof}

\begin{proof}[Proof of Theorem~\ref{thm:oracle_gaussian}]
Under \(H_0\),
\(\mathbb E(\bu_{11})=\mathbb E(\bu_{21})\), and hence $\bar\bu_1-\bar\bu_2
=
\bar\bu_1^*-\bar\bu_2^*$. Therefore, by \eqref{Tn0.equ}, we can write
\begin{equation}
\label{tTnxi.equ}
\widetilde T_n
=
\widetilde T_{n,0}
=
\frac1d\|\bxi\|^2,
\end{equation}
where \(\bxi=\sum_{j=1}^{n}\bxi_j\), \(n=n_1+n_2\), and
\(\bxi_1,\ldots,\bxi_n\) are centered independent random vectors in
\(\mathbb R^d\) defined by
\begin{equation}
\label{bxi.equ}
\bxi_j
=
\begin{cases}
\displaystyle
\frac{1}{n_1}
\widetilde{\bD}^{-1/2}\bu_{1j}^*,
& j=1,\ldots,n_1, \\[2mm]
\displaystyle
-\frac{1}{n_2}
\widetilde{\bD}^{-1/2}\bu_{2,j-n_1}^*,
& j=n_1+1,\ldots,n.
\end{cases}
\end{equation}
By the independence of the summands,
\[
\operatorname{Cov}(\bxi)
=
\sum_{j=1}^{n}\operatorname{Cov}(\bxi_j)
=
\sum_{j=1}^{n_1}\frac{1}{n_1^2}
\widetilde \bD^{-1/2}\bOmega_1\widetilde \bD^{-1/2}
+
\sum_{j=1}^{n_2}\frac{1}{n_2^2}
\widetilde \bD^{-1/2}\bOmega_2\widetilde \bD^{-1/2}=\bR,
\]
where $\bR$ is defined in \eqref{R.equ}.

To apply Theorem~3.2 of \citet{fang_large-dimensional_2024}, we first
verify that each \(\bxi_j\) has a finite fourth moment. Let
\(\bxi_j=(\xi_{j1},\ldots,\xi_{jd})^\top\). For
\(j=1,\ldots,n_1\), \eqref{bxi.equ} gives $\xi_{jk}
=n_1^{-1}
\widetilde D_{kk}^{-1/2}u_{1jk}^*, k=1,\ldots,d$. Since $\widetilde D_{kk}\geq \Var(u_{11k}^*)/n_1$, we have
\[
\E|\xi_{jk}|^4=\frac{1}{n_1^4}\frac{\E|u_{1jk}^{*}|^4}{\widetilde D_{kk}^{2}}\leq \frac{1}{n_1^2}\frac{\E|u_{1jk}^{*}|^4}{\Var^2(u_{11k}^*)}.
\]

% Suppose that the $k$th coordinate corresponds to the index
% pair $(a,b)$, so that $u_{1jk}=z_{1ja}z_{1jb}$. 
% For \(k\leftrightarrow(a,b)\), we have
% \(u_{1jk}=z_{1ja}z_{1jb}\). By the Cauchy--Schwarz inequality and Assumption~\ref{ass:coordinate_high_moment}, we have $\E|u_{1jk}^*|^4\le 16K_8$.
% $ \E|u_{1jk}|^4 = \E(|z_{1ja}|^4|z_{1jb}|^4)\leq
% \{
% \mathbb E|z_{1ja}|^8
% \mathbb E|z_{1jb}|^8\}^{1/2}
% \leq K_8$. Furthermore, Jensen's inequality gives
% $|\E(u_{1jk})|^4 \le \E|u_{1jk}|^4 \le K_8$.  Since $u_{1jk}^*=u_{1jk}-\E(u_{1jk})$, the inequality $|x-y|^4\leq 8(|x|^4+|y|^4)$ yields
% \[
% \E|u_{1jk}^*|^4
% \le 8\left\{
% \mathbb E|u_{1jk}|^4
% +
% |\mathbb E(u_{1jk})|^4
% \right\}
% \le 16K_8.
% \]
Under Assumption~\ref{ass:coordinate_nondegeneracy}, by (\ref{Eu4.equ}), we have $\max_{1\leq k\leq d}\mathbb E|\xi_{jk}|^4
\leq n_1^{-2}16K_8/c_H^2$, $j=1,\ldots,n_1$. By the same argument, $\max_{1\leq k\leq d}\E|\xi_{jk}|^4\leq n_2^{-2}16K_8/c_H^2$, $j=n_1+1,\ldots,n$. Hence,
\begin{equation}
\label{eq:xi-coordinate-fourth-sum}
\max_{1\le k\le d}
\sum_{j=1}^n
\mathbb E|\xi_{jk}|^4
\le
\frac{16K_8}{c_H^2}
\left(
\frac1{n_1}+\frac1{n_2}
\right).
\end{equation}
By the Cauchy--Schwarz inequality, we have
\[
\E\|\bxi_j\|^4=\E\left(\sum_{k=1}^d\xi_{jk}^2\right)^2\leq d\sum_{k=1}^d\E|\xi_{jk}|^4\leq d^2\max_{1\le k\le d}\E|\xi_{jk}|^4\leq
\frac{16K_8}{c_H^2}\frac{d^2}{n_{\min}^2}.
\]

We now apply Theorem~3.2 of
\citet{fang_large-dimensional_2024}. The preceding bounds show that
\(\mathbb E\|\bxi_j\|^4<\infty\) for every \(j\). Moreover,
\(\bxi_1,\ldots,\bxi_n\) are centered and independent, and
Assumption~\ref{ass:gaussian_reference_geometry} gives
\(\Lambda_2(\bR)>0\). Using the notation of Theorem~3.2 of
\citet{fang_large-dimensional_2024}, take
\(W=\bxi\), \(\bSigma_W=\operatorname{Cov}(\bxi)=\bR\), and
\(\bSigma=\bR\). Then the covariance-discrepancy terms
\(\delta_0(\bSigma_W,\bSigma)\) and
\(\delta_0'(\bSigma_W,\bSigma)\) are both zero.
Therefore, Theorem~3.2 of
\citet{fang_large-dimensional_2024} gives
\[
    \sup_{r\geq 0}|\mathbb P(\|\bxi\|\leq r)-\mathbb P(\|\bZ\|\leq r)|\leq C\{\kappa^{3/4}(\bR)\Delta_1^{1/4}+\kappa^{2/3}(\bR)\Delta_2^{1/3}\},
\]
where $C$ is an absolute constant, and 
\[
\Delta_1=\|\bR\|_{\mathrm{F}}
\sum_{j=1}^{n}\mathbb E\|\bxi_j\|^4
+
\|\bR\|_{\mathrm{op}}^{3/2}
\sum_{j=1}^{n}\mathbb E\|\bxi_j\|^3,\;\text{ and }\;\Delta_2=
\|\bR\|_{\mathrm{op}}^{1/2}
\sum_{j=1}^{n}\mathbb E\|\bxi_j\|^3
+
\sum_{j=1}^{n}\mathbb E\|\bxi_j\|^4.
\]

We next bound the moment sums appearing in $\Delta_1$ and $\Delta_2$. For any \(\bx\in\mathbb R^d\), $\|\bx\|^3=
(\sum_{k=1}^d x_k^2)^{3/2}
\le
d^{1/2}\sum_{k=1}^d |x_k|^3$, and $\|\bx\|^4=(\sum_{k=1}^d x_k^2)^2
\le
d\sum_{k=1}^d |x_k|^4$. Moreover, Lyapunov's inequality gives
$\E|\xi_{jk}|^3
\leq \{\mathbb E|\xi_{jk}|^4\}^{3/4}$. Therefore, we have
\[
\sum_{j=1}^{n}
\mathbb E\|\bxi_j\|^3
\le
d^{1/2}
\sum_{j=1}^{n}
\sum_{k=1}^d
\mathbb E|\xi_{jk}|^3\le d^{1/2}
\sum_{j=1}^{n}
\sum_{k=1}^d\big(\E|\xi_{jk}|^4\big)^{3/4}
\le d^{3/2}\left(\frac{16K_8}{c_H^2}\right)^{3/4}\left(\frac{1}{\sqrt{n_1}}+\frac{1}{\sqrt{n_2}}\right),
\]
and
\begin{equation}
\label{eq:xi-vector-fourth-sum}
\sum_{j=1}^{n}
\mathbb E\|\bxi_j\|^4\le
d
\sum_{j=1}^{n}
\sum_{k=1}^d
\mathbb E|\xi_{jk}|^4
\leq
\frac{16K_8d^2}{c_H^2}
\left(\frac{1}{n_1}+\frac{1}{n_2}\right).
\end{equation}

Since \(n_i\ge n_{\min}\), there exists a constant $K_0>0$, depending only
on $K_8$ and $c_H$, such that
\[
 \Delta_1
\leq
K_0\left(
\|\bR\|_{\mathrm{F}}d^2n_{\min}^{-1}
+
\|\bR\|_{\mathrm{op}}^{3/2}d^{3/2}n_{\min}^{-1/2}
\right),\;\text{ and }
\Delta_2
\leq
K_0\left(
\|\bR\|_{\mathrm{op}}^{1/2}d^{3/2}n_{\min}^{-1/2}
+
d^2n_{\min}^{-1}
\right).   
\]

As defined in \eqref{eq:reference-eigen-geometry}, we have $\|\bR\|_{\mathrm{F}}=\Lambda_1$ and
$\|\bR\|_{\mathrm{op}}=\lambda_1$. %and $\kappa(\bR)=(\Lambda_1\Lambda_2)^{-1/2}$.
It follows that $\Delta_1\leq K_0\delta_1$ and $\Delta_2\leq K_0\delta_2$, where $\delta_1$ and $\delta_2$ are given in \eqref{delta1.equ}. Therefore, we have
\[
    \sup_{r\geq 0}|\mathbb P(\|\bxi\|\leq r)-\mathbb P(\|\bZ\|\leq r)|\leq C_1\{\kappa^{3/4}(\bR)\delta_1^{1/4}+\kappa^{2/3}(\bR)\delta_2^{1/3}\},
\]
where $C_1=C\max\{K_0^{1/4},K_0^{1/3}\}$. 

Finally, let $\bZ\sim N(\mathbf{0},\bR)$ and recall that $T_Z=d^{-1}\|\bZ\|^2$ and $\widetilde T_{n,0}
=d^{-1}\|\bxi\|^2$. Hence, for every $r\geq0$,
\[
\left|
\mathbb P(\widetilde T_n\leq r)
-
\mathbb P(T_Z\leq r)
\right|
=
\left|
\mathbb P\bigl(\|\bxi\|\leq\sqrt{dr}\bigr)
-
\mathbb P\bigl(\|\bZ\|\leq\sqrt{dr}\bigr)
\right|.
\]
Since the mapping $r\mapsto\sqrt{dr}$ maps $[0,\infty)$ onto
$[0,\infty)$, we have 
\[
\begin{split}
   \sup_{r\geq0}
\left|
\mathbb P(\widetilde T_n\leq r)
-
\mathbb P(T_Z\leq r)
\right|&=
\sup_{t\geq0}
\left|
\mathbb P(\|\bxi\|\leq t)-\mathbb P(\|\bZ\|\leq t)
\right|\\
&\leq C_1\{\kappa^{3/4}(\bR)\delta_1^{1/4}+\kappa^{2/3}(\bR)\delta_2^{1/3}\}.
\end{split}
\]
This completes the proof.
\end{proof}
\begin{proof}[Proof of Theorem~\ref{thm:sample_oracle_transfer}]
Define $\ba=d^{-1/2}\tbD^{-1/2}\tbw$, $\bb=d^{-1/2}\tbD^{-1/2}(\hbw-\tbw)$, and $\bH=\widetilde \bD^{1/2}\widehat \bD^{-1}\widetilde \bD^{1/2}$.
Then by \eqref{eq:proposed-statistic} and \eqref{eq:oracle-statistic}, the test statistic $T_n$ and the oracle test statistic $\widetilde T_n$ can be written as $T_n=(\ba+\bb)^\top\bH(\ba+\bb)$ and $\widetilde T_n=\ba^\top\ba$, respectively.

Since \(\bH\) is diagonal, its \(k\)-th diagonal entry is $H_{kk}=\widetilde D_{kk}/\widehat D_{kk}$.
On the event \(\{r_D<1\}\), we have $\widehat D_{kk}/\widetilde D_{kk}\ge 1-r_D>0$, and hence
\[
|H_{kk}-1|
=
\frac{
\left|
\widehat D_{kk}/\widetilde D_{kk}-1
\right|
}{
\widehat D_{kk}/\widetilde D_{kk}
}
\le
\frac{r_D}{1-r_D}.
\]
In particular, on \(\{r_D\le1/2\}\), $\|\bH-\bI_d\|_{\mathrm{op}}
\le 2r_D$. It follows that $|
(\ba+\bb)^\top(\bH-\bI_d)(\ba+\bb)|
\le
2r_D(\|\ba\|+\|\bb\|)^2$. Moreover, $|\|\ba+\bb\|^2-\|\ba\|^2|=|2\ba^\top\bb+\|\bb\|^2|
\le 2\|\ba\|\|\bb\|+\|\bb\|^2$. Thus, on \(\{r_D\le1/2\}\),
\begin{equation}
\label{eq:supp-transfer-perturbation}
|T_n-\widetilde T_n|
\le
2r_D(\|\ba\|+\|\bb\|)^2
+
2\|\ba\|\|\bb\|
+
\|\bb\|^2.
\end{equation}
We next obtain high-probability bounds for
\(\|\ba\|\), \(\|\bb\|\), and \(r_D\). By \eqref{tTnxi.equ}, $\widetilde T_n
=\|\ba\|^2=d^{-1}\|\bxi\|^2$, with $\bxi=\sum_{j=1}^n\bxi_j$. Write \(\xi_k=\sum_{j=1}^n\xi_{jk}\). Expanding the fourth power and using
the independence and centering of the summands gives $\E(\xi_k^4)
=
\sum_{j=1}^n\E(\xi_{jk}^4)
+
3\sum_{j\ne\ell}
\E(\xi_{jk}^2)\E(\xi_{\ell k}^2)$. Since \(\Cov(\bxi)=\bR\) has unit diagonal, $\sum_{j=1}^n\E(\xi_{jk}^2)=1$. Moreover, by  \eqref{eq:xi-coordinate-fourth-sum}, there exists a constant \(C_T<\infty\) such that $\max_{1\le k\le d}\E|\xi_k|^4\le C_T$. By the Cauchy--Schwarz inequality,
\[
\E(\widetilde T_n^2)
=\frac1{d^2}
\sum_{k=1}^d\sum_{\ell=1}^d
\E(\xi_k^2\xi_\ell^2) \le
\frac1{d^2}
\sum_{k=1}^d\sum_{\ell=1}^d
\big(\E|\xi_k|^4\big)^{1/2}
\big(\E|\xi_\ell|^4\big)^{1/2}
\le C_T.
\]
Consequently, Markov's inequality yields
\[
\mathbb P\left\{
\|\ba\|>
\left(\frac{3C_T}{\eta}\right)^{1/4}
\right\}
=
\mathbb P\left\{
\widetilde T_n^2>\frac{3C_T}{\eta}
\right\}
\le
\frac{\eta}{3}.
\]

Next, by \eqref{hbw.equ} and \eqref{tbw.equ}, we have
\[
\hbw-\tbw=
\operatorname{vech}
\left\{
(\hbSigma_{1A}-\tbSigma_{1A})
-
(\hbSigma_{2A}-\tbSigma_{2A})
\right\} =\operatorname{vech}
\left(
-\barbz_1\barbz_1^\top
+
\barbz_2\barbz_2^\top
\right),
\]
where $\barbz_i=n_i^{-1}\sum_{j=1}^{n_i}\bz_{ij}, i=1,2$. Since \(\|\operatorname{vech}(\bA)\|\le\|\bA\|_{\mathrm{F}}\) for every
symmetric matrix \(\bA\), $\|\hbw-\tbw\|
\le \|\barbz_1\|^2+\|\barbz_2\|^2$. Moreover, \eqref{tDkk.inequ} gives \(\|\widetilde{\bD}^{-1/2}\|_{\mathrm{op}}
\le\sqrt{n_{\min}/c_H}\). Therefore,
\begin{equation}
\label{eq:b-centering-bound}
\|\bb\|
\le
\sqrt{\frac{n_{\min}}{d\,c_H}}
\left(
\|\barbz_1\|^2+\|\barbz_2\|^2
\right).
\end{equation}

We now bound the term on the right-hand side. Write $\bar z_{ia}
=n_i^{-1}\sum_{j=1}^{n_i}z_{ija}$. By the Cauchy--Schwarz inequality and \eqref{Ezia8.equ},
\[
\E\|\barbz_i\|^4=
\sum_{a,b=1}^p
\E(\bar z_{ia}^2\bar z_{ib}^2) \le
p^2
\max_{1\le a\le p}\E|\bar z_{ia}|^4\le
p^2
\max_{1\le a\le p}
\left\{
\mathbb E|\bar z_{ia}|^8
\right\}^{1/2} \le
\sqrt{C_z}
\left(
\frac p{n_i}
\right)^2, \;i=1,2.
\]
It follows from \eqref{eq:b-centering-bound} that 
\begin{equation}\label{eq:Eb-centering-bound}
\E\|\bb\|^2
\le
\frac{n_{\min}}{dc_H}
\E\left(
\|\barbz_1\|^2+\|\barbz_2\|^2
\right)^2 \le
\frac{2\sqrt{C_z}\,n_{\min}p^2}{dc_H}
\left(
\frac1{n_1^2}+\frac1{n_2^2}
\right)
\le
\frac{8\sqrt{C_z}}{c_Hn_{\min}},
\end{equation}
where the last inequality uses
\(p^2/d\le2\) and \(n_i\ge n_{\min}\). Hence Markov's inequality gives
\[
\mathbb P\left\{
\|\bb\|>
\left(
\frac{24\sqrt{C_z}}
{c_Hn_{\min}\eta}
\right)^{1/2}
\right\}
\le
\frac{\eta}{3}.
\]

Finally, Lemma~\ref{lem:diagonal_studentizer}, applied with
\(\delta=\eta/3\), gives
\[
\mathbb P\left\{
r_D>
\left(
\frac{3C_Dd}{n_{\min}\eta}
\right)^{1/2}
\right\}
\le
\frac{\eta}{3}.
\]

Let \(\mathcal E_\eta\) be the intersection of the three events on
which the preceding bounds for \(\|\ba\|\), \(\|\bb\|\), and \(r_D\)
hold. By the union bound,
\(\mathbb P(\mathcal E_\eta)\ge1-\eta\). Since \(C_T\), \(C_z\), and
\(C_D\) depend only on \(K_8\) and \(c_H\), there exists a constant
\(C>0\), depending only on these two constants, such that, on
\(\mathcal E_\eta\),
\[
\|\ba\|\le C\eta^{-1/4},
\qquad
\|\bb\|\le\frac{C}{\sqrt{n_{\min}\eta}},
\qquad\text{and }\quad
r_D\le
C\sqrt{\frac{d}{n_{\min}\eta}}.
\]
Choose the constant \(C_2\) in the theorem sufficiently large.
Then the condition
\(C_2\sqrt{d/(n_{\min}\eta)}<1/2\) ensures that
\(r_D<1/2\) on \(\mathcal E_\eta\). Substituting the preceding bounds
into \eqref{eq:supp-transfer-perturbation} gives
\begin{align*}
|T_n-\widetilde T_n|
\le
C_2\left[
\sqrt{\frac{d}{n_{\min}\eta}}
\left\{
\eta^{-1/4}
+
\frac{1}{\sqrt{n_{\min}\eta}}
\right\}^{2}
+
\frac{1}{\sqrt{n_{\min}}\eta^{3/4}}
+
\frac{1}{n_{\min}\eta}
\right]
=
\Delta(\eta),
\end{align*}
on \(\mathcal E_\eta\). This proves the first assertion.

For the explicit rate, set
\(\eta=n_{\min}^{-1/8}\). Then
\[
\Delta(\eta)
\le
C_2\left[
\sqrt{\frac{d}{n_{\min}^{7/8}}}
\left(
n_{\min}^{1/32}
+
n_{\min}^{-7/16}
\right)^2
+
n_{\min}^{-13/32}
+
n_{\min}^{-7/8}
\right].
\]
Since \(n_{\min}\ge1\), $(
n_{\min}^{1/32}
+
n_{\min}^{-7/16})^2
\le
4n_{\min}^{1/16}$. Moreover, because \(d\ge1\),
\(n_{\min}^{-13/32}\) and \(n_{\min}^{-7/8}\) are both bounded by
\((d/n_{\min}^{3/4})^{1/2}\). Consequently, after increasing \(C_2\)
if necessary,
\[
\Delta(\eta)
\le
C_2
\left(
\frac{d}{n_{\min}^{3/4}}
\right)^{1/2}.
\]
Therefore, with probability at least \(1-n_{\min}^{-1/8}\),
\[
|T_n-\widetilde T_n|
\le
C_2
\left(
\frac{d}{n_{\min}^{3/4}}
\right)^{1/2}.
\]
% If \(d=o(n_{\min}^{3/4})\), the right-hand side converges to zero,
% which proves that $|T_n-\widetilde T_n|=o_{\mathbb P}(1)$. 
This completes the proof.
\end{proof}
\begin{proof}[Proof of Corollary~\ref{cor:feasible_gaussian}]
Since \(\widetilde T_n\) and \(T_Z\) are nonnegative,
Theorem~\ref{thm:oracle_gaussian} also gives
\[
\sup_{r\in\mathbb R}
\left|
\mathbb P(\widetilde T_n\le r)-\mathbb P(T_Z\le r)
\right|
\le \varepsilon_1,
\]
where $\varepsilon_1=C_1\{
\kappa^{3/4}(\bR)\delta_1^{1/4}
+
\kappa^{2/3}(\bR)\delta_2^{1/3}\}$. The extension of the supremum to \(r\in\mathbb R\) is harmless because
both distribution functions vanish for \(r<0\).

By Theorem~\ref{thm:sample_oracle_transfer}, the event $\mathcal E_\eta
=\{|T_n-\widetilde T_n|
\le
\Delta(\eta)\}$ satisfies $\mathbb P(\mathcal E_\eta)\ge1-\eta$. For any \(r\ge0\), on \(\mathcal E_\eta\), $T_n\le r\Longrightarrow
\widetilde T_n\le r+\Delta(\eta)$. Therefore,
\begin{align*}
\mathbb P(T_n\le r)\le
\mathbb P\{\widetilde T_n\le r+\Delta(\eta)\}
+\mathbb P(\mathcal E_\eta^c)\le \mathbb P\{T_Z\le r+\Delta(\eta)\}+\varepsilon_1+\eta.
\end{align*}
Hence
\begin{equation}
\label{eq:cor-upper-bound}
\mathbb P(T_n\le r)-\mathbb P(T_Z\le r)
\le \mathbb P\{T_Z\le r+\Delta(\eta)\}-\mathbb P(T_Z\le r)+\varepsilon_1+\eta.
\end{equation}

Conversely, on \(\mathcal E_\eta\), $\widetilde T_n\le r-\Delta(\eta)\Longrightarrow T_n\le r$. Using the convention that the distribution functions vanish at negative
arguments, we obtain
\begin{align*}
\mathbb P(T_n\le r)&\ge
\mathbb P(\{\widetilde T_n\le r-\Delta(\eta)\}\cap\mathcal E_\eta) \ge
\mathbb P\{\widetilde T_n\le r-\Delta(\eta)\}-\mathbb P(\mathcal E_\eta^c) \\
&\ge \mathbb P\{T_Z\le r-\Delta(\eta)\}-\varepsilon_{1}-\eta.
\end{align*}
Consequently,
\begin{equation}
\label{eq:cor-lower-bound}
\mathbb P(T_Z\le r)-\mathbb P(T_n\le r)
\le
\mathbb P(T_Z\le r)-\mathbb P\{T_Z\le r-\Delta(\eta)\}+\varepsilon_1+\eta.
\end{equation}
Define the anti-concentration modulus of \(T_Z\) by $\omega_Z(h)
=
\sup_{s\ge0}
\mathbb P(s<T_Z\le s+h),  h\ge0$.
For every \(r\ge0\), $\mathbb P\{T_Z\le r+\Delta(\eta)\}-\mathbb P(T_Z\le r)\le\omega_Z\{\Delta(\eta)\}$. If \(r\ge\Delta(\eta)\), then $\mathbb P(T_Z\le r)-\mathbb P\{T_Z\le r-\Delta(\eta)\}
\le\omega_Z\{\Delta(\eta)\}$; if \(0\le r<\Delta(\eta)\), the nonnegativity of \(T_Z\) and
\(\mathbb P(T_Z=0)=0\) imply that $\mathbb P(T_Z\le r)
-\mathbb P\{T_Z\le r-\Delta(\eta)\}=\mathbb P(T_Z\le r)\le
\mathbb P\{0<T_Z\le\Delta(\eta)\} \le
\omega_Z\{\Delta(\eta)\}$.

Thus, combining \eqref{eq:cor-upper-bound} and
\eqref{eq:cor-lower-bound}, we obtain $|\mathbb P(T_n\le r)-\mathbb P(T_Z\le r)|\leq\varepsilon_{1}+\eta+\omega_Z\{\Delta(\eta)\}$
for every \(r\ge0\). Therefore,
\begin{equation}
\label{eq:cor-transfer-bound}
\sup_{r\ge0}
\left|
\mathbb P(T_n\le r)
-
\mathbb P(T_Z\le r)
\right|
\le
\varepsilon_{1}
+
\eta
+
\omega_Z\{\Delta(\eta)\}.
\end{equation}
By Lemma~5.10 of \citet{fang_large-dimensional_2024}, under Assumption~\ref{ass:gaussian_reference_geometry}, for every \(h>0\),
\[
\sup_{x\ge0}
\mathbb P\left(
x\le\|\bZ\|^2\le x+h
\right)
\le
C_3\kappa(\bR)h.
\]
Since $T_Z=d^{-1}\|\bZ\|^2$, it follows that, for every \(h>0\),
\begin{align*}
\omega_Z(h)
=\sup_{r\ge0}
\mathbb P\left(
r<T_Z\le r+h
\right) =
\sup_{r\ge0}
\mathbb P\left(
rd<\|\bZ\|^2\le rd+hd
\right) \le
C_3\kappa(\bR)dh.
\end{align*}

Taking \(h=\Delta(\eta)\) and substituting the resulting bound into
\eqref{eq:cor-transfer-bound}, we obtain
\[
\sup_{r\ge0}
\left|
\mathbb P(T_n\le r)
-
\mathbb P(T_Z\le r)
\right|
\le
\varepsilon_1
+
\eta
+
C_3\kappa(\bR)d\,\Delta(\eta)
=
\varepsilon(\eta).
\]
This completes the proof.

\end{proof}
\begin{proof}[Proof of Theorem~\ref{thm:bootstrap_explicit}]
Let \(C_R\) be the constant in
Lemma~\ref{lem:bootstrap-correlation-consistency}, and let
\(K_{\mathrm{G}}\) be the universal constant in the Gaussian comparison
inequality used below. Set $C_4=C_R\max\{1,3K_{\mathrm{G}}\}$.

 Conditional on the data, let
\(\bZ^{\mathrm{mb}}\mid\mathcal X
\sim N_d(\mathbf{0},\widehat{\bR})\).  Then
\(T_n^{\mathrm{mb}}\mid\mathcal X
\overset{d}{=}d^{-1}\|\bZ^{\mathrm{mb}}\|^2\), where $\overset{d}{=}$ denotes equality in distribution. %Moreover, \(T_Z=d^{-1}\|\bZ\|^2\) by \eqref{eq:gaussian-reference}. 
Since \(\bR\) is a positive semidefinite correlation matrix,
\(\Lambda_2\le\Lambda_1=\|\bR\|_{\mathrm{F}}
\le\operatorname{tr}(\bR)=d\). Moreover, \(C_4\ge C_R\). Hence the condition of the theorem implies
\[
C_R\sqrt{\frac{d}{n_{\min}\eta_2}}
=
\frac{C_R}{d}\sqrt{\frac{d^3}{n_{\min}\eta_2}}
\le
\frac{\Lambda_2}{2d}
\le \frac12.
\]
Lemma~\ref{lem:bootstrap-correlation-consistency} therefore provides
an event \(\mathcal E_{\eta_2}\), with
\(\mathbb P(\mathcal E_{\eta_2})\ge1-\eta_2\), on which every diagonal
entry of \(\widehat{\bD}\) is positive and
\eqref{eq:R-nuclear-bound} holds.

Let
\(\widehat\lambda_1\ge\cdots\ge\widehat\lambda_d\ge0\) be the eigenvalues
of \(\widehat{\bR}\), and define
\(\widehat\Lambda_s
=(\sum_{\ell=s}^d\widehat\lambda_\ell^2)^{1/2}\) for \(s=1,2\).
By the reverse triangle inequality, the Hoffman--Wielandt inequality,
and \(\|\bA\|_{\mathrm{F}}\le\|\bA\|_{S_1}\), we have $|\widehat\Lambda_s-\Lambda_s|
\le
\|\widehat{\bR}-\bR\|_{\mathrm{F}}
\le
\|\widehat{\bR}-\bR\|_{S_1}, s=1,2$. It follows from \eqref{eq:R-nuclear-bound} and the condition of the
theorem that, on \(\mathcal E_{\eta_2}\),
\[
\widehat\Lambda_2
\ge
\Lambda_2-
C_R\sqrt{\frac{d^3}{n_{\min}\eta_2}}
\ge
\frac{\Lambda_2}{2}>0.
\]
Moreover, since \(\Lambda_1\ge\Lambda_2\),
\[
\widehat\Lambda_1
\ge
\Lambda_1-
C_R\sqrt{\frac{d^3}{n_{\min}\eta_2}}
\ge
\Lambda_1-\frac{\Lambda_2}{2}
\ge
\frac{\Lambda_1}{2}.
\]
Consequently, $\kappa(\widehat{\bR})=
(\widehat\Lambda_1\widehat\Lambda_2)^{-1/2}
\le
2(\Lambda_1\Lambda_2)^{-1/2}
=
2\kappa(\bR)$. Conditional on \(\mathcal X\) and on \(\mathcal E_{\eta_2}\),
\(\widehat{\bR}\) is fixed and satisfies
\(\widehat\Lambda_2>0\).  We may therefore apply Corollary~2.3 of
\citet{gotze_large_2019} with
\(\xi=\bZ^{\mathrm{mb}}\), \(\eta=\bZ\), and \(a=0\).
For every positive semidefinite matrix \(\bQ\) satisfying
\(\Lambda_2(\bQ)>0\), the covariance-geometry factor used in
\citet{gotze_large_2019} satisfies $\varkappa(\bQ)
\le
1.8
\{\Lambda_1(\bQ)\Lambda_2(\bQ)\}^{-1/2}$. Therefore, Corollary~2.3 of \citet{gotze_large_2019} implies that
the universal constant \(K_{\mathrm{G}}>0\) above can be chosen so that
\begin{align*}
\sup_{x>0}
\left|
\mathbb P\left(
\|\bZ^{\mathrm{mb}}\|\le x\mid\mathcal X
\right)
-
\mathbb P(\|\bZ\|\le x)
\right|\le
K_{\mathrm{G}}
\left\{
\kappa(\widehat{\bR})+\kappa(\bR)
\right\}
\|\widehat{\bR}-\bR\|_{S_1}.
\end{align*}

Since \(\widehat\Lambda_2>0\) and \(\Lambda_2>0\), neither Gaussian norm
has an atom at zero. Moreover, \(r\mapsto\sqrt{dr}\) is a bijection
from \([0,\infty)\) onto \([0,\infty)\). Therefore, the left-hand side
equals
\[
\sup_{r\ge0}
\left|
\mathbb P\left(
T_n^{\mathrm{mb}}\le r\mid\mathcal X
\right)
-
\mathbb P(T_Z\le r)
\right|.
\]
Using
\(\kappa(\widehat{\bR})\le2\kappa(\bR)\) and
\eqref{eq:R-nuclear-bound}, we conclude that, on
\(\mathcal E_{\eta_2}\),
\begin{align*}
\sup_{r\ge0}
\left|
\mathbb P\left(
T_n^{\mathrm{mb}}\le r\mid\mathcal X
\right)
-
\mathbb P(T_Z\le r)
\right|\le
3K_{\mathrm{G}}C_R\kappa(\bR)
\sqrt{\frac{d^3}{n_{\min}\eta_2}}
\le
C_4\kappa(\bR)
\sqrt{\frac{d^3}{n_{\min}\eta_2}}.
\end{align*}
Using \(\mathbb P(\mathcal E_{\eta_2})\ge1-\eta_2\) completes the proof.
\end{proof}
\begin{proof}[Proof of Corollary~\ref{cor:size_validity}]
Let $F_n(r)=\mathbb P_{H_0}(T_n\le r)$ and $G_n(r)=\mathbb P(T_Z\le r)$. Define the conditional bootstrap distribution function by $F_{\mathrm{mb},n}(r)
=
\mathbb P\left(T_n^{\mathrm{mb}}\le r\mid\mathcal X\right)$. Set $a_n=\varepsilon(\eta_n)$ and $b_n=
C_4\kappa(\bR)
\sqrt{d^3/(n_{\min}\eta_{n})}$.
By the assumptions of the corollary,
\(a_n\to0\), \(b_n\to0\), and \(\eta_{n}\to0\).
Corollary~\ref{cor:feasible_gaussian} gives $\sup_{r\ge0}|F_n(r)-G_n(r)|
\le a_n$. Moreover, by Theorem~\ref{thm:bootstrap_explicit}, the event $\mathcal B_n
=\{
\sup_{r\ge0}|
 F_{\mathrm{mb},n}(r)-G_n(r)|
\le b_n\}$ satisfies
\(\mathbb P_{H_0}(\mathcal B_n^c)\le\eta_{n}\).
Because \(\Lambda_2>0\), the matrix \(\bR\) has rank at least two, and  $T_Z
\overset{d}{=}d^{-1}\sum_{\ell=1}^d\lambda_\ell\xi_\ell^2$, where
\(\xi_1,\ldots,\xi_d\) are independent standard normal random variables and
\(\lambda_1,\ldots,\lambda_d\) are the eigenvalues of \(\bR\).  In particular, \(G_n\) is
continuous on \([0,\infty)\). Therefore,
for \(u\in(0,1)\), its generalized quantile
\(q_{u,n}=\inf\{r\ge0:G_n(r)\ge u\}\) satisfies \(G_n(q_{u,n})=u\).

For all sufficiently large \(n\),
\(b_n<\min\{\alpha,1-\alpha\}\).  On \(\mathcal B_n\), the definition $c_{1-\alpha}^{\mathrm{mb}}
=
\inf\left\{
r\ge0:
 F_{\mathrm{mb},n}(r)\ge1-\alpha
\right\}$ and the bound
\(\sup_r| F_{\mathrm{mb},n}(r)-G_n(r)|\le b_n\) imply
\begin{equation}
\label{eq:supp-size-quantile-sandwich}
q_{1-\alpha-b_n,n}
\le
c_{1-\alpha}^{\mathrm{mb}}
\le
q_{1-\alpha+b_n,n}.
\end{equation}
Using the left inequality in
\eqref{eq:supp-size-quantile-sandwich},
\begin{align*}
\mathbb P_{H_0}
\left(T_n>c_{1-\alpha}^{\mathrm{mb}}\right)
&\le
\mathbb P_{H_0}
\left(T_n>q_{1-\alpha-b_n,n}\right)
+
\mathbb P_{H_0}(\mathcal B_n^c) \\
&\le
1-G_n(q_{1-\alpha-b_n,n})+a_n+\eta_{n}=
\alpha+b_n+a_n+\eta_{n}.
\end{align*}
Similarly, the right inequality gives
\begin{align*}
\mathbb P_{H_0}
\left(T_n>c_{1-\alpha}^{\mathrm{mb}}\right)
&\ge
\mathbb P_{H_0}
\left(T_n>q_{1-\alpha+b_n,n}\right)
-
\mathbb P_{H_0}(\mathcal B_n^c) \\
&\ge
1-G_n(q_{1-\alpha+b_n,n})-a_n-\eta_{n} =
\alpha-b_n-a_n-\eta_{n}.
\end{align*}
Consequently,
\[
\left|
\mathbb P_{H_0}
\left(T_n>c_{1-\alpha}^{\mathrm{mb}}\right)
-\alpha
\right|
\le
a_n+b_n+\eta_{n}
\to0,
\]
which proves the result.
\end{proof}
\begin{proof}[Proof of Theorem~\ref{thm:power_consistency}]
Under \(H_1\), since
\(\mathbb E(\bu_{i1})
=\operatorname{vech}(\bSigma_{iA})\), we have
\(\widetilde{\bw}
=\boldsymbol{\Delta}_{\Sigma}+\bar{\bu}_1^*-\bar{\bu}_2^*\).
Let $\bxi
=\widetilde{\bD}^{-1/2}
(\bar{\bu}_1^*-\bar{\bu}_2^*)$ and $\bb
=d^{-1/2}\widetilde{\bD}^{-1/2}
(\widehat{\bw}-\widetilde{\bw})$. The vector \(\bxi\) is the centered standardized noise vector used in
the proof of Theorem~\ref{thm:oracle_gaussian}; in particular,
\(\mathbb E(\bxi)=\mathbf{0}\) and \(\operatorname{Cov}(\bxi)=\bR\).
It follows that
\begin{equation}
\label{eq:power-signal-noise-decomposition}
\widetilde{\bD}^{-1/2}\widehat{\bw}
=
\boldsymbol{\theta}_n+\bxi+\sqrt d\,\bb.
\end{equation}

On the event \(\{r_D<1\}\), \(\widehat D_{kk}
\ge
(1-r_D)\widetilde D_{kk}
>0\) for every \(k\). Thus \(\widehat{\bD}\) is positive definite
on this event, and we may define
\(\bH
=
\widetilde{\bD}^{1/2}
\widehat{\bD}^{-1}
\widetilde{\bD}^{1/2}\). By
\eqref{eq:power-signal-noise-decomposition}, the test statistic can be
written, on \(\{r_D<1\}\), as
\begin{equation}
\label{eq:power-quadratic-representation}
T_n
=
\frac1d
\left(
\boldsymbol{\theta}_n+\bxi+\sqrt d\,\bb
\right)^\top
\bH
\left(
\boldsymbol{\theta}_n+\bxi+\sqrt d\,\bb
\right).
\end{equation}
Since \(\bR\) is a correlation matrix,
\begin{equation}
\label{eq:power-xi-second-moment}
\mathbb E\|\bxi\|^2
=
\operatorname{tr}(\bR)
=
d.
\end{equation}
Moreover, the calculation leading to \eqref{eq:Eb-centering-bound} does
not use \(H_0\) and gives
\begin{equation}
\label{eq:power-centering-mse}
\mathbb E\|\bb\|^2
\le
\frac{C_b}{n_{\min}}.
\end{equation}
Here \(C_b>0\) depends only on \(K_8\) and \(c_H\).

Lemma~\ref{lem:diagonal_studentizer} also gives
\begin{equation}
\label{eq:power-rD-half}
\mathbb P_{H_1}
\left(
r_D>\frac12
\right)
\le
4C_D\frac{d}{n_{\min}}.
\end{equation}
Indeed, if \(4C_Dd/n_{\min}<1\), this follows from
\eqref{eq:studentizer-high-probability} by taking
\(\delta=4C_Dd/n_{\min}\); otherwise the bound is trivial.

Define
\[
\mathcal G_n
=
\left\{
r_D\le\frac12
\right\}
\cap
\left\{
\|\bxi\|
\le
\frac14\sqrt{dh_n}
\right\}
\cap
\left\{
\|\bb\|
\le
\frac14\sqrt{h_n}
\right\}.
\]
By Markov's inequality and
\eqref{eq:power-xi-second-moment},
\(\mathbb P_{H_1}\{\|\bxi\|>4^{-1}\sqrt{dh_n}\}
\le16/h_n\). Similarly, \eqref{eq:power-centering-mse} gives
\(\mathbb P_{H_1}\{\|\bb\|>4^{-1}\sqrt{h_n}\}
\le16C_b/(n_{\min}h_n)\).
Consequently,
\begin{align}
\mathbb P_{H_1}(\mathcal G_n^c)\le
4C_D\frac{d}{n_{\min}}
+
\frac{16}{h_n}
+
\frac{16C_{b}}{n_{\min}h_n}\le
C_{5}
\left(
\frac{d}{n_{\min}}+\frac1{h_n}
\right),
\label{eq:power-good-event-bound}
\end{align}
where \(C_{5}>0\) depends only on \(K_8\) and \(c_H\).

Since \(\mathcal G_n\subset\{r_D<1\}\), the representation in
\eqref{eq:power-quadratic-representation} holds on
\(\mathcal G_n\). Moreover, \(r_D\le1/2\) implies
\(1/2\le \widehat D_{kk}/\widetilde D_{kk}\le3/2\) for every \(k\).
Therefore,
\(H_{kk}
=(\widehat D_{kk}/\widetilde D_{kk})^{-1}
\ge2/3\), and hence \(\bH\succeq(2/3)\bI_d\) on \(\mathcal G_n\).

The reverse triangle
inequality and the definition of \(\mathcal G_n\) give
\begin{align*}
\left\|
\boldsymbol{\theta}_n+\bxi+\sqrt d\,\bb
\right\|
&\ge
\|\boldsymbol{\theta}_n\|
-
\|\bxi\|
-
\sqrt d\,\|\bb\|\\
&\ge
\sqrt{dh_n}
-
\frac14\sqrt{dh_n}
-
\frac14\sqrt{dh_n}=\frac12\sqrt{dh_n}.
\end{align*}
Consequently,
\begin{equation*}
\label{eq:power-T-lower-bound}
T_n
\ge
\frac{2}{3d}
\left\|
\boldsymbol{\theta}_n+\bxi+\sqrt d\,\bb
\right\|^2
\ge
\frac{h_n}{6},
\qquad\text{on }\mathcal G_n.
\end{equation*}
On \(\{r_D\le1/2\}\), every diagonal entry of
\(\widehat{\bD}\) is positive and \(\widehat{\bR}\) is a correlation
matrix. Therefore,
\(\mathbb E(T_n^{\mathrm{mb}}\mid\mathcal X)
=d^{-1}\operatorname{tr}(\widehat{\bR})=1\). Conditional Markov's
inequality then gives
\(\mathbb P(T_n^{\mathrm{mb}}>\alpha^{-1}\mid\mathcal X)\le\alpha\).
It follows from the definition of
\(c_{1-\alpha}^{\mathrm{mb}}\) that
% \begin{equation}
% \label{eq:power-critical-value-bound}
$c_{1-\alpha}^{\mathrm{mb}}
\le
\alpha^{-1}$ on $\{r_D\le1/2\}$. If \(h_n>6/\alpha\), then, on \(\mathcal G_n\),
\[
T_n
\ge \frac{h_n}{6}
>
\alpha^{-1}
\ge
c_{1-\alpha}^{\mathrm{mb}}.
\]
Thus, \(\mathcal G_n
\subset
\left\{
T_n>c_{1-\alpha}^{\mathrm{mb}}
\right\}\). Combining this inclusion with
\eqref{eq:power-good-event-bound}, we obtain
\[
\mathbb P_{H_1}
\left(
T_n>c_{1-\alpha}^{\mathrm{mb}}
\right)
\ge
1-
C_{5}
\left(
\frac{d}{n_{\min}}+\frac1{h_n}
\right),
\]
which proves \eqref{eq:finite-sample-power-bound}.

Under condition~\textup{(i)},
\(h_n>6/\alpha\) for all sufficiently large \(n\), and $d/n_{\min}+1/h_n\to0$. Therefore, \eqref{eq:finite-sample-power-bound} implies $\mathbb P_{H_1}(T_n>c_{1-\alpha}^{\mathrm{mb}})\to 1$.

It remains to establish consistency under
condition~\textup{(ii)}. Let
\(q_n=\|\boldsymbol{\theta}_n\|^2=dh_n\), and suppose
condition~\textup{(ii)} of Theorem~\ref{thm:power_consistency} holds. Since
\(\Lambda_1\le d\), the first requirement in condition~\textup{(ii)} implies
\[
\sqrt{\frac{d}{n_{\min}}}
\le
\frac{d^{3/2}}{\sqrt{n_{\min}}\Lambda_1}
\to0.
\]
Thus \(d/n_{\min}\to0\). Lemma~\ref{lem:diagonal_studentizer} and
\eqref{eq:power-centering-mse} give $r_D=O_{\mathbb P}(\sqrt{d/n_{\min}})$ and $\|\bb\|=O_{\mathbb P}(n_{\min}^{-1/2})$.
In particular,
\begin{equation}
\label{eq:refined-rD-rate}
\frac{d}{\Lambda_1}r_D=o_{\mathbb P}(1).
\end{equation}

The fourth-moment calculation leading to
\eqref{eq:xi-vector-fourth-sum} is valid under \(H_1\) as well. A
fourth-moment expansion for the sum of the independent centered vectors
\(\bxi_j\) therefore gives
\[
\operatorname{Var}(\|\bxi\|^2)
\le
2\|\bR\|_{\mathrm{F}}^2
+
\sum_{j=1}^{n}\mathbb E\|\bxi_j\|^4
\le
2\Lambda_1^2+C\frac{d^2}{n_{\min}}.
\]
Under the first requirement in condition~\textup{(ii)}, the
last term is \(o(\Lambda_1^2)\). Hence
\begin{equation}
\label{eq:refined-noise-quadratic}
\|\bxi\|^2=d+O_{\mathbb P}(\Lambda_1).
\end{equation}
Moreover, $\operatorname{Var}(\boldsymbol{\theta}_n^\top\bxi)
=
\boldsymbol{\theta}_n^\top\bR\boldsymbol{\theta}_n
\le
\lambda_1q_n
\le
\Lambda_1q_n$. Since \(q_n/\Lambda_1\to\infty\), it follows that
\begin{equation}
\label{eq:refined-signal-noise-cross}
\boldsymbol{\theta}_n^\top\bxi=o_{\mathbb P}(q_n).
\end{equation}

Combining \eqref{eq:refined-noise-quadratic} and
\eqref{eq:refined-signal-noise-cross} yields $d^{-1}\|\boldsymbol{\theta}_n+\bxi\|^2
=
1+h_n+o_{\mathbb P}(h_n)$. Furthermore,
\begin{align*}
\left|
\frac1d\|\boldsymbol{\theta}_n+\bxi+\sqrt d\,\bb\|^2
-
\frac1d\|\boldsymbol{\theta}_n+\bxi\|^2
\right| \le
\frac{2}{\sqrt d}\|\bb\|
(\|\boldsymbol{\theta}_n\|+\|\bxi\|)
+\|\bb\|^2
=o_{\mathbb P}(h_n).
\end{align*}
Indeed, after division by \(h_n=q_n/d\), the signal-centering,
noise-centering, and squared-centering terms are, respectively,
\[
O_{\mathbb P}\left(\sqrt{\frac{d}{n_{\min}q_n}}\right),
\qquad
O_{\mathbb P}\left(\frac{d}{\sqrt{n_{\min}}q_n}\right),
\qquad\text{and}\qquad
O_{\mathbb P}\left(\frac{d}{n_{\min}q_n}\right).
\]
Each tends to zero because \(q_n/\Lambda_1\to\infty\),
\(d^{3/2}/(\sqrt{n_{\min}}\Lambda_1)\to0\), and
\(\sqrt d\le\Lambda_1\le d\).
On the event \(r_D\le1/2\), the argument leading to
\eqref{eq:supp-transfer-perturbation} gives
\[
\left|
 T_n-\frac1d
\|\boldsymbol{\theta}_n+\bxi+\sqrt d\,\bb\|^2
\right|
\le
2r_D\frac1d
\|\boldsymbol{\theta}_n+\bxi+\sqrt d\,\bb\|^2.
\]
By \eqref{eq:refined-rD-rate} and \(q_n/\Lambda_1\to\infty\),
\[
\frac{r_D}{h_n}
=
\frac{(d/\Lambda_1)r_D}{q_n/\Lambda_1}
=o_{\mathbb P}(1).
\]
Together with \(r_D=o_{\mathbb P}(1)\), this shows that the right-hand
side is \(o_{\mathbb P}(h_n)\). Consequently,
\begin{equation}
\label{eq:refined-Tn-expansion}
T_n=1+h_n+o_{\mathbb P}(h_n),
\qquad
dT_n=d+q_n+o_{\mathbb P}(q_n).
\end{equation}

It remains to control the bootstrap critical value. The first requirement
in condition~\textup{(ii)} implies \(d/n_{\min}\to0\), so
the prerequisite of Lemma~\ref{lem:bootstrap-correlation-consistency}
holds for every fixed tail probability and all sufficiently large \(n\).
By that lemma and the first requirement
in condition~\textup{(ii)},
\[
\|\widehat{\bR}-\bR\|_{\mathrm{F}}
\le
\|\widehat{\bR}-\bR\|_{S_1}
=O_{\mathbb P}\left(\sqrt{\frac{d^3}{n_{\min}}}\right)
=o_{\mathbb P}(\Lambda_1).
\]
Therefore,
\(\|\widehat{\bR}\|_{\mathrm{F}}=O_{\mathbb P}(\Lambda_1)\).
Conditionally on the data and on the event that \(\widehat{\bD}\) has
positive diagonal entries,
\[
\mathbb E(T_n^{\mathrm{mb}}\mid\mathcal X)=1,
\qquad
\operatorname{Var}(T_n^{\mathrm{mb}}\mid\mathcal X)
=
\frac{2}{d^2}\|\widehat{\bR}\|_{\mathrm{F}}^2.
\]
Conditional Chebyshev's inequality thus gives
\[
c_{1-\alpha}^{\mathrm{mb}}
\le
1+\sqrt{\frac{2}{\alpha}}\,
\frac{\|\widehat{\bR}\|_{\mathrm{F}}}{d}
=
1+O_{\mathbb P}\left(\frac{\Lambda_1}{d}\right)
=
1+o_{\mathbb P}(h_n).
\]
The last equality follows from
\(\Lambda_1/d=h_n\Lambda_1/q_n=o(h_n)\).
Together with \eqref{eq:refined-Tn-expansion}, this implies
\(T_n-c_{1-\alpha}^{\mathrm{mb}}\ge h_n+o_{\mathbb P}(h_n)\).
Therefore,
\(\mathbb P_{H_1}(T_n>c_{1-\alpha}^{\mathrm{mb}})\to1\), which
completes the proof.
\end{proof}

\begin{proof}[Proof of Theorem~\ref{thm:estimated-transformation}]
Recall that \(\bx_{ij}\) denotes the observation transformed by
\(\bA\), and let \(\bx_{ij,\widehat A}\) denote the corresponding
observation transformed by \(\widehat{\bA}\). The centered observations and sample covariance matrices computed
under the two transformations satisfy $\bx_{ij,\widehat A}-\bar{\bx}_{i,\widehat A}
=
\widehat{\bA}\bA^{-1}(\bx_{ij}-\bar{\bx}_i)$ and $\widehat{\bSigma}_{i,\widehat A}
=
\widehat{\bA}\bA^{-1}\widehat{\bSigma}_{iA}(\widehat{\bA}\bA^{-1})^\top$. Consequently,
\begin{equation}\label{Lhij.equ}
    (\bx_{ij,\widehat A}-\bar{\bx}_{i,\widehat A})
(\bx_{ij,\widehat A}-\bar{\bx}_{i,\widehat A})^\top
-
\widehat{\bSigma}_{i,\widehat A}=
\widehat{\bA}\bA^{-1}\{
(\bx_{ij}-\bar{\bx}_i)
(\bx_{ij}-\bar{\bx}_i)^\top
-
\widehat{\bSigma}_{iA}\}
(\widehat{\bA}\bA^{-1})^\top.
\end{equation}
Let \(\bL_n\) be the \(d\times d\) matrix satisfying $\operatorname{vech}
\left\{
(\widehat{\bA}\bA^{-1})
\bM
(\widehat{\bA}\bA^{-1})^\top
\right\}
=
\bL_n\operatorname{vech}(\bM)$ for every symmetric \(p\times p\) matrix \(\bM\). Such a matrix exists
and is unique because \(\operatorname{vech}\) is a linear isomorphism
from the space of symmetric \(p\times p\) matrices to
\(\mathbb R^d\). Taking \(\operatorname{vech}\) on
both sides of \eqref{Lhij.equ} gives $\bh_{ij,\widehat A}
=
\bL_n\bh_{ij}$, where $\bh_{ij}$ is defined in \eqref{eq:covariance-contributions}.
The transformation rule for the
sample covariance matrices also gives
\(\widehat{\bw}_{\widehat A}=\bL_n\widehat{\bw}\). Moreover, using
\(\bh_{ij,\widehat A}=\bL_n\bh_{ij}\) and the same multipliers under
the two transformations gives
\(\widehat{\bw}_{\widehat A}^{\mathrm{mb}}
=\bL_n\widehat{\bw}^{\mathrm{mb}}\). Here
\(\widehat{\bA}\) is computed once from the observed data and is held
fixed when the multipliers are generated.

The studentizing matrix also changes when \(\bA\) is replaced by
\(\widehat{\bA}\). Let
\(\widehat{\bD}_{\widehat A}\) denote the resulting diagonal
studentizer. Whenever both studentizers have positive diagonal
entries, define $\bG_n
=
\widehat{\bD}_{\widehat A}^{-1/2}
\bL_n
\widehat{\bD}^{1/2}$. Then $\widehat{\bD}_{\widehat A}^{-1/2}
\widehat{\bw}_{\widehat A}
=
\bG_n
\widehat{\bD}^{-1/2}\widehat{\bw}$. Hence  $T_n
=d^{-1}\|
\widehat{\bD}^{-1/2}\widehat{\bw}\|^2$, whereas $\widehat T_n
=d^{-1}\|
\bG_n
\widehat{\bD}^{-1/2}\widehat{\bw}\|^2$.
It therefore suffices to prove that
\(\|\bG_n^\top\bG_n-\bI_d\|_{\mathrm{op}}\convP0\).

Let $e_n
=\|
\widehat{\bA}\bA^{-1}-\bI_p\|_{\mathrm{op}}$. For every symmetric \(p\times p\) matrix \(\bM\), the congruence
expansion gives $\|
(\widehat{\bA}\bA^{-1})
\bM
(\widehat{\bA}\bA^{-1})^\top
-\bM\|_{\mathrm{F}}
\le
(2e_n+e_n^2)\|\bM\|_{\mathrm{F}}$. Since $\|\operatorname{vech}(\bM)\|
\le
\|\bM\|_{\mathrm{F}}
\le
\sqrt{2}\,
\|\operatorname{vech}(\bM)\|$, it follows from the definition of \(\bL_n\) that
\begin{equation}
\label{eq:Ln-perturbation}
\|\bL_n-\bI_d\|_{\mathrm{op}}
\le
\sqrt{2}(2e_n+e_n^2).
\end{equation}
By Lemma~\ref{lem:diagonal_studentizer} and \(d/n_{\min}\to0\), we have \(r_D=o_{\mathbb P}(1)\). Moreover, the
moment and nondegeneracy assumptions imply, uniformly over \(k\), that
\(c_H(n_1^{-1}+n_2^{-1})\le \widetilde D_{kk}
\le K_8^{1/2}(n_1^{-1}+n_2^{-1})\). On the event \(r_D<1\),
\(\widehat D_{kk}\ge(1-r_D)\widetilde D_{kk}>0\) for every \(k\). Consequently, with probability tending to one,
\(\widehat{\bD}\) has positive diagonal entries and
\[
\frac{\max_k\widehat D_{kk}}
{\min_k\widehat D_{kk}}
\le
\frac{1+r_D}{1-r_D}
\frac{\max_k\widetilde D_{kk}}
{\min_k\widetilde D_{kk}}
=
O_{\mathbb P}(1).
\]
To account for the change in the diagonal studentizer, express
\(\bL_n\) in the original standardized coordinates by setting $\bQ_n
=
\widehat{\bD}^{-1/2}
\bL_n
\widehat{\bD}^{1/2}$. By \eqref{eq:Ln-perturbation},
\begin{equation}
\label{eq:Qn-perturbation}
\|\bQ_n-\bI_d\|_{\mathrm{op}}
\le
\left(
\frac{\max_k\widehat D_{kk}}
{\min_k\widehat D_{kk}}
\right)^{1/2}
\sqrt{2}(2e_n+e_n^2).
\end{equation}
Since the condition of the theorem gives
\(\sqrt d\,e_n\convP0\), it also gives \(e_n\convP0\). Together with \eqref{eq:Qn-perturbation}, this yields
\(\sqrt d\,\|\bQ_n-\bI_d\|_{\mathrm{op}}\convP0\).

We next examine how the diagonal studentizer changes. Since
\(\bh_{ij,\widehat A}=\bL_n\bh_{ij}\), the corresponding bootstrap
covariance matrix satisfies
\(\widehat{\bOmega}_{n,\widehat A}
=\bL_n\widehat{\bOmega}_n\bL_n^\top\). Using
\(\widehat{\bOmega}_n
=\widehat{\bD}^{1/2}\widehat{\bR}\widehat{\bD}^{1/2}\) and the
definition of \(\bQ_n\), we obtain $\widehat{\bOmega}_{n,\widehat A}
=
\widehat{\bD}^{1/2}
\bQ_n\widehat{\bR}\bQ_n^\top
\widehat{\bD}^{1/2}$. Taking diagonal entries gives
\begin{equation}
\label{eq:Dhat-A-representation}
\widehat{\bD}_{\widehat A}
=
\widehat{\bD}^{1/2}
\operatorname{diag}
\left(
\bQ_n\widehat{\bR}\bQ_n^\top
\right)
\widehat{\bD}^{1/2}.
\end{equation}

Let \(\boldsymbol e_k\) be the \(k\)-th canonical basis vector. On the
event that \(\widehat{\bD}\) has positive diagonal entries,
\(\widehat{\bR}\) is a positive semidefinite correlation matrix. Therefore, \(\|\widehat{\bR}\|_{\mathrm{op}}\le d\) and
\(\|\widehat{\bR}\boldsymbol e_k\|^2
\le\|\widehat{\bR}\|_{\mathrm{op}}\widehat R_{kk}\le d\).
Expanding \(\bQ_n=\bI_d+(\bQ_n-\bI_d)\) gives
\begin{equation}
\label{eq:Qn-diagonal-perturbation}
\max_{1\le k\le d}
\left|
\boldsymbol e_k^\top
\bQ_n\widehat{\bR}\bQ_n^\top
\boldsymbol e_k-1
\right|
\le
2\sqrt d\,
\|\bQ_n-\bI_d\|_{\mathrm{op}}
+
d\|\bQ_n-\bI_d\|_{\mathrm{op}}^2
\convP0.
\end{equation}
It follows from \eqref{eq:Dhat-A-representation} and
\eqref{eq:Qn-diagonal-perturbation} that
\(\widehat{\bD}_{\widehat A}\) has positive diagonal entries with
probability tending to one. Therefore, \(\bG_n\) is well defined on
an event whose probability tends to one. Moreover,
\eqref{eq:Dhat-A-representation} gives
\begin{equation}
\label{eq:Gn-Qn-representation}
\bG_n
=
\left\{
\operatorname{diag}
\left(
\bQ_n\widehat{\bR}\bQ_n^\top
\right)
\right\}^{-1/2}
\bQ_n.
\end{equation}
By \eqref{eq:Qn-diagonal-perturbation}, $\|
\operatorname{diag}(
\bQ_n\widehat{\bR}\bQ_n^\top)
-\bI_d\|_{\mathrm{op}}
\convP0$. Since all its diagonal entries are bounded away from zero with
probability tending to one, its inverse also converges to
\(\bI_d\) in operator norm.  Moreover,
\(\|\bQ_n-\bI_d\|_{\mathrm{op}}\convP0\), so that
\(\|\bQ_n\|_{\mathrm{op}}=O_{\mathbb P}(1)\) and
\(\|\bQ_n^\top\bQ_n-\bI_d\|_{\mathrm{op}}\convP0\).
Therefore, by \eqref{eq:Gn-Qn-representation},
\begin{align}
\label{eq:Gn-isometry}
\|\bG_n^\top\bG_n-\bI_d\|_{\mathrm{op}}
\le
\|\bQ_n\|_{\mathrm{op}}^2
\|\{
\operatorname{diag}(
\bQ_n\widehat{\bR}\bQ_n^\top)\}^{-1}
-\bI_d\|_{\mathrm{op}}
+
\|\bQ_n^\top\bQ_n-\bI_d\|_{\mathrm{op}}\convP0.
\end{align}
The exact relation between the two standardized covariance-difference
vectors and the Rayleigh-quotient bound give
\begin{equation*}
\label{eq:Tn-plug-in-bound}
|\widehat T_n-T_n|
\le
\|\bG_n^\top\bG_n-\bI_d\|_{\mathrm{op}}\,T_n.
\end{equation*}
Together with \eqref{eq:Gn-isometry}, this implies $\widehat T_n
=
T_n\{1+o_{\mathbb P}(1)\}$.

For the multiplier bootstrap, use the same multipliers under the two
transformations, with \(\widehat{\bA}\) held fixed conditional on the
data. The same Rayleigh-quotient argument gives
\begin{equation}
\label{eq:bootstrap-transformation-sandwich}
\begin{aligned}
\left\{
1-\|\bG_n^\top\bG_n-\bI_d\|_{\mathrm{op}}
\right\}
T_n^{\mathrm{mb}}
\le
\widehat T_n^{\mathrm{mb}}\le
\left\{
1+\|\bG_n^\top\bG_n-\bI_d\|_{\mathrm{op}}
\right\}
T_n^{\mathrm{mb}}
\end{aligned}
\end{equation}
whenever
\(\|\bG_n^\top\bG_n-\bI_d\|_{\mathrm{op}}<1\).

On the event that \(\widehat{\bD}\) has positive diagonal entries,
\(\widehat{\bR}\) is a nonzero correlation matrix. Hence the
conditional Gaussian quadratic form \(T_n^{\mathrm{mb}}\) is positive
almost surely, and \(c_{1-\alpha}^{\mathrm{mb}}>0\) for
\(0<\alpha<1\). Taking generalized conditional
\((1-\alpha)\)-quantiles in
\eqref{eq:bootstrap-transformation-sandwich} therefore gives
\[
\left|
\frac{\widehat c_{1-\alpha}^{\mathrm{mb}}}
{c_{1-\alpha}^{\mathrm{mb}}}
-1
\right|
\le
\|\bG_n^\top\bG_n-\bI_d\|_{\mathrm{op}}
\convP0.
\]
Together with \eqref{eq:Gn-isometry}, this yields $\widehat c_{1-\alpha}^{\mathrm{mb}}
=
c_{1-\alpha}^{\mathrm{mb}}
\{1+o_{\mathbb P}(1)\}$, which completes the proof.
\end{proof}

\bibliographystyle{agsm}
\bibliography{references}

@article{wang_two-sample_2024,
	title = {Two-sample test for high-dimensional covariance matrices: {A} normal-reference approach},
	volume = {204},
	shorttitle = {Two-sample test for high-dimensional covariance matrices},
	doi = {10.1016/j.jmva.2024.105354},
	journal = {Journal of Multivariate Analysis},
	author = {Wang, Jingyi and Zhu, Tianming and Zhang, Jin-Ting},
	month = nov,
	year = {2024},
	pages = {105354},
}

@article{li_two_2012,
	title = {Two sample tests for high-dimensional covariance matrices},
	volume = {40},
	doi = {10.1214/12-AOS993},
	number = {2},
	journal = {The Annals of Statistics},
	author = {Li, Jun and Chen, Song Xi},
	month = apr,
	year = {2012},
	pages = {908--940},
}

@article{box_general_1949,
	title = {A general distribution theory for a class of likelihood criteria},
	volume = {36},
	doi = {10.1093/biomet/36.3-4.317},
	number = {3-4},
	journal = {Biometrika},
	author = {Box, G. E. P.},
	year = {1949},
	pages = {317--346},
}

@article{nagao_some_1973,
	title = {On some test criteria for covariance matrix},
	volume = {1},
	doi = {10.1214/aos/1176342464},
	number = {4},
	journal = {The Annals of Statistics},
	author = {Nagao, Hisao},
	year = {1973},
	pages = {700--709},
}

@article{ledoit_some_2002,
	title = {Some hypothesis tests for the covariance matrix when the dimension is large compared to the sample size},
	volume = {30},
	doi = {10.1214/aos/1031689018},
	number = {4},
	journal = {The Annals of Statistics},
	author = {Ledoit, Olivier and Wolf, Michael},
	year = {2002},
	pages = {1081--1102},
}

@article{chen_tests_2010,
	title = {Tests for high-dimensional covariance matrices},
	volume = {105},
	doi = {10.1198/jasa.2010.tm09560},
	number = {490},
	journal = {Journal of the American Statistical Association},
	author = {Chen, Song Xi and Zhang, Li-Xin and Zhong, Ping-Shou},
	year = {2010},
	pages = {810--819},
}

@article{cai_two-sample_2013,
	title = {Two-{Sample} {Covariance} {Matrix} {Testing} and {Support} {Recovery} in {High}-{Dimensional} and {Sparse} {Settings}},
	volume = {108},
	doi = {10.1080/01621459.2012.758041},
	number = {501},
	journal = {Journal of the American Statistical Association},
	author = {Cai, Tony and Liu, Weidong and Xia, Yin},
	month = mar,
	year = {2013},
	pages = {265--277},
}

@article{chang_comparing_2017,
	title = {Comparing {Large} {Covariance} {Matrices} under {Weak} {Conditions} on the {Dependence} {Structure} and its {Application} to {Gene} {Clustering}},
	volume = {73},
	doi = {10.1111/biom.12552},
	number = {1},
	journal = {Biometrics},
	author = {Chang, Jinyuan and Zhou, Wen and Zhou, Wen-Xin and Wang, Lan},
	month = mar,
	year = {2017},
	pages = {31--41},
}

@article{cai_limiting_2011,
	title = {Limiting laws of coherence of random matrices with applications to testing covariance structure and construction of compressed sensing matrices},
	volume = {39},
	doi = {10.1214/11-AOS879},
	number = {3},
	journal = {The Annals of Statistics},
	author = {Cai, Tony and Jiang, Tiefeng},
	year = {2011},
	pages = {1496--1525},
}

@article{gotze_large_2019,
	title = {Large ball probabilities, {Gaussian} comparison and anti-concentration},
	volume = {25},
	doi = {10.3150/18-BEJ1062},
	number = {4A},
	journal = {Bernoulli},
	author = {Götze, Friedrich and Naumov, Alexey and Spokoiny, Vladimir and Ulyanov, Vladimir},
	month = nov,
	year = {2019},
	pages = {2538--2563},
}

@article{chernozhukov_high-dimensional_2023,
	title = {High-{Dimensional} {Data} {Bootstrap}},
	volume = {10},
	doi = {10.1146/annurev-statistics-040120-022239},
	number = {1},
	journal = {Annual Review of Statistics and Its Application},
	author = {Chernozhukov, Victor and Chetverikov, Denis and Kato, Kengo and Koike, Yuta},
	month = mar,
	year = {2023},
	pages = {427--449},
}

@article{chernozhukov_gaussian_2013,
	title = {Gaussian approximations and multiplier bootstrap for maxima of sums of high-dimensional random vectors},
	volume = {41},
	doi = {10.1214/13-AOS1161},
	number = {6},
	journal = {The Annals of Statistics},
	author = {Chernozhukov, Victor and Chetverikov, Denis and Kato, Kengo},
	year = {2013},
	pages = {2786--2819},
}

@article{chernozhukov_central_2017,
	title = {Central limit theorems and bootstrap in high dimensions},
	volume = {45},
	doi = {10.1214/16-AOP1113},
	number = {4},
	journal = {The Annals of Probability},
	author = {Chernozhukov, Victor and Chetverikov, Denis and Kato, Kengo},
	year = {2017},
	pages = {2309--2352},
}

@article{fang_large-dimensional_2024,
  title={Large-dimensional central limit theorem with fourth-moment error bounds on convex sets and balls},
  author={Fang, Xiao and Koike, Yuta},
  journal={The Annals of Applied Probability},
  volume={34},
  number={2},
  pages={2065--2106},
  year={2024},
  doi={10.1214/23-AAP2014},
  publisher={Institute of Mathematical Statistics}
}

@article{liberzon_molecular_2015,
  title={The Molecular Signatures Database Hallmark Gene Set Collection},
  author={Liberzon, Arthur and Birger, Ceren and Thorvaldsd{\'o}ttir, Helga and Ghandi, Mahmoud and Mesirov, Jill P. and Tamayo, Pablo},
  journal={Cell Systems},
  volume={1},
  number={6},
  pages={417--425},
  year={2015},
  doi={10.1016/j.cels.2015.12.004}
}

@article{tcga_brca_molecular_portraits_2012,
  title={Comprehensive molecular portraits of human breast tumours},
  author={{The Cancer Genome Atlas Network}},
  journal={Nature},
  volume={490},
  number={7418},
  pages={61--70},
  year={2012},
  doi={10.1038/nature11412}
}

@article{goldman_visualizing_2020,
  title={Visualizing and interpreting cancer genomics data via the {Xena} platform},
  author={Goldman, Mary J. and Craft, Brian and Hastie, Mim and Repe{\v{c}}ka, Kristupas and McDade, Fran and Kamath, Akhil and Banerjee, Ayan and Luo, Yunhai and Rogers, Dave and Brooks, Angela N. and Zhu, Jingchun and Haussler, David},
  journal={Nature Biotechnology},
  volume={38},
  number={6},
  pages={675--678},
  year={2020},
  doi={10.1038/s41587-020-0546-8}
}

@article{schott_test_2007,
	title = {A test for the equality of covariance matrices when the dimension is large relative to the sample sizes},
	volume = {51},
	doi = {10.1016/j.csda.2007.03.004},
	number = {12},
	journal = {Computational Statistics \& Data Analysis},
	author = {Schott, James R.},
	month = aug,
	year = {2007},
	pages = {6535--6542},
}

@article{srivastava_testing_2010,
	title = {Testing the equality of several covariance matrices with fewer observations than the dimension},
	volume = {101},
	doi = {10.1016/j.jmva.2009.12.010},
	number = {6},
	journal = {Journal of Multivariate Analysis},
	author = {Srivastava, Muni S. and Yanagihara, Hirokazu},
	month = jul,
	year = {2010},
	pages = {1319--1329},
}

@article{wu2020projected,
  title={Projected tests for high-dimensional covariance matrices},
  author={Wu, Tung-Lung and Li, Ping},
  journal={Journal of Statistical Planning and Inference},
  volume={207}, pages={73--85}, year={2020},
  doi={10.1016/j.jspi.2019.11.003}
}

@article{ding2025two,
  title={Two Sample Test for Covariance Matrices in Ultra-High Dimension},
  author={Ding, Xiucai and Hu, Yichen and Wang, Zhenggang},
  journal={Journal of the American Statistical Association},
  volume={120}, number={552}, pages={2210--2221}, year={2025},
  doi={10.1080/01621459.2024.2423971}
}

@article{veitch2026rank,
  title={Rank-adaptive covariance testing with applications to genomics and neuroimaging},
  author={Veitch, David and He, Yinqiu and Park, Jun Young},
  journal={Biometrics},
  volume={82}, number={1}, pages={ujag052}, year={2026},
  doi={10.1093/biomtc/ujag052}
}

@article{parker_pam50_2009,
  title={Supervised risk predictor of breast cancer based on intrinsic subtypes},
  author={Parker, Joel S. and Mullins, Michael and Cheang, Maggie C. U. and others},
  journal={Journal of Clinical Oncology}, volume={27}, number={8}, pages={1160--1167},
  year={2009}, doi={10.1200/JCO.2008.18.1370}
}

@article{oshi_e2f_2020,
  title={The {E2F} Pathway Score as a Predictive Biomarker of Response to Neoadjuvant Therapy in {ER+/HER2-} Breast Cancer},
  author={Oshi, Masanori and Takahashi, Hideo and Tokumaru, Yoshihisa and Yan, Li and Rashid, Omar M. and Nagahashi, Masayuki and Matsuyama, Ryusei and Endo, Itaru and Takabe, Kazuaki},
  journal={Cells}, volume={9}, number={7}, pages={1643},
  year={2020}, doi={10.3390/cells9071643}
}

\end{document}